\documentclass{llncs}
\usepackage{amsmath,amsfonts,amssymb}
\usepackage{array,xspace,latexsym,stmaryrd,scalefnt,qtree}

\newcommand*\intp{p}
\newcommand*\pate{\eta}
\newcommand*\alignStepNb{2\xspace}

\newcommand*\VersionLongue{
  \newcommand\CourteLongue[2]{##2}
  \newcommand\Courte[1]{}
  \newcommand\Longue[1]{##1}
  \pagestyle{plain}
}

\VersionLongue

\newcommand*\seq{s}
\newcommand*\s{{}s\xspace}
\newcommand*\wfchain{well-founded chain\xspace}
\newcommand*\chain{looping chain\xspace}
\newcommand*\nbLits{\numn_l}
\newcommand*\concDom{E}
\newcommand*\nbSubIts{\numn_{it}}
\newcommand*\measIt{\mu}
\newcommand*\measLit{\nu}

\newcommand*\droite{\ \hfil}
\newcommand*\gauche{\hfill\ }
\newcommand*\Mathem{\underline}
\newcommand*\labeled[1]{$^{\ \left(#1\right)}$}
\newcommand*\relpos[2]{#1\!\setminus#2}
\newcommand*\litIn{\sqsubset}
\newcommand*\explicitref{$(\star)$\xspace}
\newcommand*\prem{P}

\DeclareMathOperator*\purified{purified}
\let\bigAnd=\bigwedge
\let\bigOr=\bigvee
\newcommand*\cycle[1]{\circlearrowleft\!\text{\scriptsize$(#1)$}}
\newcommand*\closed{\times}
\newcommand*\same{.}
\newcommand*\intof{}
\newcommand*\intn{\intof n}
\newcommand*\intm{\intof m}
\newcommand*\intk{\intof k}
\newcommand*\intq{\intof q}
\newcommand*\intj{\intof j}
\newcommand*\intl{\intof l}
\newcommand*\inti{\intof i}
\newcommand*\instantiateRuleCondition{\context(\sch_1)\Implies\cstr\substitution{\IRupbnd/\var}}
\newcommand*\litToSch[1]{\bigwedge_{#1}}
\newcommand*\finshift{$\stdeqshift$-finite\xspace}
\newcommand*\atom{A}

\newcommand*\Schlit{literal\xspace}
\newcommand*\Schatom{atom\xspace}
\newcommand*\IRupbnd{\expr}
\newcommand*\paramof[1]{\mathnormal{\mathsf #1}}
\newcommand*\paramm{\paramof m}
\newcommand*\param{\paramof n}
\newcommand*\shiftable{s}
\newcommand*\setofshiftables{\mathcal S}
\newcommand*\Shiftable{shiftable\xspace}
\newcommand*\childmodel[3]{#1^{#2}_{#3}}
\newcommand*\interval{I}
\newcommand*\numl{l}
\newcommand*\numi{i}
\newcommand*\numk{k}
\newcommand*\numj{j}
\newcommand*\numn{n}
\newcommand*\neutrOp{\varepsilon}
\newcommand*\conc{C}
\newcommand*\looping{\emph{Looping}\xspace}
\newcommand*\interpprop[1]{#1_p}
\newcommand*\interpenv[1]{\anenv_{#1}}
\newcommand*\StCstr{\And}
\newcommand*\LC{\mathcal{LC}}
\newcommand*\emptinesslhm{\schemaOpbis {\var'}{\cstr'}{\pat'}}
\newcommand*\cstrsplitlhm{\schemaOp\var\cstr\pat}%
\newcommand*\envapp[2]{|#1|_{#2}}
\newcommand*\neverbelong{\not\maybelong}
\newcommand*\maybelong{\litIn_\Diamond}
\newcommand*\exprf{f}
\newcommand*\expr{e}
\newcommand*\alwaysoccur{\litIn_\Box}
\newcommand*\Patt{\mathfrak P}
\newcommand*\lang{\mathcal L}
\newcommand*\sig{\Sigma}
\newcommand*\tuple[1]{\langle#1\rangle}

\newcommand*\aformula{\phi}
\newcommand*\xxdef[1]{\stackrel{\text{\tiny def}}{#1}}
\newcommand*\eqdef{\xxdef=}
\newcommand*\stdOr{\bigvee_{\var=1}^\param}
\newcommand*\stdAnd{\bigwedge_{\var=1}^\param}
\newcommand*\varof{\paramof}
\newcommand*\var{\varof i}
\newcommand*\varj{\varof j}

\newcommand*\setofexprs{\asetof\expr}
\newcommand*\cstrirred{constraint-irreducible\xspace}
\newcommand*\cstrirredy{constraint-irreducibility\xspace}
\newcommand*\representative{maximal companion\xspace}

\newcommand*\refinement{\triangleright}
\newcommand*\asetofprops{\asetof\prop}
\newcommand*\asetof[1]{\mathcal{\MakeUppercase{#1}}}
\newcommand*\Z{\mathbb Z}
\newcommand*\N{\bbbn}
\newcommand*\childparent{\prec}
\newcommand*\sep{,\ }%

\newcommand*\TSchDP{{\sc t}-\SchDP}
\newcommand*\asetofpos{\mathcal P}
\newcommand*\asetofconstraints{\mathcal C}
\newcommand*\pat{\pi}
\newcommand*\cstr{C}
\newcommand*\pos{p}
\newcommand*\aposq{q}
\newcommand*\asubstitution{\sigma}

\newcommand*\schcstr[1]{\cstr_{#1}}
\newcommand*\schpat[1]{\mathrm\Pi_{#1}}
\newcommand*\prop{P}
\newcommand*\propq{Q}
\newcommand*\abranch{b}

\newcommand*\trace{\leadsto}
\newcommand*\emptyPos{\epsilon}
\newcommand*\posOrder{\leq}
\newcommand*\Xor{\oplus}
\newcommand*\Equiv{\Leftrightarrow}
\newcommand*\literalset{\mathcal L}
\DeclareMathOperator*\context{Context}

\newcommand*\DPLL{{\sc dpll}\xspace}
\newcommand*\SchDP{{\sc \DPLL{}$^\star$}\xspace}
\newcommand*\SchCal{{\sc stab}\xspace}

\newcommand*\substitution[1]{[#1]}
\newcommand*\emptiness{\emph{Emptiness}\xspace}
\newcommand*\instantiaterule{\emph{Unfolding}\xspace}
\newcommand*\algebraic{\emph{Algebraic simplification}\xspace}
\newcommand*\propsimpl{\emph{Expansion}\xspace}
\newcommand*\propsplit{\emph{Propositional splitting}\xspace}
\newcommand*\constraintsplit{\emph{Constraint splitting}\xspace}

\newcommand*\intervalise{\emph{Interval splitting}\xspace}

\newcommand*\integervars{\mathcal{IV}}
\newcommand*\linear{linear expression\xspace}
\newcommand*\linears{linear expressions\xspace}
\newcommand*\constraint{linear constraint\xspace}
\newcommand*\constraints{linear constraints\xspace}
\newcommand*\schemaOr[3]{\ensuremath{\bigvee_{{#1}|{#2}} #3}}
\newcommand*\schemaAnd[3]{\ensuremath{\bigwedge_{{#1}|{#2}} #3}}
\newcommand*\schemaOp[3]{\ensuremath{\bigOp_{{#1}|{#2}} #3}}
\newcommand*\schemaOpbis[3]{\ensuremath{\bigOpbis_{{#1}|{#2}} #3}}
\newcommand*\schematicOr[4]{\ensuremath{\bigvee_{{#1}={#2}}^{#3} #4}}
\newcommand*\schematicAnd[4]{\ensuremath{\bigwedge_{{#1}={#2}}^{#3} #4}}
\newcommand*\schematicOp[4]{\ensuremath{\bigOp_{{#1}={#2}}^{#3} #4}}
\newcommand*\schematicOpbis[4]{\ensuremath{\bigOpbis_{{#1}={#2}}^{#3} #4}}
\DeclareMathOperator*\bigOp{\Delta}
\DeclareMathOperator*\bigOpbis{\nabla}

\DeclareMathOperator\LinExpr{\mathcal{LE}}
\DeclareMathOperator\Op{\vartriangle}
\DeclareMathOperator\Opbis{\triangledown}
\let\And=\wedge
\let\Or=\vee
\let\Implies=\Rightarrow
\newcommand*\Emph[1]{\textbf{#1}}

\allowdisplaybreaks[1]

\newcommand*\arithmetic{translated\xspace}
\newcommand*\Arithmetic{Translated\xspace}
\newcommand*\mxxbase[1]{\textstyle#1_{base}}
\newcommand*\minbase{\mxxbase{\min}}
\newcommand*\maxbase{\mxxbase{\max}}
\newcommand*\mxxit[1]{\textstyle#1_{ind}}
\newcommand*\minit{\mxxit{\min}}
\newcommand*\maxit{\mxxit{\max}}
\newcommand*{\isdef}{\stackrel{\mbox{\tiny def}}{=}}
\newcommand*{\tab}{{\cal T}} % a generic notation for a tableau
\newcommand*{\lit}{L}
\newcommand*{\node}{\alpha} % a notation for a node in the tableau
\newcommand*{\nodeb}{\beta} % a notation for a node in the tableau
\newcommand*{\nodesch}[2]{\sch_{#2}(#1)}
\newcommand*{\nodeschfull}[2]{\sch\literalset_{#2}(#1)}
\newcommand*{\nodepos}[2]{\asetofpos_{#2}(#1)}

\newcommand*{\nodeinterpsch}[2]{\litToSch{\literalset_{#2}(#1)}}
\newcommand*{\nodeinterp}[2]{\literalset_{#2}(#1)}
\newcommand*\aniteration{\schemaOp\var\cstr\pat}
\newcommand*\eqshift{\rightrightarrows}
\newcommand*\stdeqshift{\eqshift^\param}
\newcommand*\stdsubset[1]{\finitesubset{#1}\stdeqshift}
\newcommand*\deviation{\delta}

\newcommand*{\sch}{S}
\newcommand*{\setofschemata}{{\cal S}}

\newcommand*{\pour}{\text{ for }}
\newcommand*\finitesubset[2]{{#1}/{#2}}
\newcommand*\repr[1]{[#1]}
\newcommand*\alignedAnd{\times}
\newcommand*\bowl[1]{\mathfrak B_{#1}}

\newcommand*{\interp}{{\cal I}}
\newcommand*{\interpj}{{\cal J}}

\newcommand*\truevalue{$true$\xspace}
\newcommand*\falsevalue{$false$\xspace}
\newcommand*{\trueformula}{\top}
\newcommand*{\trueconstraint}{\top}
\newcommand*{\falseconstraint}{\bot}
\newcommand*{\falseformula}{\bot}
\newcommand*{\anenv}{\rho}
\newcommand*\maxshift{\max}

\newcommand*{\measureAlone}[2]{m^{#1}_{#2}}
\newcommand*{\measure}[3]{\measureAlone{#1}{#2}(#3)}
\newcommand*{\measurebranch}[3]{m_{#1}(#2,#3)}

\newcommand*{\elementary}{regularly nested\xspace}
\newcommand*{\Elementary}{Regularly Nested\xspace}
\newcommand*{\strat}{\mathfrak S}

\begin{document}
  \title{A Decidable Class of Nested Iterated Schemata\Longue{\\(extended version)}}
  \author{Vincent Aravantinos \and Ricardo Caferra \and Nicolas Peltier}
  \institute{Grenoble University (LIG/CNRS)}

  \maketitle
  \begin{abstract}
    Many problems can be specified by patterns of propositional formulae depending on a parameter,
    e.g. the specification of a circuit usually depends on the number of bits of its input.
    We define a logic whose formulae, called \emph{iterated schemata}, allow to express such patterns.
    Schemata extend propositional logic with
    indexed propositions, e.g. $\prop_\var$, $\prop_{\var+1}$, $\prop_1$ or $\prop_{\param}$,
    and with generalized connectives, e.g. $\stdAnd$ or $\stdOr$
    \Longue{(called \emph{iterations}) }%
    where $\param$ is an (unbound) integer variable called a \emph{parameter}.
    The expressive power of iterated schemata is strictly greater than propositional logic:
    it is even out of the scope of first-order logic.
    We define a proof procedure, called \SchDP,
    that can prove that a schema is satisfiable for at least one value of its parameter,
    in the spirit of the \DPLL procedure \cite{dpll}.
    \CourteLongue{But}{However the converse problem, i.e.}
    proving that a schema is unsatisfiable \emph{for every value of the parameter}, is undecidable \cite{tab09}
    so \SchDP does not terminate in general.
    Still, \Longue{we prove that }\SchDP terminates for schemata of a syntactic subclass called \emph{\elementary{}}.
    \Longue{%
      This is the first non trivial class for which \SchDP is proved to terminate.
      Furthermore the class of \elementary schemata is the first decidable class to allow nesting of iterations,
      i.e. to allow schemata of the form $\schematicAnd\var1\param(\schematicAnd\varj1\param\dots)$.
    }%
  \end{abstract}

  \section{Introduction}
    The specification of problems in propositional logic often leads to propositional formulae that depend on a parameter:
    the $n$-queens problem depends on $n$,
    the pigeonhole problem depends on the number of considered pigeons,
    a circuit may depend on the number of bits of its input,
    etc.
    Consider for instance a specification of a carry propagate adder circuit
    i.e. a circuit that takes as input two $\param$-bit vectors and computes their sum:
    \label{ex:adder}
    \CourteLongue
    {$Adder \eqdef \schematicAnd\var1\param Sum_\var\And\schematicAnd\var1\param Carry_\var\And \neg C_1$}
    {\[Adder \eqdef \schematicAnd\var1\param Sum_\var\And\schematicAnd\var1\param Carry_\var\And \neg C_1\]}
    \CourteLongue{
      where
      $\param$ is the number of bits of the input,
      $Sum_\var\eqdef S_\var \Equiv (A_\var\Xor B_\var)\Xor C_\var$,
      $Carry_\var \eqdef C_{\var+1}\Equiv (A_\var \And B_\var) \Or (B_\var\And C_\var) \Or (A_\var\And C_\var)$,
      $\Xor$ \CourteLongue{is}{denotes} the exclusive OR,
      $A_1,\dotsc,A_\param$ (resp. $B_1,\dotsc,B_\param$)
      are the bits of the first (resp. second) operand of the circuit,
      $S_1,\dotsc,S_\param$ is the output (the \textbf Sum),
      and $C_1,\dotsc,C_\param$ are the intermediate \textbf Carries.
    }
    {
      where:
      \[\begin{aligned}
        Sum_\var   &\eqdef S_\var \Equiv (A_\var\Xor B_\var)\Xor C_\var\\
        Carry_\var &\eqdef C_{\var+1}\Equiv (A_\var \And B_\var) \Or (B_\var\And C_\var) \Or (A_\var\And C_\var)\\
        \Xor &\text{ denotes the exclusive OR}\\
        A_1,\dotsc,A_\param &\text{ denotes the first operand of the circuit}\\
        B_1,\dotsc,B_\param &\text{ denotes the second operand of the circuit}\\
        S_1,\dotsc,S_\param &\text{ denotes the output (the \textbf Sum) of the circuit}\\
        C_1,\dotsc,C_\param &\text{ denotes the intermediate \textbf Carries of the circuit}\\
      \end{aligned}\]
    }%

    Presently, automated reasoning on such specifications requires that we give a concrete value to the parameter $\param$.
    Besides the obvious loss of generality, this instantiation hides the \emph{structure} of the initial problem which can be
    however a useful information when reasoning about such specifications:
    the structure of the proof can in many cases be \emph{guided} by the \CourteLongue{one}{structure} of the original specification.
    %Indeed the (satisfiability or validity) proof itself may have a structure, which can often be \emph{guided} by the structure of the specification.
    This gave us the idea to consider parameterized formulae at the object level and to design a logic to reason about them.

    Notice that schemata not only arise naturally from practical problems,
    but also have a deep conceptual interpretation, putting bridges between logic and computation.
    %The usual way of abstracting from propositional logic is to switch to first-order logic.
    As well as first\Longue{ or higher}-order logic abstracts from propositional logic via \emph{quantification},
    schemata allow to abstract via \emph{computation}\Courte{, in a complementary way}.
    Indeed, a schema can be considered as a very specific algorithm taking as input a value for the 
    parameter and generating a propositional formula depending on this value.
    So a schema can be seen as an algorithm whose codomain is the set of propositional formulae
    (its domain is the set of integers\Longue{ in this presentation, but one can imagine any type of parameter}).
    \Longue{Thus schemata can be seen as a different -- and complementary -- way to abstract from propositional logic.}

    If we want to prove, e.g. that the implementation of a parameterized specification is correct,
    we need to prove that the corresponding schema is valid \emph{for every value of the parameter}.
    As usual we actually deal with unsatisfiability:
    we say that a schema is \emph{unsatisfiable} iff 
    \CourteLongue{it is (propositionally) unsatisfiable for every value of its parameter}
    {every propositional formula obtained by giving a value to the parameter is unsatisfiable}.
    In \cite{tab09} we introduced a first proof procedure for propositional schemata, called \SchCal.
    Notice that there is an easy way to systematically look for a counter-example%
    \Longue{ (i.e. find a value of the parameter for which the schema is satisfiable)}:
    we can just enumerate all the values and check the satisfiability of the corresponding formula with a SAT solver.
    However this naive procedure does not terminate when the schema is unsatisfiable.
    On the other hand, \SchCal not only terminates (and much more efficiently) when the schema is satisfiable,
    but it can also terminate when the schema is unsatisfiable.
    However it still \emph{does not terminate in general},
    as we proved that the (un)satisfiability problem is undecidable for schemata \cite{tab09}.
    \Longue{As a consequence there cannot exist a complete calculus for schemata
    (the set of unsatisfiable schemata is not recursively enumerable\Longue).}
    Still, we proved that \SchCal terminates for a particular class of schemata,
    called \emph{regular}, which is thus decidable%
    \Longue{ (this class contains the carry propagate adder described previously)}.

    An important restriction of the class of regular schemata is that it cannot contain nested iterations,
    e.g. $\schematicOr\var1\param{\schematicOr\varj1\param{\prop_\var}\Implies\propq_\varj}$.
    Nested iterations occur frequently in the specification of practical problems.
    We take the example of a binary multiplier which computes the product of 
    two bit vectors $A=(A_1,\dots,A_\param)$ and $B=(B_1,\dots,B_\param)$ using the following decomposition:
    \CourteLongue{$A.B=A.\sum_{\var=1}^\param B_\var.2^{\var-1}=\sum_{\var=1}^\param A.B_\var.2^{\var-1}$.}
    {\[A.B=A.\sum_{\var=1}^\param B_\var.2^{\var-1}=\sum_{\var=1}^\param A.B_\var.2^{\var-1}\]}
    The circuit is mainly an iterated sum:%
    \CourteLongue{
    $\text{``}S^1=0\text{''}
    \And
    \schematicAnd\var1\param (B_\var\Implies Add(S^\var,\text{``}A.2^{\var-1}\text{''},S^{\var+1}))\And(\neg B_\var\Implies (S^{\var+1}\Equiv S^\var))$
    }
    {
    \[\text{``}S^1=0\text{''}
    \And
    \schematicAnd\var1\param (B_\var\Implies Add(S^\var,A.2^{\var-1},S^{\var+1}))\And(\neg B_\var\Implies (S^{\var+1}\Equiv S^\var))\]
    }%
    where $S^\var$ denotes the $\var^{th}$ partial sum 
    (hence $S^\param$ denotes the final result)
    and $Add(x,y,z)$ denotes any schema specifying a circuit which computes the sum $z$ of $x$ and $y$
    (for instance the previous $Adder$ schema).
    We express %
    \Longue{``$S^1=0$'' by $\schematicAnd\var1\param\neg S^1_\var$, and }%
    ``$A.2^{\var-1}$'' by the bit vector $Sh^\var=(Sh^\var_1,\dots,Sh^\var_{2\param})$
    ($Sh$ for \emph{Sh}ift):%
    \CourteLongue
    {
      $(\schematicAnd\varj1\param{Sh^1_\varj\Equiv A_\varj})
      \And
      (\schematicAnd\varj\param{2\param}{\neg Sh^1_\varj})
      \And
      (\schematicAnd\var1\param{\neg Sh^\var_1\And\schematicAnd\varj1{2\param}{(Sh^\var_{\varj+1}\Equiv Sh^\var_\varj)}})$
    }
    {\[
    \left(\schematicAnd\varj1\param{Sh^1_\varj\Equiv A_\varj}\right)
    \And
    \left(\schematicAnd\varj\param{2\param}{\neg Sh^1_\varj}\right)
    \And
    \left(\schematicAnd\var1\param{\neg Sh^\var_1\And\schematicAnd\varj1{2\param}{(Sh^\var_{\varj+1}\Equiv Sh^\var_\varj)}}\right)\]
    }%
    \Courte{and ``$S^1=0$'' by $\schematicAnd\var1\param\neg S^1_\var$.}
    This schema obviously contains nested iterations%
    \footnote{However it does not belong to the decidable class presented in Section \ref{sec:complete}.}.

    \SchCal does not terminate in general on such specifications. 
    We introduce in this paper a new proof procedure, called \SchDP,
    which is an extension of the \DPLL procedure \cite{dpll}. 
    Extending \DPLL to  schemata is a complex task,
    because the formulae depend on an \emph{unbounded} number of propositional variables 
    (e.g. $\schematicOr\var1\param\prop_\var$ ``contains'' $\prop_1,\ldots,\prop_\param$).
    Furthermore, propagating the value given to an atom is not straightforward as in \DPLL
    (in $\schematicOr\var1\param\prop_\var$ if the value of e.g. $\prop_2$ is fixed then we must 
    propagate the assignment to $\prop_\var$ but only in the case where $\var=2$).
    The main advantage of \SchDP over \SchCal is that it can operate on subformulae occurring at a deep position in the schema%
    \Longue{ (in contrast to \SchCal, which only handles root formulae, by applying decomposition rules)}. 
    This feature \CourteLongue{is}{turns out to be} essential for handling nested iterations.
    \Longue{We prove that }\SchDP is sound, complete for satisfiability detection and terminates on a class of schemata,
    called \emph{\elementary{}}%
    \CourteLongue.{, which is obtained from regular schemata by removing the restriction on nested iterations.}

    \Courte{
    A first version of the \SchDP calculus was presented in \cite{sd09}.
    Together with a new termination result we give a thoroughly revised and simplified
    presentation of the proof procedure.
    Due to space restrictions, proofs are omitted or simply sketched
    (detailed proofs can be found in \cite{rapport09}).
    }%

    The paper is organized as follows.
    Section \ref{sec:schemata} defines the syntax and semantics of iterated schemata.
    Section \ref{sec:proof_procedure} presents the \SchDP proof procedure.
    Section \ref{\CourteLongue{sec:looping_detection}{sec:termination}}
    deals with the detection of cycles in proofs, which is the main tool allowing termination.
    Section \ref{sec:complete} presents the class of \emph{\elementary schemata}, for which \Longue{we show that }\SchDP terminates.
    \Longue{Termination is also proven for some simple derivatives of this class. }%
    Section \ref{sec:conclusion} concludes the paper and \CourteLongue{overviews}{briefly presents} related works.

  \section{Schemata of Propositional Formulae}
    \label{sec:schemata}
    \Longue{Consider the usual signature $\sig\eqdef\{0,s,+,-\}$
    and a countable set of \emph{integer variables} denoted by $\integervars$.}
    Terms on \CourteLongue{the signature $\{0,s,+,-\}$}{$\sig$}
    and \Courte{on a countable set of integer variables }$\integervars$
    are called \emph{linear expressions}, whose set is written $\LinExpr$.
    As usual we simply write $\intn$ for $s^\intn(0)$ ($\intn>0$)
    and $\intn.\expr$ for $\expr+\dotsb+\expr$ ($\intn$ times).
    Linear expressions are considered
    modulo the usual properties of the arithmetic symbols
    (e.g. $s(0)+s(s(0))-0$ is assumed to be the same as $s(s(s(0)))$ and written $3$).
    \Longue{Consider the structure $\lang\eqdef\tuple{\sig;=,<,>}$ of linear arithmetic
    (i.e. same as Presburger arithmetic except that negative integers are also considered). }%
    The set of first-order formulae \CourteLongue{built on $\LinExpr$ and $=,<,>$}{of $\lang$} is called the set of \emph{linear constraints} 
    (or in short \emph{constraints}), written $\LC$.
    \CourteLongue{I}{As usual, i}f $\cstr_1,\cstr_2\in\LC$, we write $\cstr_1\models\cstr_2$
    iff $\cstr_2$ is a logical consequence of $\cstr_1$.
    This relation is well known to be decidable 
    \Longue{using decision procedures for arithmetic without multiplication }see e.g. \cite{Coo72}.
    It is also well known that linear arithmetic admits quantifier elimination.
    %Usual algorithms can be used for this \cite{Coo72}.
    \CourteLongue{C}{From now on, c}losed terms of $\sig$ (i.e. integers) are denoted by $\intn,\intm,\inti,\intj,\intk,\intl$,
    \linears by $\expr,\exprf$, 
    constraints by $\cstr,\cstr_1,\cstr_2,\dots$
    and integer variables by $\param,\var,\varj$
    \CourteLongue
      {to make clear the distinction between integer variables and expressions of the meta-language}
      {(we use this particular typesetting to clearly make the distinction with variables of the meta-language)}.

    To make technical details simpler, and w.l.o.g., only schemata in negative normal form (n.n.f.) are considered.
    \CourteLongue{A}{We say that a} linear constraint \emph{encloses} a variable $\var$ iff
    there exist $\expr_1,\expr_2\in\LinExpr$ s.t. $\var$ does not occur in $\expr_1,\expr_2$
    and $\cstr\models\expr_1\leq\var\And\var\leq\expr_2$.
    \begin{definition}[Schemata]
      \label{def:schemata}
      \label{def:patterns}
      For every $\numk\in\N$, let $\asetofprops_k$ be a set of symbols.
      The set $\Patt$ of \emph{formula patterns} (or, for short, \emph{patterns}) is the smallest set s.t.
      \begin{itemize}
        \item $\trueformula,\falseformula\in\Patt$
        \item If $\numk\in\N$, $\prop\in\asetofprops_\numk$ and $\expr_1,\dotsc,\expr_\numk\in\LinExpr$
          then $\prop_{\expr_1,\dotsc,\expr_\numk}\in\Patt$ and $\neg\prop_{\expr_1,\dotsc,\expr_\numk}\in\Patt$.
        \item If $\pat_1,\pat_2\in\Patt$ 
          then $\pat_1\Or\pat_2\in\Patt$ and $\pat_1\And\pat_2\in\Patt$.
        \item If $\pat\in\Patt$, $\var\in\integervars$, $\cstr\in\LC$ and $\cstr$ encloses $\var$
          then $\schemaAnd\var\cstr\pat\in\Patt$ and $\schemaOr\var\cstr\pat\in\Patt$.
      \end{itemize}
      A \emph{schema} $\sch$ is a pair (written as a conjunction) $\pat\And\cstr$,
      where $\pat$ is a pattern and $\cstr$ is a constraint.
      $\cstr$ is called the \emph{constraint} of $\sch$, written $\schcstr\sch$.
      $\pat$ is called its \emph{pattern}, written $\schpat\sch$.
    \end{definition}
    The first three items \CourteLongue{differ}{define a language that differs} from propositional logic 
    only in \CourteLongue{the}{its} atoms which we call \emph{indexed propositions}
    ($\expr_1,\dotsc,\expr_\numk$ are called \emph{indices}).
    The real novel part is the last item.
    Patterns of the form $\schemaAnd\var\cstr\pat$ or $\schemaOr\var\cstr\pat$ are called \emph{iterations}.
    $\cstr$ is called the \emph{domain} of the iteration.
    In \cite{tab09} only domains of the form $\expr_1\leq\var\And\var\leq\expr_2$ were handled,
    but as we shall see in Section \ref{sec:proof_procedure},
    more general classes of constraints are required to define the \SchDP procedure.
    If $\cstr$ is unsatisfiable then the iteration is \emph{empty}.
    Any occurrence of $\var$ in $\pat$ is \emph{bound} by the iteration.
    A variable occurrence which is not bound is \emph{free}.
    A variable which has free occurrences in a pattern is a \emph{parameter} of the pattern.
    A pattern which is just an indexed proposition $\prop_{\expr_1,\dotsc,\expr_\numk}$ is called an \emph{\Schatom{}}.
    An \Schatom or the negation of an \Schatom is called a \emph{\Schlit{}}.

    In \cite{tab09} and \cite{sd09} a schema was just a pattern,
    however constraints appear so often that it is more convenient to integrate them to the definition of schema.
    Informally, a pattern gives a ``skeleton'' with ``holes''
    and the constraint specifies how the holes can be filled
    (this choice fits the abstract definition of schema in \cite{COR06}).
    \Longue{
    This new definition can be emulated with the definition of \cite{sd09} 
    (as one can see from the upcoming semantics $\bigOr_\cstr\trueformula$ is equivalent to $\cstr$). }%
    In the following we assume w.l.o.g. 
    that $\schcstr\sch$ entails $\param_1\geq0\And\dotsb\And\param_\numk\geq0$ 
    where $\param_1,\dots,\param_\numk$ are the parameters of $\schpat\sch$.
    \begin{example}
      \label{ex:pattern}
      \CourteLongue
      {$\sch\eqdef
      \prop_1\And\bigAnd_{1\leq\var\And\var\leq\param}(\propq_\var\And\bigOr_{1\leq\varj\leq\param+1\And\var\neq\varj}
      \neg\prop_\param\Or\prop_{\varj+1})\StCstr\param\geq1$
      is a schema.
      }
      {$\sch$ is a schema:
      \[\sch\eqdef
      \prop_1\And\bigAnd_{1\leq\var\And\var\leq\param}(\propq_\var\And\bigOr_{1\leq\varj\leq\param+1\And\var\neq\varj}
      \neg\prop_\param\Or\prop_{\varj+1})\StCstr\param\geq1\]}
      $\prop_1$, $\propq_\var$, $\prop_\param$ and $\prop_{\varj+1}$ are indexed propositions.
      \CourteLongue
      {$\bigOr_{1\leq\varj\leq\param+1\And\var\neq\varj}\neg\prop_\var\Or\prop_{\var+1}$
      and
      $\bigAnd_{1\leq\var\And\var\leq\param} (\propq_\varj\And\bigOr_{1\leq\varj\leq\param+1\And\var\neq\varj}
      \neg\prop_\var\Or\prop_{\var+1})$
      are the only iterations, of domains $1\leq\varj\leq\param+1\And\var\neq\varj$ and $1\leq\var\leq\param$.}
      {The only iterations of $\sch$ are:
      \[\bigOr_{1\leq\varj\leq\param+1\And\var\neq\varj}\neg\prop_\var\Or\prop_{\var+1}\]
      and
      \[\bigAnd_{1\leq\var\And\var\leq\param} (\propq_\varj\And\bigOr_{1\leq\varj\leq\param+1\And\var\neq\varj}
      \neg\prop_\var\Or\prop_{\var+1})\]
      Their respective domains are
      $1\leq\varj\leq\param+1\And\var\neq\varj$
      and
      $1\leq\var\leq\param$.}
      $\param$ is the only parameter of $\sch$.
      \CourteLongue
      {Finally $\schpat\sch$ is $\prop_1\And\bigAnd_{1\leq\var\And\var\leq\param}(\propq_\var\And\bigOr_{1\leq\varj\leq\param+1\And\var\neq\varj}
      \neg\prop_\param\Or\prop_{\varj+1})$
      and $\schcstr\sch$ is $\param\geq1$.}
      {Finally:\[\schpat\sch=\prop_1\And\bigAnd_{1\leq\var\And\var\leq\param}(\propq_\var\And\bigOr_{1\leq\varj\leq\param+1\And\var\neq\varj}
      \neg\prop_\param\Or\prop_{\varj+1})\]
      and $\schcstr\sch=\param\geq1$.}
    \end{example}
    Schemata are denoted by $\sch$, $\sch_1$, $\sch_2\dots$,
    parameters by $\param,\param_1,\param_2\dots$,
    bound variables by $\var,\varj$.
    $\schemaOp\var\cstr\sch$ and $\schemaOpbis\var\cstr\sch$ denote generic iterations
    (i.e. $\schemaOr\var\cstr\sch$ or $\schemaAnd\var\cstr\sch$),
    $\Op$ and $\Opbis$ denote generic binary connectives (\Longue{i.e. }$\Or$ or $\And$),
    \Longue{finally }$\schematicOp\var{\expr_1}{\expr_2}\sch$ denotes $\schemaOp\var{\expr_1\leq\var\And\var\leq\expr_2}\sch$.

    Let $\sch$ be a schema
    and $\schemaOp{\var_1}{\cstr_1}{\sch_1}$, \dots, $\schemaOp{\var_\numk}{\cstr_\numk}{\sch_\numk}$ 
    be all the iterations occurring in $\sch$.
    Then $\schcstr\sch\And\cstr_1\And\dots\And\cstr_\numk$ is called the \emph{constraint context}
    of $\sch$, written $\context(\sch)$.
    Notice that $\context(\sch)$ loses the information on the \emph{binding positions} of variables.
    This can be annoying if a variable name is bound by two different iterations or if it is both bound and free in the schema.
    So we assume that all schemata are such that this situation does not hold%
    \footnote{The proof system defined in Section \ref{sec:proof_procedure}
    preserves this property, except for the rule \emptiness which duplicates an iteration;
    but we may safely assume that the variables of one of the duplicated iterations 
    are renamed so that the\Longue{ desired} property is fulfilled.}.

    \emph{Substitutions} on integer variables map integer variables to linear arithmetic expressions.
    We write $\substitution{\expr_1/\var_1,\ldots,\expr_\numk/\var_\numk}$ 
    for the substitution mapping $\var_1,\ldots,\var_\numk$ to $\expr_1,\ldots,\expr_\numk$ respectively.
    The application of a substitution $\asubstitution$ to an arithmetic expression $\expr$,
    written $\expr\asubstitution$, is defined as usual.
    Substitution application is naturally extended to schemata
    (notice that bound variables are not replaced).
    A substitution is \emph{ground} iff it maps integer variables to integers (i.e. ground arithmetic expressions).
    An \emph{environment} $\anenv$ of a schema $\sch$ is a ground substitution
    mapping all parameters of $\sch$ and such that $\schcstr\sch\anenv$ is true.
    \begin{definition}[Propositional Realization]
      \label{def:realization}
      Let $\pat$ be a pattern and $\anenv$ a ground substitution.
      The propositional formula $\envapp\pat\anenv$ is defined as follows:
      \begin{itemize}
        \item $\envapp{\prop_{\expr_1,\dotsc,\expr_\numk}}\anenv\eqdef\prop_{\expr_1\anenv,\dotsc,\expr_\numk\anenv}$,
          $\envapp{\neg\prop_{\expr_1,\dotsc,\expr_\numk}}\anenv\eqdef\neg\prop_{\expr_1\anenv,\dotsc,\expr_\numk\anenv}$,
        \item $\envapp{\trueformula}\anenv\eqdef\trueformula$, $\envapp{\falseformula}\anenv\eqdef\falseformula$,
          $\envapp{\pat_1\And\pat_2}\anenv\eqdef\envapp{\pat_1}\anenv\And\envapp{\pat_2}\anenv$, 
          $\envapp{\pat_1\Or\pat_2}\anenv\eqdef\envapp{\pat_1}\anenv\Or\envapp{\pat_2}\anenv$
        \item $\envapp{\schemaOr\var\cstr\pat}\anenv\eqdef
          \CourteLongue
          {\bigOr\left\{
          \envapp{\pat\substitution{\inti/\var}}{\anenv\cup\substitution{\inti/\var}}
          \middle|
          \inti\in\Z\text{ s.t. }\cstr\substitution{\inti/\var}\anenv\text{ is valid}\right\}}
          {\displaystyle\bigOr_{\inti\in\Z\text{ s.t. }\cstr\substitution{\inti/\var}\anenv\text{ is valid}}
          \envapp{\pat\substitution{\inti/\var}}{\anenv\cup\substitution{\inti/\var}}}
          $
        \item $\envapp{\schemaAnd\var\cstr\pat}\anenv\eqdef
          \CourteLongue
          {\bigAnd\left\{\envapp{\pat\substitution{\inti/\var}}{\anenv\cup\substitution{\inti/\var}}\middle|
          \inti\in\Z\text{ s.t. }\cstr\substitution{\inti/\var}\anenv\text{ is valid}\right\}}
          {\displaystyle\bigAnd_{\inti\in\Z\text{ s.t. }\cstr\substitution{\inti/\var}\anenv\text{ is valid}}
          \envapp{\pat\substitution{\inti/\var}}{\anenv\cup\substitution{\inti/\var}}}
          $
      \end{itemize}
      When $\anenv$ is an environment of a schema $\sch$,
      we define $\envapp\sch\anenv$ as $\envapp{\schpat\sch}\anenv$.
      $\envapp\sch\anenv$ is called a \emph{propositional realization} of $\sch$.
    \end{definition}
    Notice that $\trueformula, \falseformula, \Or, \And, \neg$ on the right-hand members 
    of equations have their standard \emph{propositional} meanings.
    $\bigOr$ and $\bigAnd$ on the right-hand members are meta-operators 
    denoting respectively the \emph{propositional} formulae $\dotsb\Or\dotsb\Or\dotsb$ and $\dotsb\And\dotsb\And\dotsb$ 
    or $\falseformula$ and $\trueformula$ when the \CourteLongue{sets are empty}{conditions are not verified}.
    On the contrary all those symbols on the left-hand members are \emph{pattern} connectives.

    We now make precise the semantics outlined in the introduction.
    Propositional logic semantics are defined as usual.
    A \emph{propositional interpretation} \Longue{of a (propositional) formula $\aformula$ }is a function mapping every
    propositional variable \Longue{of $\aformula$ }to a truth value \truevalue or \falsevalue.
    \begin{definition}[Semantics]
      \label{def:semantics}
      Let $\sch$ be a schema.
      An \emph{interpretation $\interp$ of the schemata language} is the pair of
      an environment $\interpenv\interp$ of $\sch$
      and a propositional interpretation $\interpprop\interp$ of $\envapp\sch{\interpenv\interp}$.
      A schema $\sch$ is \emph{true} in $\interp$ iff 
      $\envapp\sch{\interpenv\interp}$ is true in $\interpprop\interp$,
      in which case $\interp$ is a \emph{model} of $\sch$.
      $\sch$ is \emph{satisfiable} iff it has a model.
    \end{definition}
    Notice that an empty iteration $\schemaOr\var\cstr\pat$ (resp. $\schemaAnd\var\cstr\pat$) is
    \CourteLongue{equivalent to $\falseformula$ (resp. $\trueformula$)}{always false (resp. true)}.%
    \begin{example}
      \label{ex:semantics}
      \CourteLongue%
      {Let $\sch\eqdef\prop_1\And\schematicAnd\var1\param(\prop_\var\Implies\prop_{\var+1})\And\neg\prop_{\param+1}\And\param\geq0$
      (as usual, $\sch_1\Implies\sch_2$ is a shorthand for $\neg\sch_1\Or\sch_2$).
      Then 
      $\envapp\sch{\param\mapsto0}=\prop_1\And\neg\prop_1$,
      $\envapp\sch{\param\mapsto1}=\prop_1\And(\prop_1\Implies\prop_2)\And\neg\prop_2$,
      $\envapp\sch{\param\mapsto2}=\prop_1\And(\prop_1\Implies\prop_2)\And(\prop_2\Implies\prop_3)\And\neg\prop_3$,
      etc.
      $\sch$ is clearly unsatisfiable.}%
      {Consider the following schema:
      \[\sch\eqdef\prop_1\And\schematicAnd\var1\param(\prop_\var\Implies\prop_{\var+1})\And\neg\prop_{\param+1}\And\param\geq0\]
      (as usual, $\sch_1\Implies\sch_2$ is a shorthand for $\neg\sch_1\Or\sch_2$).
      Then 
      \[\begin{aligned}
      \envapp\sch{\param\mapsto0}&=\prop_1\And\neg\prop_1\\
      \envapp\sch{\param\mapsto1}&=\prop_1\And(\prop_1\Implies\prop_2)\And\neg\prop_2\\
      \envapp\sch{\param\mapsto2}&=\prop_1\And(\prop_1\Implies\prop_2)\And(\prop_2\Implies\prop_3)\And\neg\prop_3\\
      \text{etc.}\\
      \end{aligned}\]
      $\sch$ is clearly unsatisfiable.
      Notice that $\param\mapsto-\numk$ is \emph{not an environment} of $\sch$ for any $\numk>0$.}%
    \end{example}
    The set of satisfiable schemata is recursively enumerable but not recursive \cite{tab09}.
    Hence there cannot be a refutationally complete proof procedure for schemata.
    \Longue{Notice that the semantics are different from the ones
    in \cite{tab09} and \cite{sd09} but easily seen to be equivalent.}%

    The next definitions will be useful in the definition of \SchDP.
    Let $\aformula$ be a propositional formula and $\lit$ a (propositional) literal.
    We say that $\lit$ occurs \emph{positively} in $\aformula$, written $\lit\litIn\aformula$,
    iff there is an occurrence of $\lit$ in $\aformula$ which is not in the scope of a negation.
    As we consider formulae in n.n.f., a negative literal occurs positively in $\aformula$ iff it simply occurs in $\aformula$.
    \begin{definition}
      \label{def:alwaysoccur}
      Let $\sch$ be a schema and $\lit$ a \Schlit
      s.t. the parameters of $\lit$ are parameters of $\sch$.
      \Longue\par
      We write $\lit\alwaysoccur\sch$ iff for every environment $\anenv$ of $\sch$,
      $\envapp\lit\anenv\litIn\envapp\sch\anenv$.
      \Longue\par
      We write $\lit\maybelong\sch$ iff there is an environment $\anenv$ of $\sch$
      s.t. $\envapp\lit\anenv\litIn\envapp\sch\anenv$.
    \end{definition}
    \begin{example}
      \label{ex:alwaysoccur}
      Consider $\sch$ as in Example \ref{ex:semantics}.
      We have $\prop_1\alwaysoccur\sch$, $\prop_{\param+1}\alwaysoccur\sch$, $\prop_2\not\alwaysoccur\sch$.
      However $\prop_2\maybelong\sch$
      and $\prop_2\alwaysoccur(\sch\And\param\geq1)$.
      Finally $\prop_0\neverbelong\sch$
      and $\prop_{\param+2}\neverbelong\sch$.
      \Longue{Notice that $\neg\prop_1\maybelong\sch$ as 
      $\sch_1\Implies\sch_2$ is a shorthand for $\neg\sch_1\Or\sch_2$.}%
    \end{example}
    \newcommand*\occurFormula{\aformula}%
    Suppose $\lit$ has the form $\prop_{\expr_1,\dotsc,\expr_\numk}$ (resp. $\neg\prop_{\expr_1,\dotsc,\expr_\numk}$).
    For a literal $\lit'\litIn\sch$ of indices $\exprf_1,\dotsc,\exprf_\numk$,
    $\occurFormula_\lit(\lit')$ denotes the formula%
    \CourteLongue
    { $\exists\var_1\dots\var_\numn(\cstr_{\var_1}\And\dotsb\And\cstr_{\var_\numn}\And\expr_1=\exprf_1\And\dotsb\And\expr_\numk=\exprf_\numk)$}%
    {:\[\exists\var_1\dots\var_\numn(\cstr_{\var_1}\And\dotsb\And\cstr_{\var_\numn}\And\expr_1=\exprf_1\And\dotsb\And\expr_\numk=\exprf_\numk)\]}
    where $\var_1,\dots,\var_\numn$ are all the bound variables of $\sch$ occurring in $\exprf_1,\dots,\exprf_\numk$
    and $\cstr_{\var_1},\dotsc,\cstr_{\var_\numn}$ are the domains of the iterations binding $\var_1,\dots,\var_\numn$.
    \Longue{Then }$\occurFormula_\lit(\sch)$ denotes
    \CourteLongue%
    {$\bigvee\{\occurFormula_\lit(\prop_{\exprf_1,\dotsc,\exprf_\numk})\mid\prop_{\exprf_1,\dotsc,\exprf_\numk}\litIn\sch\}$
    (resp. $\bigvee\{\occurFormula_\lit(\neg\prop_{\exprf_1,\dotsc,\exprf_\numk})\mid\neg\prop_{\exprf_1,\dotsc,\exprf_\numk}\litIn\sch\}$).}
    {the following formula:
    \[\bigvee\{\occurFormula_\lit({\prop_{\exprf_1,\dotsc,\exprf_\numk}})\mid\prop_{\exprf_1,\dotsc,\exprf_\numk}\litIn\sch\}
    \quad(\text{resp. }
    \bigvee\{\occurFormula_\lit({\neg\prop_{\exprf_1,\dotsc,\exprf_\numk}})\mid\neg\prop_{\exprf_1,\dotsc,\exprf_\numk}\litIn\sch\})\]}
    \begin{proposition}
      \label{thm:always_occur_decidable}
      $\lit\alwaysoccur\sch$ 
      iff $\forall\param_1,\dots,\param_\numl(\schcstr\sch\Implies\occurFormula_\lit(\sch))$ is valid,
      where $\param_1\dots\param_\numl$ are all the parameters of $\sch$.
      $\lit\maybelong\sch$
      iff $\exists\param_1,\dots,\param_\numl(\schcstr\sch\And\occurFormula_\lit(\sch))$ is valid.
    \end{proposition}
    \Longue{
    \begin{proof}(Sketch)
      Let $\anenv$ be an environment of $\sch$.
      Assume that $\lit$ has the form $\prop_{\expr_1,\dotsc,\expr_\numk}$
      (the case $\neg\prop_{\expr_1,\dotsc,\expr_\numk}$ is similar).
      From Definition \ref{def:realization}, it is easily seen (by induction on the number of nested iterations) that
      $\envapp\lit\anenv\litIn\envapp\sch\anenv$
      iff there is a literal $\prop_{\exprf_1,\dotsc,\exprf_\numk}\litIn\sch$
      s.t. $\lit\anenv=\prop_{\exprf_1,\dotsc,\exprf_\numk}(\anenv\cup\substitution{\numi_1/\var_1,\dotsc,\numi_\numn/\var_\numn})$
      where $\var_1,\dots,\var_\numn$ are all the bound variables occurring in $\exprf_1,\dots,\exprf_\numk$
      and $\numi_1,\dotsc,\numi_\numn$ are such that
      s.t. $\cstr_1(\anenv\cup\substitution{\numi_1/\var_1,\dotsc,\numi_\numn/\var_\numn})$, \dots,
      $\cstr_\numn(\anenv\cup\substitution{\numi_1/\var_1,\dotsc,\numi_\numn/\var_\numn})$ are valid.
      It is then obvious that
      $\envapp\lit\anenv\litIn\envapp\sch\anenv$
      iff
      $\occurFormula\anenv$ is valid.
      The result follows easily.
      \qed
    \end{proof}}
    \begin{example}
      Consider $\sch$ as \Longue{defined }in Example \ref{ex:semantics}.
      For any \Longue{expression }$\expr$,
      $\prop_\expr\alwaysoccur\sch$\Longue{ (resp. $\prop_\expr\maybelong\sch$)}
      iff
      $\forall\param(\param\geq0)\Implies[\expr=1\Or
      \exists\var(1\leq\var\And\var\leq\param\And\expr=\var)
      \Or\exists\var(1\leq\var\And\var\leq\param\And\expr=\var+1)
      \Or\expr=\param+1]$
      \Longue{(resp. $\exists\param\exists\var(\param\geq0)\And(1\leq\var\And\var\leq\param)\And(\expr=1\Or\expr=\var\Or\expr=\var+1\Or\expr=\param+1)$)}
      is valid.
    \end{example}
    Then, by decidability of linear arithmetic, both $\alwaysoccur$ and $\maybelong$ are decidable.
    Besides, it is easy to compute the set $\literalset(\sch)\eqdef\{\lit\mid\lit\alwaysoccur\sch\}$
    for a \Longue{given }schema $\sch$%
    %\CourteLongue{ (notice that $\literalset(\sch)$ is finite as for every literal $\lit\in\literalset(\sch)$,
    %we must have $\lit\litIn\aformula$ for any propositional realization $\aformula$ of $\sch$).}%
    \CourteLongue.{: one just take any propositional realization $\aformula$ of $\sch$,
    and check for every literal $\lit\litIn\aformula$ if $\lit\alwaysoccur\sch$.
    If yes then it belongs to $\literalset(\sch)$ otherwise it does not.
    It is enough to do this with only one propositional realization 
    as for any $\lit\in\literalset(\sch)$, we must have $\lit\litIn\aformula$ for \emph{every} propositional realization $\aformula$.}%

  \section{A Proof Procedure: \SchDP}
  \label{sec:proof_procedure}
  We provide now a set of (sound) deduction rules
  (in the spirit of the Davis-Putnam-Logemann-Loveland procedure for propositional logic \cite{dpll})
  complete w.r.t. satisfiability
  (we know that it is not possible to get refutational completeness).
  Compared to other proof procedures \cite{tab09}
  \SchDP allows to rewrite subformulae occurring at deep positions inside a schema
  --- in particular occurring in the scope of iterated connectives:
  this is crucial to handle nested iterations.
  \Longue{\subsection{Extension Rules}}%
    \SchDP is a tableaux-like procedure: 
    rules are given to construct a tree whose root is the formula that one wants to refute.
    The formula is refuted iff all the branches are contradictory.

    As usual with tableaux related methods,
    the aim of branching is to browse the possible interpretations of the schema.
    As a schema interpretation assigns a truth value to each atom and a number to each parameter,
    there are two branching rules:
    one for atoms, called \propsplit{} 
    (this rule assigns a value to propositional variables, as the splitting rule in \DPLL),
    and one for parameters, called \constraintsplit{}.
    However \constraintsplit does not give a value to the parameters, 
    but rather restricts their values by refining the constraint of the schema (i.e. $\schcstr\sch$),
    e.g. the parameter can be either greater or lower than a given integer,
    leading to two branches in the tableaux.
    Naturally, in order to analyze a schema, one has to investigate the contents of iterations.
    So a relevant constraint to use for the branching is the one that states the emptiness of some iteration.
    In the branch where the iteration is empty, we can replace it by its neutral element 
    (i.e. $\trueformula$ for $\bigAnd$ and $\falseformula$ for $\bigOr$),
    which is done by \constraintsplit
    (this may also entails the emptiness of some other iterations,
    and thus their replacement by their neutral elements too,
    this is handled by \algebraic).
    Then in the branch where the iteration is not empty, we can unfold the iteration:
    this is done by the \instantiaterule rule.

    Iterations might occur in the scope of other iterations.
    Thus their domains might depend on variables bound by the outer iterations.
    \constraintsplit is of no help in this case,
    indeed it makes a branching only according to the values of the \emph{parameter}:
    bound variables are out of its scope.
    Hence we define the rule \emptiness that can make a ``deep'' branching,
    i.e. a branching not in the tree, but in the schema itself:
    it ``separates'' an iteration into two distinct ones,
    depending on the constraint stating the emptiness of the inner iteration,
    e.g. $\schematicOr\var1\param \schematicOr\varj3\var \prop_\var \And\param \geq 2$
    is replaced by $\schematicOr\var3\param \schematicOr\varj3\var \prop_\var \Or \schematicOr\var12 \falseformula \And\param \geq 2$.
    \Longue{The reader can notice that \constraintsplit and \emptiness are very similar.
    It could be possible to merge both into only one rule e.g. by considering infinite iterations
    but this would complicate all the formalism for only a little gain in the proof system.
    Furthermore, \emptiness differs ``conceptually'' from \constraintsplit in the sense that
    its role is not to browse interpretations but only to analyse a formula.}%

    \constraintsplit strongly affects the application of \propsplit.
    Indeed \propsplit only applies on atoms occurring in \emph{\CourteLongue{all}{every}} instance\Courte{s} of the schema
    (\Longue{which is }formalized by Definition \ref{def:alwaysoccur}),
    and we saw in Example \ref{ex:alwaysoccur} that this depends on the constraint\Longue{ of the schema}.
    Once an atom \CourteLongue{$\atom$ is}{has been} given the value $\truevalue$ (resp. $\falsevalue$)
    we can \CourteLongue{replace}{substitute} it \CourteLongue{by}{with} $\trueformula$ (resp. $\falseformula$).
    \CourteLongue{But}{However} this is not as simple as in the propositional case as \CourteLongue{$\atom$}{this atom} may occur 
    in a realization of the schema without occurring in the schema itself
    (e.g. $\prop_1$ in $\schematicAnd\var1\param\prop_\var$ \explicitref),
    so we cannot just \CourteLongue{replace it by $\trueformula$}{substitute $\trueformula$ to it}.
    The simplification is performed by the rule \propsimpl
    which wraps the indexed propositions that are more general than the considered atom ($\prop_\var$ in \explicitref)
    with an iteration whose domain is a disunification constraint stating that the proposition is distinct from the \Longue{considered }atom
    (this gives for \explicitref: $\schematicAnd\var1\param{\schemaAnd\varj{\var\neq1\And\varj=0}{\prop_\var}}$).
    The introduced iteration is very specific because the bound variable always equals $0$
    (\Longue{actually }this variable is not used 
    \Longue{and does not even occur in the wrapped proposition }%
    but we assign it $0$ to satisfy the condition
    in Definition \ref{def:patterns} that it has to be enclosed by the domain).
    Whereas usual iterations shall be considered as ``for loops'',
    this \CourteLongue{one is}{iteration shall be considered as} an ``if then else''.
    It \Longue{all }makes sense when \emptiness or \constraintsplit is applied:
    \emph{if} the condition holds (i.e. if the wrapped \Longue{indexed }proposition differs from the atom) 
    \emph{then} the contents of the iteration hold (i.e. we keep the indexed proposition\Longue{ as is})
    \emph{else} the iteration is empty (i.e. we replace it by its neutral element).
    In \explicitref, \emptiness applies:
    \CourteLongue
    {$\schemaAnd\var{1\leq\var\leq\param\And\exists\varj(\var\neq1\And\varj=0)}{\schemaAnd\varj{\var\neq1\And\varj=0}{\prop_\var}}
      \And
      \schemaAnd\var{1\leq\var\leq\param\And\forall\varj(\var=1\Or\varj\neq0)}\trueformula$ }
    {\[\schemaAnd\var{1\leq\var\leq\param\And\exists\varj(\var\neq1\And\varj=0)}{\schemaAnd\varj{\var\neq1\And\varj=0}{\prop_\var}}
    \And
    \schemaAnd\var{1\leq\var\leq\param\And\forall\varj(\var=1\Or\varj\neq0)}\trueformula\]}%
    (of course the domains can be simplified to allow reader-friendly presentation\Longue{:
    $\schemaAnd\var{2\leq\var\leq\param}{\schemaAnd\varj{\var\neq1\And\varj=0}{\prop_\var}}
      \And
      \schemaAnd\var\falseconstraint\trueformula$}%
    ).
    Then \algebraic gives:
    \CourteLongue
    {$\schemaAnd\var{1\leq\var\leq\param\And\exists\varj(\var\neq1\And\varj=0)}{\prop_\var}$,}
    {\[\schemaAnd\var{1\leq\var\leq\param\And\exists\varj(\var\neq1\And\varj=0)}{\prop_\var}\]}
    i.e. $\schematicAnd\var2\param{\prop_\var}$, as expected.
    \CourteLongue{T}{All t}his process may seem cumbersome,
    but it is actually a uniform and powerful way of propagating constraints about nested iterations along the schema.
    \Longue{The alternative would be to consider different expansion rules depending on the fact that 
    $\exprf_1,\ldots,\exprf_\numk$ occur in an iteration or not, which would be rather tedious.}%

    Finally we may know that an iteration is empty without knowing
    which value of the bound variable satisfies the domain constraint%
    \CourteLongue.{ e.g. if a constraint, that we know not to be empty, contains $\expr\leq\var\And\exprf\leq\var$
    then how can we know which rank of $\expr$ or $\exprf$ can indeed be reached?}
    \CourteLongue{Then}{In such cases,} the \intervalise rule adds some constraints on the involved expressions to ensure this knowledge.

    We now define \SchDP formally.
    \Courte{\vspace*{-0.1cm}}
    \begin{definition}[Tableau]
      A \emph{tableau} is a tree $\tab$
      s.t. each node $\node$ in $\tab$ is labeled with a pair
      $(\nodesch\node\tab,\nodeinterp\node\tab)$
      containing a schema and a finite set of literals.
    \end{definition}
    \Courte{\vspace*{-0.1cm}}
    If $\node$ is the root of the tree then $\nodeinterp\node\tab=\emptyset$ and
    $\nodesch\node\tab$ is called the \emph{root schema}.
    The transitive closure of the child-parent relation is written $\childparent$.
    For a set of literals $\literalset$, $\litToSch\literalset$ denotes the pattern $\bigAnd_{\lit\in\literalset}\lit$.

    As usual a tableau is generated from another tableau by applying extension rules
    written $\frac\prem\conc$ (resp. $\frac\prem{\conc_1|\conc_2}$)
    where $\prem$ is the premise and $\conc$ (resp. $\conc_1,\conc_2$) the conclusion(s).
    Let $\node$ be a leaf of a tree $\tab$,
    if the label of $\node$ matches the premise
    then we can \emph{extend} the tableau by adding to $\node$
    a child (resp. two children) labeled with $\conc\asubstitution$ (resp. $\conc_1\asubstitution$ and $\conc_2\asubstitution$),
    where $\asubstitution$ is the matching substitution.
    A leaf $\node$ is \emph{closed} iff
    $\schpat{\nodesch\node\tab}$ is equal to $\falseformula$ or $\schcstr{\nodesch\node\tab}$ is unsatisfiable.
    
    When used in a premise, $\sch[\pat]$
    means that the schema $\pat$ occurs in $\sch$;
    then in a conclusion, $\sch[\pat']$ denotes
    $\sch$ in which $\pat$ has been substituted with $\pat'$.
    \Courte{\vspace*{-0.1cm}}
    \begin{definition}[\SchDP rules]
      \label{def:rules}
      \let\temparraystretch=\arraystretch%
      \renewcommand{\arraystretch}{1.2}%
      The \emph{extension rules} are:
      \Courte{\vspace*{-0.1cm}}
      \begin{itemize}
        \item \propsplit. 
          \Courte{\vspace*{-0.5cm}}
          \[
           \begin{array}{c|c}
             \multicolumn2c{(\sch,\literalset)}\\\hline
             (\sch,\literalset\cup \prop_{\expr_1,\dotsc,\expr_\numk})\ \ &\ (\sch,\literalset\cup\neg\prop_{\expr_1,\dotsc,\expr_\numk})\\
           \end{array}
           \Courte{\vspace*{-0.2cm}}
          \]
          if
          either $\prop_{\expr_1,\dots,\expr_\numk}\alwaysoccur\sch$
          or $\neg\prop_{\expr_1,\dots,\expr_\numk}\alwaysoccur\sch$,
          and neither $\prop_{\expr_1,\dots,\expr_\numk}\maybelong\litToSch\literalset\And\schcstr\sch$
          nor $\neg\prop_{\expr_1,\dots,\expr_\numk}\maybelong\litToSch\literalset\And\schcstr\sch$%
          %, and all variables of $\expr_1,\dotsc,\expr_\numk$ are parameters
          .
          \vspace*{0.03cm}
        \item \constraintsplit.
          For $(\bigOp,\neutrOp)\in\{(\bigAnd,\trueformula),(\bigOr,\falseformula)\}$:
          %\vspace*{-0.25cm}
          \[\begin{array}{c|c}
              \multicolumn2c{(\sch[\cstrsplitlhm],\literalset)}\\\hline
              (\sch[\cstrsplitlhm]\And\exists\var\cstr,\literalset)
              \ \ & \
              (\sch[\neutrOp]\And\forall\var\neg\cstr,\literalset)
          \end{array}
          \Courte{\vspace*{-0.25cm}}
          \]
          if $\schcstr\sch\And\forall\var\neg\cstr$ is satisfiable
          and free variables of $\cstr$ other than $\var$ are parameters.

        \item Rewriting:
          \Courte{\vspace*{-0.3cm}}
            \[
              \begin{array}c
              (\sch_1,\literalset)\\\hline
              (\sch_2,\literalset)
              \end{array}
              \Courte{\vspace*{-0.2cm}}
            \]
            where $\schcstr{\sch_2}=\schcstr{\sch_1}$ and $\schpat{\sch_1}\to\schpat{\sch_2}$ by the following rewrite system:
            \begin{itemize}
              \item \algebraic. For every pattern $\pat$:
                \CourteLongue
                {\smallskip\par$\begin{aligned}
                    \neg\trueformula              & \to\falseformula & 
                    \pat\And\trueformula     & \to\pat     & \pat\And\falseformula & \to\falseformula & 
                    \textstyle\schemaAnd\var\cstr\trueformula & \to\trueformula  & \pat\And\pat     & \to\pat\\
                    \neg\falseformula             & \to\trueformula  & 
                    \pat\Or\trueformula      & \to\trueformula  & \pat\Or\falseformula  & \to\pat     & 
                    \textstyle\schemaOr\var\cstr\falseformula & \to\falseformula & 
                    \pat\Or \pat        & \to\pat\\
                    \multicolumn{6}{l}{\text{if $\context(\sch_1)\And \exists\var\cstr$ is unsatisfiable:}}
                    &\textstyle\schemaAnd \var\cstr\pat&\to\trueformula
                    &\textstyle\schemaOr \var\cstr\pat&\to\falseformula\\
                  \multicolumn{8}{l}{\text{if $\context(\sch_1)\Implies \exists\var\cstr$ is valid 
                  and $\pat$ does not contain $\var$:}}
                  & \textstyle\schemaOp\var\cstr\pat &\to\pat\\
                \end{aligned}$}
                {\[\begin{aligned}
                    \neg\trueformula              & \to\falseformula & 
                    \pat\And\trueformula     & \to\pat     & \pat\And\falseformula & \to\falseformula & 
                    \schemaAnd\var\cstr\trueformula & \to\trueformula  & \pat\And\pat     & \to\pat\\
                    \neg\falseformula             & \to\trueformula  & 
                    \pat\Or\trueformula      & \to\trueformula  & \pat\Or\falseformula  & \to\pat     & 
                    \schemaOr\var\cstr\falseformula & \to\falseformula & 
                    \pat\Or \pat        & \to\pat\\
                    \multicolumn{6}{l}{\text{if $\context(\sch_1)\And \exists\var\cstr$ is unsatisfiable:}}
                    &\schemaAnd \var\cstr\pat&\to\trueformula
                    &\schemaOr \var\cstr\pat&\to\falseformula\\
                  \multicolumn{8}{l}{\text{if $\context(\sch_1)\Implies \exists\var\cstr$ is valid 
                  and $\pat$ does not contain $\var$:}}
                  & \schemaOp\var\cstr\pat &\to\pat\\
                  \end{aligned}
                \]}%
              \item \instantiaterule.
                For $(\Op,\bigOp)\in\left\{(\And,\bigAnd),(\Or,\bigOr)\right\}$:
                \[\begin{aligned}
                  \schemaOp\var\cstr\pat&
                  \to\pat\substitution{\IRupbnd/\var}\Op\schemaOp\var{\cstr\And\var\neq\IRupbnd}\pat
                  \quad\text{ if }\instantiateRuleCondition\text{ is valid}
                  \end{aligned}
                  \Courte{\vspace*{-0.2cm}}
                \]
                $\IRupbnd$ can be chosen arbitrarily%
                \footnote{e.g. in Section \ref{sec:elem_term} we choose the maximal integer fulfilling the desired property.}.
              \item \emptiness.
                For $(\Op,\bigOp)\in\{(\And;\bigAnd),(\Or;\bigOr)\}$,
                $(\bigOpbis,\neutrOp)\in\{(\bigAnd;\trueformula),(\bigOr;\falseformula)\}$:
                \[
                  \schemaOp\var\cstr{(\pat[\emptinesslhm])} 
                  \to
                  \schemaOp\var{\cstr\And\exists\var'\cstr'}{(\pat[\emptinesslhm])} 
                  \Op
                  \schemaOp\var{\cstr\And\forall\var'\neg\cstr'}{(\pat[\neutrOp])}
                  \Courte{\vspace*{-0.2cm}}
                \]
                if $\context(\sch_1)\And\forall\var'\neg \cstr'$ is satisfiable and $\var$ occurs free in $\cstr'$.
              \item \propsimpl.
                \[\begin{aligned}
                  \prop_{\expr_1,\dotsc,\expr_\numk}&
                  \to\schemaAnd\var{(\expr_1\neq \exprf_1\Or\dotsb\Or \expr_\numk\neq \exprf_\numk)\And\var=0}{\prop_{\expr_1,\dotsc,\expr_\numk}}
                  &&\text{ if }\prop_{\exprf_1,\dotsc,\exprf_\numk}\in\literalset\\
                  \prop_{\expr_1,\dotsc,\expr_\numk}&
                  \to\schemaOr\var{(\expr_1\neq \exprf_1\Or\dotsb\Or \expr_\numk\neq \exprf_\numk)\And\var=0}{\prop_{\expr_1,\dotsc,\expr_\numk}}
                  &&\text{ if }\neg\prop_{\exprf_1,\dotsc,\exprf_\numk}\in\literalset\\
                  \end{aligned}
                  \Courte{\vspace*{-0.2cm}}
                \]
                if $\context(\sch_1)\And \expr_1=\exprf_1\And\dotsb\And \expr_\numk=\exprf_\numk$ is satisfiable.
                $\var$ is a fresh variable.
            \end{itemize}

        \item \intervalise.
          For $\intk,\intl\in\N$, $\bigOp\in\left\{\bigAnd,\bigOr\right\}$, $\lhd\in\left\{<,\leq,\geq,>\right\}$:
          %\vspace*{-0.2cm}
          \[
            \begin{array}{c|c}
              \multicolumn2c{(\sch[\schemaOp\var{\cstr\And\intk.\var\lhd \expr_1\And\intl.\var\lhd\expr_2}\pat],
              \literalset)}
              \\\hline
              (\sch[\schemaOp\var{\cstr\And\intk.\var\lhd\expr_1}\pat] \And\intl.\expr_1\!\lhd\intk.\expr_2,\literalset)
              \ \ &\ 
              (\sch[\schemaOp\var{\cstr\And\intl.\var\lhd\expr_2}\pat] \And 
              \intl. \expr_1\!\not{\!\!\lhd}\intk.\expr_2,\literalset)
              \end{array}
              \Courte{\vspace*{-0.2cm}}
           \]
           if every free variable of $\cstr$ is either $\var$ or a parameter,
           all variables of $\expr_1,\expr_2$ are parameters
           and $\numk>0$, $\numl>0$.
      \end{itemize}
      \renewcommand{\arraystretch}{\temparraystretch}
    \end{definition}

      \Courte{
        A \emph{derivation} is a (possibly infinite) sequence of tableaux $(\tab_\numi)_{\numi\in\interval}$
        s.t. $\interval$ is either $\N$ or $[0..\numk]$ for some $\numk\geq0$, and s.t. for all $\numi>0$,
        $\tab_\numi$ is obtained from $\tab_{\numi-1}$ by applying a rule.
        A derivation is \emph{fair} iff no leaf can be indefinitely ``freezed''
        i.e.
        either there is $\numi\in\interval$ s.t. $\tab_\numi$
        contains an irreducible, not closed, leaf or if for all $\numi\in\interval$ and every leaf $\node$ in
        $\tab_\numi$ there is $\numj\geq\numi$ s.t. a rule is applied on $\node$ in $\tab_\numj$.
        \begin{theorem}[Soundness and Completeness w.r.t. Satisfiability]
          \label{thm:soundness_and_completeness}
          Consider a fair derivation $(\tab_\numi)_{\numi\in\interval}$.
          $\tab_0$ is satisfiable iff there is $\numi\in\interval$ s.t. $\tab_\numi$ contains an irreducible and not closed leaf.
        \end{theorem}
        Appendix \ref{sec:example} (and \cite{rapport09}) contains an example of a tableau generated by \SchDP.
      }

      \CourteLongue\section\subsection{Looping Detection}
      \label{sec:looping_detection}
      The above extension rules do not terminate in general,
      but this is not surprising as the satisfiability problem is undecidable \cite{tab09}.
      Non-termination comes from the fact that iterations can be infinitely unfolded 
      (consider e.g. $\schematicOr\var1\param\prop_\var\And\neg\prop_\var$),
      thus leading to infinitely many new schemata.
      However\CourteLongue{, often,}{ it is often the case that} newly obtained schemata have already been seen
      (up to some relation that remains to be defined)
      i.e. the procedure is \emph{looping}
      (e.g. $\schematicOr\var1\param\prop_\var\And\neg\prop_\var$ \Longue{will }generate\Courte{s}
      $\schematicOr\var1{\param-1}\prop_\var\And\neg\prop_\var$,
      then $\schematicOr\var1{\param-2}\prop_\var\And\neg\prop_\var$,
      then \dots
      which are all equal up to a shift of $\param$).
      This is actually an algorithmic interpretation of a proof by mathematical induction.
      We now define precisely the notion of looping.
      \Longue\par%
      We start with a very general definition:
      \begin{definition}[Looping]
        \label{def:looping_schemata}
        Let $\sch_1,\sch_2$ be two schemata
        having the same parameters $\param_1,\dotsc,\param_\numk$,
        we say that $\sch_1$ \emph{loops} on $\sch_2$ iff 
        for every model $\interp$ of $\sch_1$ 
        there is a model $\interpj$ of $\sch_2$
        s.t. $\interpenv\interpj(\param_\numj) < \interpenv\interp(\param_\numj)$ for some $\numj\in1..\numk$
        and $\interpenv\interpj(\param_\numl) \leq \interpenv\interp(\param_\numl)$ for every $\numl\neq\numj$.
        The induced relation among schemata is called the \emph{looping relation}.
      \end{definition}
      Looping is undecidable
      (e.g. if $\sch_2=\falseformula$ then $\sch_1$ loops on $\sch_2$ iff $\sch_1$ is unsatisfiable).
      It is trivially transitive.
      An advantage of Definition \ref{def:looping_schemata} is that,
      contrarily to the definitions found in \cite{tab09,sd09},
      it is independent of the \Longue{considered }proof procedure.
      \Longue{However we still have to make precise the link with \SchDP:}
      \begin{definition}
        \label{def:looping_nodes}
        Let $\node,\nodeb$ be nodes in a tableau $\tab$.
        $\nodeschfull\node\tab$ denotes the schema $\nodesch\node\tab\And\nodeinterpsch\node\tab$.
        Then $\nodeb$ \emph{loops} on $\node$ iff $\nodeschfull\nodeb\tab$ loops on $\nodeschfull\node\tab$%
      .\end{definition}
      \Longue{Following the terminology of \cite{Brotherston}, 
      $\nodeb$ is called a \emph{bud} node and $\node$ its \emph{companion} node. }%
      The \looping rule closes a leaf that loops on some existing node of the tableau.
      From now on, \SchDP denotes the extension rules, plus the \looping rule.
      \Courte{Theorem \ref{thm:soundness_and_completeness} still holds \cite{rapport09}.}
      \Longue{\par An example of a tableau generated by \SchDP can be found in Appendix \ref{sec:example}.

      \subsection{Soundness and Completeness}
        \begin{definition}
         \label{def:tableaux_semantics}
         Let $\interp$ be an interpretation and $\node$ a leaf of a tableau $\tab$.
         We write $\interp \models_\tab \node$ iff $\interp \models \nodeschfull\node\tab$
         (or simply $\interp \models \node$ when $\tab$ is obvious from the context).
         We write $\interp\models\tab$ iff there exists a leaf $\node$ in $\tab$ s.t. $\interp\models\node$.
         The definitions of \emph{model} and \emph{satisfiability} naturally extend.
        \end{definition}
        \begin{lemma}
         \label{lem:equ}
         Let $\tab$, $\tab'$ be tableaux s.t. $\tab'$ is obtained by applying an extension rule
         (i.e. any rule except the \looping rule)
         on a leaf $\node$ of $\tab$.
         Let $\interp$ be an interpretation.
         $\interp \models \node$
         iff there exists a child $\nodeb$ of $\node$ in $\tab'$ s.t. $\interp \models \nodeb$.
        \end{lemma}
        \begin{proof}
          For \propsplit there are two branches $\nodeb_1$, $\nodeb_2$.
          If $\interp\models\node$ then
          either $\interpprop\interp \models \envapp{\prop_{\expr_1,\dots,\expr_k}}{\interpenv\interp}$
          or     $\interpprop\interp \models \envapp{\neg\prop_{\expr_1,\dots,\expr_k}}{\interpenv\interp}$,
          and consequently either $\interp\models\nodeb_1$ or $\interp\models\nodeb_2$.
          Conversely it is easily seen that if we have 
          $\interp\models\nodeb_1$ or $\interp\models\nodeb_2$
          then 
          $\interp\models\node$.

          Similarly we write $\nodeb_1$, $\nodeb_2$ for the two branches of \constraintsplit.
          By completeness of linear arithmetic,
          either $\models\exists\var(\cstr\interpenv\interp)$
          or $\models\forall\var\neg(\cstr\interpenv\interp)$
          for any interpretation $\interp$
          (notice that by the application condition of the rule,
          all variables occurring in $\cstr$ are parameters,
          thus $\var$ is the only free variable of $\cstr\interpenv\interp$).
          Suppose $\interp\models\node$, 
          then in the first case $\interp\models\nodeb_1$,
          in the second case the iteration is empty so $\interp\models\nodeb_2$.
          Conversely if $\interp\models\nodeb_1$ then it is trivial that $\interp\models\node$
          and        if $\interp\models\nodeb_2$ then $\models\forall\var\neg(\cstr\interpenv\interp)$
          and thus $\envapp\cstrsplitlhm{\interpenv\interp}=\neutrOp$
          (following the notations of \constraintsplit), so $\interp\models\node$.

          We do not detail all the rewrite rules, which have only one conclusion.
          Suppose $\nodeb$ is obtained from $\node$ by rewriting a schema $\sch_1$ into $\sch_2$.
          Let $\interp$ be a model of $\node$ or $\nodeb$
          (whether one proves the ``only if'' or the ``if'' implication of the lemma).
          It is easily proved (using the side conditions of each rewrite)
          that 
          $\interpprop\interp(\envapp{\sch_1}{\interpenv\interp})
          =\interpprop\interp(\envapp{\sch_2}{\interpenv\interp})$,
          i.e. for any $\interp$ the propositional realizations under $\interpenv\interp$ 
          of $\sch_1$ and $\sch_2$ have the same value w.r.t. $\interpprop\interp$.
          Actually we even have that for all rules except \propsimpl,
          and for every environment $\anenv$ of $\sch_1$,
          $\envapp{\sch_1}\anenv$ and $\envapp{\sch_2}\anenv$
          are equivalent.

          Consider \intervalise.
          Suppose we have $\interp\models\node$
          then either $(\intl.\expr_1\lhd\intk.\expr_2)\interpenv\interp$
          or $(\intl.\expr_1\not{\!\!\lhd}\intk.\expr_2)\interpenv\interp$
          is valid (by completeness of linear arithmetic).
          Furthermore, it is easily seen that for every $\var$,
          $\intl.\expr_1\lhd\intk.\expr_2$ and $\intk.\var\lhd\expr_1$ entail $\intl.\var\lhd\expr_2$ 
          (as $\intk,\intl>0$; one has to carefully make the distinction between the cases $\expr_1=0$ and $\expr_1\neq0$),
          and $\intl.\expr_1\not{\!\!\lhd}\intk.\expr_2$ and $\intl.\var\lhd\expr_2$ entail $\intk.\var\lhd\expr_1$.
          Consequently, in both cases removing the entailed constraint does not affect the propositional realization of the iteration,
          and thus $\interp\models\nodeb_1$ or $\interp\models\nodeb_2$.
          Now suppose $\interp\models\nodeb_1$.
          Then $(\intl.\expr_1\lhd\intk.\expr_2)\interpenv\interp$ is valid.
          Thus we have in the exact same way: $\forall\var(\numk.\var\lhd\expr_1\Implies\numl.\var\lhd\expr_1)\interpenv\interp$.
          And thus $(\numk.\var\lhd\expr_1)\interpenv\interp$
          is equivalent to $(\numk.\var\lhd\expr_1\And\numl.\var\lhd\expr_1)\interpenv\interp$.
          So $\interp\models\node$.
          The case $\interp\models\nodeb_2$ is similar.
          \qed
        \end{proof}
        A leaf is \emph{irreducible} iff no rule of \SchDP applies to it.
        \begin{lemma}
         \label{lem:irr}
         If a leaf $\node$ in $\tab$ is irreducible and not closed then $\tab$ is satisfiable.
        \end{lemma}
        \begin{proof}
         We first show that $\nodesch\node\tab$ does not contain iterations.
         If there are iterations then there are iterations which are not contained into any other iteration.
         Let $\schemaOp\var\cstr\pat$ be such an iteration.
         $\schemaOp\var\cstr\pat$ cannot be empty by irreducibility of \constraintsplit and \emptiness.
         So by irreducibility of \intervalise and elimination of quantifiers in linear arithmetic,
         $\cstr$ can be restricted to a non-empty disjunction of inequalities $\expr_1\leq\var\And\var\leq\expr_2$.
         Thus $\cstr\substitution{\expr_1/\var}$ is valid and \instantiaterule can apply which is impossible.
         Hence there cannot be any iteration which is not contained into any other iteration,
         thus there cannot be any iteration at all.
         So $\nodesch\node\tab$ is constructed only with $\And,\Or,\neg$ and indexed propositions.
         Hence it is easily seen that any literal $\lit$ s.t. $\lit\litIn\nodesch\node\tab$
         satisfies $\lit\alwaysoccur\nodesch\node\tab$.
         Thus either $\lit\in\nodeinterp\node\tab$ or $\lit^c\in\nodeinterp\node\tab$
         by irreducibility of \propsplit.
         As a consequence if there existed such a literal,
         \propsimpl and then \algebraic would have applied,
         turning every occurrence of $\lit$ into $\falseformula$ or $\trueformula$.
         Hence $\nodesch\node\tab$ does not contain any literal,
         and by irreducibility of \algebraic
         $\nodesch\node\tab$ is either $\falseformula$, impossible as the branch is not closed,
         or $\trueformula$, which is satisfiable.
         Finally $\literalset$ cannot contain two contradictory \Schlit{}s,
         because the application conditions of \propsplit ensure that 
         $\lit$ is added to $\nodeinterp\node\tab$ only if 
         neither $\lit   \maybelong \nodeinterpsch\node\tab\And\schcstr{\nodesch\node\tab}$
         nor     $\lit^c \maybelong \nodeinterpsch\node\tab\And\schcstr{\nodesch\node\tab}$.
         We conclude with Lemma \ref{lem:equ} that the initial tableau is satisfiable.
         \qed
        \end{proof}
        \begin{theorem}[Soundness]
         \label{thm:soundness}
         Let $\tab$ be a tableau. 
         If a tableau $\tab'$ is obtained from $\tab$ by application of the extension rules,
         and if $\tab'$ contains an irreducible and not closed leaf
         then $\tab$ is satisfiable.
        \end{theorem}
        \begin{proof}
         This follows immediately from Lemmata \ref{lem:equ} and \ref{lem:irr}.
         \qed
        \end{proof}
         
        We now prove that the procedure is complete w.r.t. satisfiability
        i.e. that if $\sch$ has a model then every sequence of tableaux constructed from $\sch$
        (in a fair way) eventually contains an irreducible and not closed branch.
        To do this we assume the existence of a model,
        then we define a well-founded measure w.r.t. this model
        and we show that it strictly decreases at each rule application (Lemma \ref{lem:decrease}).
        Thus there will be a leaf s.t. this measure is minimal, so no rule can apply on it.
        We then use Lemma \ref{lem:equ} to show that this leaf cannot be closed.
        This is formalized in the proof of Theorem \ref{thm:completeness}.

        Intuitively, we take a model $\interp$
        and apply \SchDP by focusing only on the branch for which $\interp$ is a model
        (by Lemma \ref{lem:equ}, there always exists such a branch).
        Then, in this branch, all iterations will progressively be unfolded.
        This process will stop because the iteration has a fixed length in $\interp$.
        Concretely, the iterations will be unfolded rank by rank until it only remains an empty iteration
        which will then be removed either by \constraintsplit or \emptiness.
        Once this is done for all iterations, what remains is the propositional realization of the original schema w.r.t. $\interpenv\interp$
        (except that in the meantime, some literals may have been evaluated, leading to some possible simplifications).
        So we have a propositional formula and \algebraic applies until we obtain $\trueformula$,
        i.e. a node which is irreducible and not closed.

        As the reader will see, the presented measure is not trivial
        (in particular $\measureAlone3\interp$).
        We first outline the encountered problems that justify such a definition.
        From the explanations of the previous \S,
        the measure must be greater when the schema contains an iteration
        (and the longer the iteration is, the greater shall be the measure).
        For instance such a measure would be strictly lower after an application of \instantiaterule.
        \emptiness divides an iteration into two iterations such that the sum 
        of their lengths is equal to the length of the original iteration, thus the measure remains the same.
        This is easily circumvented, e.g. by squaring the length of iterations.
        A bigger problem occurs with \propsimpl which \emph{adds} iterations where there was no iteration before.
        A natural solution is to define another measure that decreases on \propsimpl,
        e.g. the number of possible applications of the rule.
        Then we give this measure a higher priority (via a lexicographic ordering).
        But this does not work because \instantiaterule duplicates the pattern $\pat$
        (following the notations of the rule)
        and can thus increase the number of possible applications of \propsimpl.
        Similar problems are encountered with \emptiness:
        this rule also makes a kind of unfolding
        but, following the notations of the rule, $\cstr\And\exists\var'\cstr'$ can be unsatisfiable.
        In such a case, it means that the unfolding is fake,
        it just allows us to introduce the information $\forall\var'\neg\cstr'$ in the rightmost iteration.
        So in this case we have just introduced a new iteration, without even decomposing the original one.
        Once again we could define another measure, e.g. the number of possible applications of \emptiness,
        but this is increased by \instantiaterule.

        We now present formally our solution which requires the two following ``auxiliary'' functions:
        \[
        \begin{aligned}
          \measIt(x,0)       & =(x+2)^2                  & \measLit(0)       & =1\\
          \measIt(x,\numk+1) & =(\measIt(x,\numk)+x+2)^2 & \measLit(\numk+1) & =\measIt(\measLit(\numk),0)+1\\
        \end{aligned}
        \]
        The following results are easily proved (most of them by induction).
        They sum up all the properties of $\measLit$ and $\measIt$ that are useful to prove that the measure decreases.
        \begin{proposition}
          \label{thm:measure_prop}
          \noindent
          \begin{enumerate}
            \item $\forall x,\numk\in\N$, $\measIt(x,\numk)\geq4$
            \item $\forall x,\numk\in\N$, $\measIt(x,\numk)\geq x$
            \item $\forall x,y,\numk\in\N$, $x<y\Implies\measIt(x,\numk)<\measIt(y,\numk)$
            \item $\forall \numk_1,\numk_2\in\N$, $\numk_1<\numk_2\Implies\measLit(\numk_1)<\measLit(\numk_2)$
            \item $\forall x,y,\numk\in\N$ s.t. $y\geq1$, $\measIt(x,\numk)+y<\measIt(x+y,\numk)$
            \item $\forall x,y,\numk\in\N$, $\measIt(x,\numk)+\measIt(y,\numk)<\measIt(x+y,\numk+1)$
            \item $\forall \numk\in\N$, $\measLit(\numk+1)>\measIt(\measLit(\numk),0)$
          \end{enumerate}
        \end{proposition}
        We can now define the measure.
        Let $\interp$ be an interpretation and $\node$ a node of a tableau $\tab$.
        We set $\measurebranch\interp\node\tab\eqdef
        (\measure1\interp\node,\measure2\interp\node,\measure3\interp{\nodesch\node\tab},\measure4{}\node,\measure5{}\node)$
        ordered using the lexicographic extension of the usual ordering on natural numbers,
        where $\measure1\interp\node$, $\measure2\interp\node$, $\measure3\interp{\nodesch\node\tab}$, $\measure4{}\node$, $\measure5{}\node)$
        are defined as follows:
        \begin{enumerate}
          \item For a parameter $\param$ of $\nodesch\node\tab$,
            $\measure1{\interp,\param}\node\eqdef \interpenv\interp(\param)$.
           $\measure1\interp\node$ is defined as the multiset extension of 
           $\measure1{\interp,\param}\node$ to all parameters of $\nodesch\node\tab$.
         \item ${\measure2\interp\node}$ is the number of atoms (different from $\falseformula,\trueformula$) that occur in
           $\envapp{\nodesch\node\tab}{\interpenv\interp}$ but not in $\envapp{\nodeinterpsch\node\tab}{\interpenv\interp}$.
         \item  $\measure3\interp{\nodesch\node\tab}$ is defined by induction on the structure of $\schpat{\nodesch\node\tab}$:
           \begin{itemize}
             \item{$\measure3\interp\trueformula \isdef\measure3\interp\falseformula \isdef 1$}
             \item{$\measure3\interp{\neg\pat} \isdef \measure3\interp\pat+1$.}
             \item{$\measure3\interp{\pat_1 \Op \pat_2} \isdef \measure3{\interp}{\pat_1} 
               + \measure3\interp{\pat_2}$
               }
             \item{$\measure3\interp{\schemaOp\var\cstr\pat} \isdef 
               \measIt(\sum_{\inti\in\concDom} \measure3\interp{\pat\substitution{\inti/\var}}, \nbSubIts)$ 
               where $\concDom$ is the set $\{\inti\in\Z \mid \cstr\interpenv\interp\substitution{\inti/\var}\text{ is valid}\}$
               ($\concDom$ is finite since $\cstr$ encloses $\var$)
               and $\nbSubIts$ is the number of iterations $\emptinesslhm$ occurring in $\pat$
               s.t. \emptiness can apply on $\schemaOp\var\cstr\pat[\emptinesslhm]$
               (with the notations of \emptiness).
               }
             \item{$\measure3\interp{\prop_{\expr_1,\dots,\expr_\numk}} \isdef \measLit(\nbLits)$
               where $\nbLits$ is the number of literals $\prop_{\exprf_1,\dots,\exprf_\numk}\in\nodeinterp\node\tab$
               s.t. \propsimpl applies on $\prop_{\expr_1,\dots,\expr_\numk}$
               %w.r.t. $\prop_{\exprf_1,\dots,\exprf_\numk}$
               .
               }
           \end{itemize}
         \item $\measure4{}\node$ is the number of possible applications of \intervalise on $\node$.
         \item $\measure5{}\node$ is the number of iterations $\schemaOp\var\cstr\pat$ of $\nodesch\node\tab$ s.t.
           $\schcstr{\nodesch\node\tab}\And\forall\var\neg\cstr$ is satisfiable.
        \end{enumerate}

        An \emph{extended child} of a node $\node$ is 
        a child of $\node$ if $\node$ is not a bud node,
        or the companion node of $\node$ otherwise.
        For extended children we have the following weaker version of Lemma \ref{lem:equ}:
        \begin{proposition}
          \label{thm:ext_equ}
          Let $\tab$, $\tab'$ be tableaux s.t. $\tab'$ is obtained by applying any rule of \SchDP on a leaf $\node$ of $\tab$.
          If $\node$ is satisfiable then there exists a satisfiable extended child of $\node$ in $\tab'$.
        \end{proposition}
        \begin{proof}
          Indeed if $\node$ is a bud node then it follows from Definitions \ref{def:looping_schemata} and \ref{def:looping_nodes},
          otherwise it follows from Lemma \ref{lem:equ}.
          \qed
        \end{proof}
        Let $\interp$ be a model of $\node$ and $\nodeb$ a satisfiable extended child of $\node$,
        $\childmodel\interp\node\nodeb$ is defined as follows:
        if $\node$ is a bud node then there is $\interpj$ s.t. 
        $\interpj(\param)<\interp(\param)$ for some parameter $\param$ and $\interpj\models\nodeb$.
        We set $\childmodel\interp\node\nodeb\eqdef\interpj$.
        If $\node$ is not a bud node then $\childmodel\interp\node\nodeb\eqdef\interp$.
        We can now prove the main lemma, which states that the measure strictly decreases when applying a rule.
        \begin{lemma}
         \label{lem:decrease}
         Let $\tab$, $\tab'$ be tableaux s.t.
         $\tab'$ is deduced from $\tab$ by applying a rule on a leaf $\node$.
         If there is a model $\interp$ of $\node$
         then for every satisfiable extended child $\nodeb$ of $\node$ in $\tab'$
         we have $\measurebranch{\childmodel\interp\node\nodeb}\nodeb{\tab'} < \measurebranch\interp\node\tab$.
        \end{lemma}
        \begin{proof}
         By inspection of the extension rules:
         \begin{center}
         \begin{tabular}{cccccccc}
           Rule &$\measure1\interp\node$&$\measure2\interp\node$&$\measure3\interp{\nodesch\node\tab}$&$\measure4{}\node$&$\measure5{}\node$\\\hline
           \propsplit       & $\leq$ & $<$    &        &        &     \\\hline
           \constraintsplit & $\leq$ & $\leq$ & $\leq$ & $\leq$ & $<$ \\\hline
           \algebraic       & $\leq$ & $\leq$ & $<$    &        &     \\\hline
           \instantiaterule & $\leq$ & $\leq$ & $<$    &        &     \\\hline
           \emptiness       & $\leq$ & $\leq$ & $<$    &        &     \\\hline
           \propsimpl       & $\leq$ & $\leq$ & $<$    &        &     \\\hline
           \intervalise     & $\leq$ & $\leq$ & $\leq$ & $<$    &     \\\hline
           \looping         & $<$    &        &        &        &     \\\hline
         \end{tabular}
         \end{center}
         $\leq$ (resp. $<$) means that the corresponding measure \emph{does not increase} 
         (resp. \emph{strictly decreases}) by application of the rule.
         \begin{itemize}
           \item For $\measure1\interp\node$, \looping strictly decreases by Definition \ref{def:looping_schemata}.
             In all other cases, $\childmodel\interp\node\nodeb=\interp$ so all parameters have the same values,
             thus $\measure1\interp\node$ is constant.
           \item When \propsplit applies,
             either $\prop_{\expr_1,\dots,\expr_\numk}\alwaysoccur\sch$
             or $\neg\prop_{\expr_1,\dots,\expr_\numk}\alwaysoccur\sch$ 
             (following the notations of \propsplit),
             so either $\envapp{\prop_{\expr_1,\dots,\expr_\numk}}{\interpenv\interp}\litIn\envapp{\nodesch\node\tab}{\interpenv\interp}$
             or $\envapp{\neg\prop_{\expr_1,\dots,\expr_\numk}}{\interpenv\interp}\litIn\envapp{\nodesch\node\tab}{\interpenv\interp}$.
             Hence either $\prop_{\expr_1,\dots,\expr_\numk}$
             or $\neg\prop_{\expr_1,\dots,\expr_\numk}$ is added to $\nodeinterp\node\tab$,
             thus $\measure2\interp\node$ strictly decreases.
             It is obvious that other rules (except \looping) cannot increase the number of atoms in $\envapp{\nodesch\node\tab}{\interpenv\interp}$
             (even rules that duplicate a pattern, namely \instantiaterule or \emptiness:
             indeed the number of atoms is increased in the schema but its propositional realization remains the same,
             see the proof of Lemma \ref{lem:equ})
             thus they cannot increase $\measure2\interp\node$.
             \looping has already been shown to be decreasing so
             we do not mind that this rule possibly increases
             due to the lexicographic ordering.
             The same argument allows us to omit the \propsplit rule in the following,
             and similarly for the each subsequent measure.
           \item We detail the case of $\measure3\interp{\nodesch\node\interp}$, rule by rule
             (the notations used here --- $\nbLits$, $\nbSubIts$, $\concDom$ --- are the same 
             as in the definition of $\measureAlone3\interp$):
             \begin{enumerate}
               \item \constraintsplit:
                 the pattern does not change but one must take care that $\context(\nodesch\node\tab)$ does change,
                 so $\nbLits$ and $\nbSubIts$ may increase.
                 However \constraintsplit only strengthens the context and so cannot increase those numbers.
               \item \algebraic: obvious by inspection of all rules.
                 The two first items of Proposition \ref{thm:measure_prop} enable to conclude for rules involving an iteration.
               \item \instantiaterule:
                 the result follows from the fifth item of Proposition \ref{thm:measure_prop}
                 (taking $y=\measure3\interp{\pat\substitution{\expr/\var}}$;
                 $\measure3\interp\pat\geq1$ for every pattern $\pat$ so indeed $y\geq1$).
                 Similarly to \constraintsplit, $\nbLits$ and $\nbSubIts$ cannot increase.
               \item For \emptiness, the result follows from the sixth item of Proposition \ref{thm:measure_prop}:
                 it is easily seen that $\nbSubIts$ strictly decreases
                 from the application conditions of \emptiness
                 and because those conditions are not satisfied anymore after the rewrite.
                 $\measIt$ and $\nbSubIts$ have been precisely defined to handle this rule.
               \item For \propsimpl, the result follows from the seventh item of Proposition \ref{thm:measure_prop}:
                 it is obvious that $\nbLits$ strictly decreases during the rewrite.
                 $\measLit$ and $\nbLits$ have been precisely defined to handle this rule.
               \item \intervalise only changes the domain of an iteration.
                 With a similar reasoning as in Lemma \ref{lem:equ},
                 one easily gets that in both branches the set $\concDom$ remains the same.
                 Furthermore, similarly to the \constraintsplit case, $\nbLits$ and $\nbSubIts$ cannot increase.
             \end{enumerate}
           \item The last two measures follow the conditions of the corresponding decreasing rules,
             easily entailing their corresponding behaviors.
             It is obvious that \constraintsplit cannot increase $\measure4{}\node$.
         \qed
         \end{itemize}
        \end{proof}
        A \emph{derivation} is a (possibly infinite) sequence of tableaux $(\tab_\numi)_{\numi\in\interval}$
        s.t. $\interval$ is either $[0..\numk]$ for some $\numk\geq0$, or $\N$ and s.t. for all $\numi>0$,
        $\tab_\numi$ is obtained from $\tab_{\numi-1}$ by applying one of the rules. 
        A derivation is \emph{fair} iff either there is $\numi\in\interval$ s.t. $\tab_\numi$
        contains an irreducible, not closed, leaf or if for all $\numi\in\interval$ and every leaf $\node$ in
        $\tab_\numi$ there is $\numj\geq\numi$ s.t. a rule is applied on $\node$ in $\tab_\numj$ 
        (i.e. no leaf can be indefinitely ``freezed'').
        \begin{theorem}[Model Completeness]
          \label{thm:completeness}
          Let $\tab_0$ be a satisfiable tableau.
          If $(\tab_\numi)_{\numi\in\interval}$ is a fair derivation then 
          there are $\numk\in\interval$ and a leaf $\node_\numk$ in $\tab_\numk$ 
          s.t. $\node_\numk$ is irreducible and not closed.
        \end{theorem}
        \begin{proof}
         By Lemma \ref{lem:equ}, for every model $\interp$ of $\tab_0$ and for all $\numk\in\interval$,
         $\tab_\numk$ contains a leaf $\node_\numk$ s.t. $\interp \models \node_\numk$.
         Consider such $\interp,\numk,\node_\numk$ s.t.
         $\measurebranch\interp{\node_\numk}{\tab_\numk}$ is minimal
         (exist since $\measurebranch\interp{\node_\numk}{\tab_\numk}$ is well-founded).
         By Lemma \ref{lem:equ}, $\node_\numk$ is not closed.
         Suppose that $\node_\numk$ is not irreducible.
         Then, since the derivation is fair, there is $\numl>\numk$ s.t.
         a rule is applied on $\node_\numk$ in the tableau $\tab_\numl$.
         By Proposition \ref{thm:ext_equ} there exists a satisfiable extended child $\nodeb$ of $\node_\numk$ in $\tab_\numl$
         and $\measurebranch{\childmodel\interp\node\nodeb}\nodeb{\tab_\numl}<\measurebranch\interp{\node_\numk}{\tab_\numk}$
         by Lemma \ref{lem:decrease}.
         This is impossible by minimality of $\measurebranch\interp{\node_\numk}{\tab_\numk}$.
         Hence $\node_\numk$ is irreducible.
         \qed
        \end{proof}

  \section{Looping Refinements}%
  \label{sec:termination}}%
    The notion of loop introduced in Definition \ref{def:looping_nodes} is undecidable,
    thus, in practice, we use decidable refinements of looping.
    \Longue{\begin{definition}
      A \Longue{binary }relation between schemata is a \emph{looping refinement} iff it is a subset of the looping relation.
    \end{definition}}%
    Termination proofs work by showing that the set of schemata
    which are generated by the procedure is finite up to some \Longue{(decidable) }looping refinement.
    We make precise this notion:
    \begin{definition}
      \label{def:lemma}
      \Courte{A binary relation between schemata is a \emph{looping refinement} iff it is a subset of the looping relation.}
      \CourteLongue{\let\lemma=\sch}{\def\lemma{\repr\sch}}%
      Let $\setofschemata$ be a set of schemata and $\refinement$ a looping refinement.
      A schema $\lemma\in\setofschemata$ is a \emph{$\refinement$-\representative}%
      \Longue{(or just \emph\representative when $\refinement$ is obvious from the context) }%
      w.r.t. $\setofschemata$ iff
      there is no $\sch'\in\setofschemata$ s.t. $\lemma\refinement\sch'$.
      The set of all $\refinement$-\representative{}s w.r.t. $\setofschemata$ is written $\finitesubset\setofschemata\refinement$.
      If \CourteLongue{it}{$\finitesubset\setofschemata\refinement$} is finite then
      we say that $\setofschemata$ is \emph{finite up to $\refinement$}.
    \end{definition}
    \Longue{Notice that we use the notations $\repr\sch$ and $\finitesubset\setofschemata\refinement$ as if 
    we were talking of an equivalence class and a quotient set
    \emph{but $\refinement$ is generally not an equivalence.}
    However the underlying intuition is often very close 
    and we think that using this notation makes it easier to understand the proofs,
    as soon as the reader is clearly aware that this is not an equivalence relation.}

    \subsection{Equality up to a Shift}
      \label{sec:eq_shift}
      We now present perhaps the simplest refinement of looping.
      A \emph\Shiftable is a schema, a \constraint, a pattern, a \linear or a tuple of those.
      The refinement is defined on \Shiftable\s (and not only on schemata)
      in order to handle those objects in a uniform way.
      This is useful in the termination proof of Section \ref{sec:complete}\Courte{ (and even more in \cite{rapport09})}.
      \begin{definition}
      \label{def:eqshift}
        Let $\shiftable$, $\shiftable'$ be \Shiftable\s and $\param$ a variable.
        If $\shiftable'=\shiftable\substitution{\param-\intk/\param}$ for some $\intk>0$,
        then $\shiftable'$ is \emph{equal to $\shiftable$ up to a shift of $\intk$ on $\param$},
        written $\shiftable'\stdeqshift\shiftable$ 
        (or $\shiftable'\stdeqshift_\intk\shiftable$ when we want to make $\intk$ explicit).
      \end{definition}
      \CourteLongue
      {The fact that we use syntactic equality makes the refinement less powerful but simpler to implement and easier to reason with.}
      {Notice that we use a \emph{syntactical} equality
      e.g. we do not care about associativity or commutativity of $\And$ and $\Or$ when the \Shiftable\s are schemata,
      nor do we use \constraint equivalence when the \Shiftable\s are \constraints.
      This makes this refinement less powerful but trivial to implement and easier to reason with.}
      \CourteLongue
      {It is easy to check that $\stdeqshift$ is a looping refinement when restricted to schemata having $\param$ as a parameter.}
      {\begin{proposition}
        \label{thm:eqshift_various}
        Let $\param$ be a variable, the restriction of $\stdeqshift$ to schemata 
        having $\param$ as a parameter is a looping refinement.
      \end{proposition}
        \begin{proof}
          Let $\sch_1,\sch_2$ be schemata s.t. $\sch_1\stdeqshift_\intk\sch_2$ for some $\intk>0$.
          Let $\interp$ be a model of $\sch_1$.
          We define $\interpj$ s.t. 
          $\interpenv\interpj(\param) \eqdef \interpenv\interp(\param)-\intk$,
          $\interpenv\interpj(\paramm) \eqdef \interpenv\interp(\paramm)$ for $\paramm\neq\param$
          and $\interpprop\interpj\eqdef\interpprop\interp$.
          It is obvious that $\envapp{\sch_1}{\interpenv\interp} = \envapp{\sch_2}{\interpenv\interpj}$
          and as $\interpprop\interpj=\interpprop\interp$,
          $\interpj\models\sch_2$.
          It is also obvious that $\interpenv\interpj(\param) < \interpenv\interp(\param)$.
          \qed
        \end{proof}
      }
      \Longue{
        \begin{proposition}
          \label{thm:triangle}
          For all \Shiftable\s $\shiftable$, $\shiftable_1$, $\shiftable_2$,
          if $\shiftable_1\stdeqshift\shiftable$ and $\shiftable_2\stdeqshift\shiftable$ 
          then either $\shiftable_1\stdeqshift\shiftable_2$ or 
          $\shiftable_2\stdeqshift\shiftable_1$ or $\shiftable_1=\shiftable_2$.
        \end{proposition}
      }%
      \CourteLongue{Furthermore}{Finally}
      $\stdeqshift$ is transitive but neither reflexive\Longue{ (e.g. $\prop_\param\not\stdeqshift\prop_\param$)},
      nor irreflexive\Longue{ (e.g. $\prop_1\stdeqshift\prop_1$)}.
      It is irreflexive for \Shiftable\s containing $\param$,
      and reflexive for \Shiftable\s not containing $\param$%
      \Longue{ (in which case equality up to a shift just amounts to equality)}.
      \Longue{
        \begin{definition}
          A set of \Shiftable\s $\setofshiftables$ s.t. all its different elements are comparable 
          w.r.t. $\stdeqshift$ is called a \emph{\chain}.
          We extend the notion of \representative to \Shiftable\s:
          a \Shiftable $\shiftable$ is a \emph{\representative{}} w.r.t. a set of \Shiftable\s $\setofshiftables$ 
          iff there is no $\shiftable'\in\setofshiftables$ s.t. $\shiftable\stdeqshift\shiftable'$.
          If a \chain $\setofshiftables$ contains a \Shiftable $\shiftable$ which is a \representative w.r.t. $\setofshiftables$
          then $\setofshiftables$ is a \emph{\wfchain{}}.
        \end{definition}
        From the previous remarks a \chain has the form:
        $\dotsb\stdeqshift\shiftable_{i-1}\stdeqshift\shiftable_i\stdeqshift\shiftable_{i+1}\stdeqshift\dotsb$,
        hence justifying the name ``\chain''.
        Then, by considering all its totally comparable subsets, any set of schemata can be seen as a union of \chain\s.
        A \wfchain has the form
        $\dotsb\stdeqshift\shiftable_2\stdeqshift\shiftable_1\stdeqshift\shiftable_0$,
        where $\shiftable_0$ is a \representative w.r.t. the chain.
      }%

      We focus now on sets which are \emph{finite up to equality up to a shift}, in short ``\finshift''%
      \Longue{ (i.e. sets which are finite unions of \wfchain\s)}:
      termination proofs go by showing that the set of all schemata possibly generated by \SchDP is \finshift,
      thus ensuring that the \looping rule will eventually apply.
      To prove such results we need to reason by induction on the structure of a schema.
      To do this properly we need closure properties for \finshift sets
      i.e. if we know that two sets are \finshift,
      we would like to be able to combine them and preserve the \finshift property.
      This is generally not possible,
      e.g. for two \finshift sets of \Shiftable\s $\setofshiftables_1$ and $\setofshiftables_2$,
      the set $\setofshiftables_1\times\setofshiftables_2$
      (remember that \Shiftable\s are closed by tuple construction)
      is generally \emph{not} \finshift.
      For instance take
      $\setofshiftables_1=\{\prop_\param,\prop_{\param-1},\prop_{\param-2},,\dotsc\}$
      ($\setofshiftables_1$ is \finshift with $\stdsubset{\setofshiftables_1}=\{\prop_\param\}$)
      and $\setofshiftables_2=\{\prop_\param\}$ (which is finite and thus \finshift).
      Then $\{(\prop_\param,\prop_\param), (\prop_{\param-1},\prop_\param), (\prop_{\param-2},\prop_\param), \dotsc\}$ is not \finshift:
      indeed for every $\numi\in\N$, $(\prop_{\param-\numi},\prop_\param)$ is a \representative in $\setofshiftables_1\times\setofshiftables_2$,
      there is thus an infinite set of \representative\s.
      Consequently $\setofshiftables_1\times\setofshiftables_2$ is not \finshift.
      \Longue{This example also shows that \finshift sets are not even closed by cartesian product with a finite set. }%
      Hence we have to restrict our closure operators.
      \begin{definition}
        \label{def:deviation}
        Let $\param$ be a variable.
        A \Shiftable $\shiftable$ is \emph\arithmetic w.r.t. $\param$ iff
        for every \linear $\expr$ occurring in $\shiftable$ and containing $\param$ 
        there is $\intk\in\Z$ s.t. $\expr=\param+\intk$
        (i.e. neither $\intk.\param$ nor $\param+\var$ are allowed,
        where $\intk\in\Z$, $\intk\neq0$ and $\var\in\integervars$).

        Assume that $\shiftable$ is \arithmetic w.r.t. $\param$.
        The \emph{deviation} of $\shiftable$ w.r.t. $\param$, written $\deviation(\shiftable)$,
        is defined as
        $\deviation(\shiftable)\eqdef\max\{\intk_1-\intk_2\mid 
        \intk_1,\intk_2\in\Z,
        \param+\intk_1, \param+\intk_2\text{ occur in }\shiftable
        \}$.
        $\deviation(\shiftable)\eqdef0$ if $\shiftable$ does not contain $\param$.
        Let $\numk\in\N$, we write $\bowl\numk$ \CourteLongue{$\eqdef$}{for the set} $\{\shiftable\mid\deviation(\shiftable)\leq\numk\}$.
      \end{definition}
      \begin{theorem}
        \label{thm:congruence}
        Let $\setofshiftables_1$ and $\setofshiftables_2$ be two sets of \Shiftable\s \arithmetic 
        w.r.t. a variable $\param$.
        If $\setofshiftables_1$ and $\setofshiftables_2$ are \finshift
        then, for any $\numk\in\N$,\Longue{ the set}
        $\setofshiftables_1\times\setofshiftables_2\cap \bowl\numk$\Longue{,
        written $\setofshiftables_1\alignedAnd_\numk\setofshiftables_2$,}
        is \finshift.
      \end{theorem}
      \Longue{
        One can notice that, in the counter-example given before Definition \ref{def:deviation},
        the deviations of schemata in $\setofshiftables_1\times\setofshiftables_2$ are unbounded.
        \begin{proof}
          We construct a bijective function $f:
          \stdsubset{\setofshiftables_1}\times\stdsubset{\setofshiftables_2}\times[-\numk..\numk]
          \to\stdsubset{\setofshiftables_1\alignedAnd_\numk\setofshiftables_2}$:
          as $\stdsubset{\setofshiftables_1}$, $\stdsubset{\setofshiftables_2}$ and $[-\numk..\numk]$ are finite,
          $\stdsubset{\setofshiftables_1\alignedAnd_\numk\setofshiftables_2}$ is finite,
          hence the result.
          Informally $f$ associates to each pair of \representative\s a \representative in $\setofshiftables_1\alignedAnd_\numk\setofshiftables_2$,
          however there are as many new \representative\s as there are possible deviations
          (actually twice as many), 
          hence the dependency on $[-\numk..\numk]$.
          Let $\repr{\shiftable_1}\in\stdsubset{\setofshiftables_1}$,
          $\repr{\shiftable_2}\in\stdsubset{\setofshiftables_2}$ and $d\in[-\numk..\numk]$,
          we now construct $f(\repr{\shiftable_1},\repr{\shiftable_2},d)$.

          First of all if $\repr{\shiftable_1}$ or $\repr{\shiftable_2}$ does not contain $\param$
          then $f(\repr{\shiftable_1},\repr{\shiftable_2},d)\eqdef(\repr{\shiftable_1},\repr{\shiftable_2})$
          independently of $d$.
          Then for any pair $(\shiftable_1,\shiftable_2)\in\setofshiftables_1\alignedAnd_\numk\setofshiftables_2$
          s.t. $\shiftable_1$ or $\shiftable_2$ does not contain $\param$,
          it is easily seen that
          $(\shiftable_1,\shiftable_2)\stdeqshift (\repr{\shiftable_1},\repr{\shiftable_2})$
          (we let the reader observe that this would not necessarily be the case 
          if both $\shiftable_1$ and $\shiftable_2$ contained $\param$).
          Furthermore $(\repr{\shiftable_1},\repr{\shiftable_2})$ is a \representative.
          Indeed suppose that there is another $(\shiftable_1',\shiftable_2')$
          s.t. $(\repr{\shiftable_1},\repr{\shiftable_2})\stdeqshift(\shiftable_1',\shiftable_2')$
          then necessarily $\repr{\shiftable_1}\stdeqshift\shiftable_1'$
          which contradicts the fact that $\repr{\shiftable_1}$ is a \representative w.r.t. $\setofshiftables_1$.

          So from now on we assume that both $\repr{\shiftable_1}$ and $\repr{\shiftable_2}$ contain $\param$.
          Hence every \Shiftable $\shiftable$ s.t. $\shiftable\stdeqshift\shiftable_1$ or $\shiftable\stdeqshift\shiftable_2$
          also contains $\param$.
          As a consequence, for every \Shiftable $\shiftable$,
          $\maxshift(\shiftable) \eqdef \max\{\intk\in\Z \mid \param+\intk\text{ occurs in }\shiftable\}$ is well-defined.
          
          We first prove that 
          $\{(\shiftable_1,\shiftable_2)
            \mid
            \shiftable_1\stdeqshift\repr{\shiftable_1},
            \shiftable_2\stdeqshift\repr{\shiftable_2},
            \maxshift(\shiftable_1)-\maxshift(\shiftable_2)=d
          \}$
          is a \chain.
          Thus we prove that for all $\shiftable_1,\shiftable_1'\in\setofshiftables_1$
          and $\shiftable_2,\shiftable_2'\in\setofshiftables_2$
          s.t.
          $\shiftable_1\stdeqshift\repr{\shiftable_1}$, 
          $\shiftable_1'\stdeqshift\repr{\shiftable_1}$,
          $\shiftable_2\stdeqshift\repr{\shiftable_2}$,
          $\shiftable_2'\stdeqshift\repr{\shiftable_2}$,
          $\maxshift(\shiftable_1)-\maxshift(\shiftable_2)=d$
          and $\maxshift(\shiftable_1')-\maxshift(\shiftable_2')=d$
          there is $\intk>0$ s.t.
          either $(\shiftable_1',\shiftable_2')\stdeqshift_\intk(\shiftable_1,\shiftable_2)$
          or $(\shiftable_1,\shiftable_2)\stdeqshift_\intk(\shiftable_1',\shiftable_2')$
          or $(\shiftable_1,\shiftable_2)=(\shiftable_1',\shiftable_2')$.
          By Proposition \ref{thm:triangle}
          we have
          either $\shiftable_1\stdeqshift_\intk\shiftable_1'$,
          $\shiftable_1'\stdeqshift_\intk\shiftable_1$ or $\shiftable_1=\shiftable_1'$
          for some $\intk>0$
          and either $\shiftable_2\stdeqshift_{\intk'}\shiftable_2'$,
          $\shiftable_2'\stdeqshift_{\intk'}\shiftable_2$ or $\shiftable_2=\shiftable_2'$
          for some $\intk'>0$.
          Suppose $\shiftable_1'\stdeqshift_\intk\shiftable_1$
          then $\maxshift(\shiftable_1)-\maxshift(\shiftable_1')=\intk$.
          From
          $\maxshift(\shiftable_1)-\maxshift(\shiftable_2)=d$
          and $\maxshift(\shiftable_1')-\maxshift(\shiftable_2')=d$
          it easily follows that $\maxshift(\shiftable_2)-\maxshift(\shiftable_2')=\intk$.
          As $\intk>0$, this entails that we cannot have
          $\shiftable_2\stdeqshift\shiftable_2'$ or $\shiftable_2=\shiftable_2'$.
          Hence the only possibility is $\shiftable_2'\stdeqshift\shiftable_2$,
          and more precisely $\shiftable_2'\stdeqshift_{\intk'}\shiftable_2$ with $\intk'=\intk$.
          As a consequence $(\shiftable_1',\shiftable_2')\stdeqshift_\intk(\shiftable_1,\shiftable_2)$.
          The case $\shiftable_1\stdeqshift_\intk\shiftable_1'$ is symmetric,
          and the case $\shiftable_1=\shiftable_1'$ easily entails $\shiftable_2=\shiftable_2'$ 
          by taking $\intk=0$ in the previous equations.
          
          Then we prove that it is a \wfchain.
          Notice that if $(\shiftable_1,\shiftable_2)\stdeqshift(\shiftable_1',\shiftable_2')$
          then $\shiftable_1\stdeqshift\shiftable_1'$.
          So if a \chain
          $\dotsb\stdeqshift
          (\shiftable_1^{\numi-1},\shiftable_2^{\numi-1})\stdeqshift
          (\shiftable_1^\numi,\shiftable_2^\numi)\stdeqshift
          (\shiftable_1^{\numi+1},\shiftable_2^{\numi+1})\stdeqshift
          \dotsb$
          does not contain a \representative then one of the \chain\s
          $\dotsb\stdeqshift
          \shiftable_1^{\numi-1}\stdeqshift
          \shiftable_1^\numi\stdeqshift
          \shiftable_1^{\numi+1}\stdeqshift
          \dotsb$
          or
          $\dotsb\stdeqshift
          \shiftable_2^{\numi-1}\stdeqshift
          \shiftable_2^\numi\stdeqshift
          \shiftable_2^{\numi+1}\stdeqshift
          \dotsb$
          does not contain a \representative, either.
          By hypothesis this is false in our case.
          As a consequence there is indeed a \representative for the \chain
          $\{(\shiftable_1,\shiftable_2)
            \mid
            \shiftable_1\stdeqshift\repr{\shiftable_1},
            \shiftable_2\stdeqshift\repr{\shiftable_2},
            \maxshift(\shiftable_1)-\maxshift(\shiftable_2)=d
          \}$,
          we set $f(\repr{\shiftable_1},\repr{\shiftable_2},d)$ to be this \representative.
          It is now trivial that
          $\stdsubset{(\setofshiftables_1\alignedAnd_\numk\setofshiftables_2)}=
          f[\stdsubset{\setofshiftables_1}\times\stdsubset{\setofshiftables_2}\times[-\numk..\numk]]$:
          for any pair $(\shiftable_1,\shiftable_2)\in\setofshiftables_1\alignedAnd_\numk\setofshiftables_2$,
          $(\shiftable_1,\shiftable_2)\stdeqshift 
          f(\repr{\shiftable_1},\repr{\shiftable_2},\maxshift(\shiftable_1)-\maxshift(\shiftable_2))$
          and $f(\repr{\shiftable_1},\repr{\shiftable_2},\maxshift(\shiftable_1)-\maxshift(\shiftable_2))$
          is a \representative w.r.t. $\setofshiftables_1\alignedAnd_\numk\setofshiftables_2$.
          Notice that $|\maxshift(\shiftable_1)-\maxshift(\shiftable_2)|\leq\numk$ because
          $(\shiftable_1,\shiftable_2)\in\bowl\numk$.
          \qed
        \end{proof}
      }%
      As trivial corollaries we get (where all \Longue{the involved }\Shiftable\s are \arithmetic w.r.t. $\param$):
      \begin{itemize}
        \item $\{\sch_1\Op\sch_2\mid\sch_1\in\setofschemata_1,\sch_2\in\setofschemata_2\}\cap\bowl\numk$, 
          where $\Op\in\{\And,\Or\}$,
          is \finshift when $\setofschemata_1$ and $\setofschemata_2$ are \finshift.
        \item $\{(\schemaAnd i\cstr{\schpat\sch})\And{\schcstr\sch}\mid\sch\in\setofschemata,
          \cstr\in\asetofconstraints\}\cap\bowl\numk$
          is \finshift when $\setofschemata$ and $\asetofconstraints$ are \finshift%
          \Longue{ }.
        \Longue{\item $\{\expr_1=\expr_2\mid\expr_1\in\setofexprs_1,\expr_2\in\setofexprs_2\}\cap\bowl\numk$
          is \finshift when $\setofexprs_1$ and $\setofexprs_2$ are sets of linear expressions, \finshift
          (this corollary will be useful in the proof of Lemma \ref{thm:cstr_finite_up_to_shift},
          which explains why equality up to a shift is defined on \Shiftable\s and not only on schemata).}%
      \end{itemize}

    \subsection{Refinement Extensions}
      \label{sec:extensions}
      \CourteLongue{W}{Equality up to a shift is generally not powerful enough to detect cycles, so w}%
      e now define simple extensions that allow better detection.
      Consider for example the schema $\sch$ defined in Example \ref{ex:semantics}.
      Using \SchDP there is a branch which contains:
      $\sch'\eqdef
      \prop_1\And\schematicAnd\var1{\param-1}(\prop_\var\Implies\prop_{\var+1})\And\neg\prop_\param\And\neg\prop_{\param+1}
      \And\param\geq0\And\param-1\geq0$.
      $\sch'$ loops on $\sch$
      but $\sch'$ is not equal to $\sch$ up to a shift.
      However $\neg\prop_{\param+1}$ is pure in $\sch'$
      (i.e. $\prop_{\param+1}\not\litIn\sch'$)
      so $\neg\prop_{\param+1}$ may be evaluated to $\truevalue$.
      Therefore we obtain $\prop_1\And\schematicAnd\var1{\param-1}(\prop_\var\Implies\prop_{\var+1})\And\neg\prop_\param\And\param\geq0\And\param-1\geq0$,
      i.e. $\sch\substitution{\param-1/\param}\And\param\geq0$.
      But $\param-1\geq0$ entails $\param\geq0$ so we can remove $\param\geq0$ and finally get $\sch\substitution{\param-1/\param}$.
      \Longue\par%
      We now generalise this example, thereby introducing two new looping refinements:
      the \emph{pure literal} extension and the \emph{redundant constraint} extension%
      \Longue{ (both of them actually take an existing looping refinement and extend it into a more powerful one,
      hence the name ``extension'')}.
      \Longue{%
      We could have defined them as rules rather than looping refinements,
      however this way the results can be useful not only to \SchDP but also to any other system working with iterated schemata
      e.g. they are applicable without any modification to the system \SchCal defined in \cite{tab09}.
      \subsubsection{Pure Literals.}}%
      \Courte\par%
        As usual a literal $\lit$ is \emph{(propositionally) pure} in a formula $\aformula$ iff 
        its complement does not occur positively in $\aformula$.
        The pure literal rule is standard in propositional theorem proving:
        it consists in evaluating a literal $\lit$ to $\truevalue$ in a formula $\aformula$
        if $\lit$ is pure in $\aformula$.
        It is well-known that this operation preserves satisfiability%
        \Longue{ but it is now often omitted as looking for occurrences of a literal generally costs more than the benefits of its removal.
        In our case, however, dropping this optimization frequently results in non termination}.
        \Longue\par
        The notion of pure literal has to be adapted to schemata.
        The conditions on $\lit$ must be strengthened in order to take iterations into account.
        For instance, if $\lit=\prop_\param$ and $\sch=\schematicOr\var1{2\param} \neg \prop_\var$ 
        then $\lit$ is not pure in $\sch$
        since $\neg\prop_\var$ is the complement of $\lit$ for $\var=\param$ (and $1\leq\param\leq 2\param$).
        \CourteLongue{However}{On the other hand} $\prop_{2\param+1}$ is pure in $\sch$ (since $2\param+1\not\in[1..2\param]$).
        \Longue{It is actually easy to see that $\maybelong$ is the right tool to formalize this notion.}

        \begin{definition}
          \label{def:pure_literal}
          A literal $\lit$ is \emph{pure} in a schema $\sch$ 
          iff for every environment $\anenv$ of $\sch$,
          $\envapp\lit\anenv$ is propositionally pure in $\envapp\sch\anenv$.
        \end{definition}
        It is easily seen that $\lit$ is pure in $\sch$ iff $\lit^c\neverbelong\sch$,
        thus by decidability of $\maybelong$, it is decidable to determine if a literal is pure or not.

      The substitution of an indexed proposition $\prop_{\expr_1,\dotsc,\expr_\numk}$ by a pattern $\pat'$
      in a pattern $\pat$, written $\pat\substitution{\pat'/\prop_{\expr_1,\dotsc,\expr_\numk}}$,
      is defined as follows:
      \CourteLongue{
        $\prop_{\expr_1,\dotsc,\expr_\numk}\substitution{\pat'/\prop_{\expr_1,\dotsc,\expr_\numk}}\eqdef\pat'$;
        $\propq_{\exprf_1,\dotsc,\exprf_\numk}\substitution{\pat'/\prop_{\expr_1,\dotsc,\expr_\numk}}\eqdef\propq_{\exprf_1,\dotsc,\exprf_\numk}$  
        if $\prop\neq\propq$ or $\exprf_\numi\neq\expr_\numi$ for some $\numi\in[1..\numk]$;
        $(\pat_1\Op\pat_2)\substitution{\pat'/\prop_{\expr_1,\dotsc,\expr_\numk}} \eqdef
        \pat_1\substitution{\pat'/\prop_{\expr_1,\dotsc,\expr_\numk}}\Op\pat_2\substitution{\pat'/\prop_{\expr_1,\dotsc,\expr_\numk}}$;
        $(\schemaOp\var\cstr\pat)\substitution{\pat'/\prop_{\expr_1,\dotsc,\expr_\numk}}\\\eqdef
        \schemaOp\var\cstr{\pat\substitution{\pat'/\prop_{\expr_1,\dotsc,\expr_\numk}}}$.
      }
      {\[\begin{aligned}
        \prop_{\expr_1,\dotsc,\expr_\numk}\substitution{\pat'/\prop_{\expr_1,\dotsc,\expr_\numk}}  &\eqdef \pat'\\
        \propq_{\exprf_1,\dotsc,\exprf_\numk}\substitution{\pat'/\prop_{\expr_1,\dotsc,\expr_\numk}} &\eqdef\propq_{\exprf_1,\dotsc,\exprf_\numk}
        \hspace*{1.2cm}\text{ if }\prop\neq\propq\text{ or }\exprf_\numi\neq\expr_\numi\text{ for some }\numi\in[1..\numk]\\
        (\pat_1\Op\pat_2)\substitution{\pat'/\prop_{\expr_1,\dotsc,\expr_\numk}} &\eqdef
        \pat_1\substitution{\pat'/\prop_{\expr_1,\dotsc,\expr_\numk}}\Op\pat_2\substitution{\pat'/\prop_{\expr_1,\dotsc,\expr_\numk}}
        \hspace*{1.55cm}(\Op\in\{\Or,\And\})\\
        (\schemaOp\var\cstr\pat)\substitution{\pat'/\prop_{\expr_1,\dotsc,\expr_\numk}} &\eqdef
        \schemaOp\var\cstr{\pat\substitution{\pat'/\prop_{\expr_1,\dotsc,\expr_\numk}}}
        \hspace*{3.55cm}(\bigOp\in\{\bigvee,\bigwedge\})\\
      \end{aligned}
      \]}%
      \Longue{Notice that this is a trivial syntactic substitution, 
      e.g. $(\neg\prop_1\And\schematicOr\var1\param\prop_\var)\substitution{\trueformula/\prop_1}
      =\neg\trueformula\And\schematicOr\var1\param\prop_\var$
      and not $\neg\trueformula\And(\trueformula\Or\schematicOr\var2\param\prop_\var)$.
      Actually the latter would be a mistake because we do not know whether $\param\geq1$ or not. }%
      \CourteLongue{For}{The definition naturally extends to} a schema $\sch$\CourteLongue{, we set}{ with}
      $\sch\substitution{\pat'/\prop_{\expr_1,\dotsc,\expr_\numk}} \eqdef \schpat\sch\substitution{\pat'/\prop_{\expr_1,\dotsc,\expr_\numk}}$.
      \CourteLongue{\par
        It is easy to show that for a literal $\lit$ which is pure
        in a schema $\sch$, if $\sch$ (resp. $\sch\substitution{\trueformula/\lit}$) has a model $\interp$
        then $\sch\substitution{\trueformula/\lit}$ (resp. $\sch$) has a model $\interpj$
        s.t. $\interpenv\interp(\param)=\interpenv\interpj(\param)$ for every parameter $\param$ of $\sch$.}
      {\begin{proposition}\label{thm:pure_lit}
        Let $\lit$ be a literal pure in a schema $\sch$. 
        If $\sch$ has a model $\interp$
        then $\sch\substitution{\trueformula/\lit}$ has a model $\interpj$
        s.t. $\interpenv\interp(\param)=\interpenv\interpj(\param)$ for every parameter $\param$ of $\sch$.

        Conversely if $\sch\substitution{\trueformula/\lit}$ has a model $\interp$
        then $\sch$ has a model $\interpj$
        s.t. $\interpenv\interp(\param)=\interpenv\interpj(\param)$ for every parameter $\param$ of $\sch$.
      \end{proposition}
        \begin{proof}
          Let $\interp$ be a model of $\sch$.
          $\envapp\sch{\interpenv\interp}$ is thus satisfiable.
          As $\lit$ is pure in $\sch$, $\envapp\lit{\interpenv\interp}$ is pure in $\envapp\sch{\interpenv\interp}$
          (and thus in $\envapp{\sch\substitution{\trueformula/\lit}}{\interpenv\interp}$).
          So by the classical result that the satisfiability of a propositional formula is preserved when removing a pure literal,
          $\envapp\sch{\interpenv\interp}\substitution{\trueformula/\envapp\lit{\interpenv\interp}}$ is satisfiable.
          As $\envapp\sch{\interpenv\interp}\substitution{\trueformula/\envapp\lit{\interpenv\interp}}
          =\envapp{\sch\substitution{\trueformula/\lit}}{\interpenv\interp}\substitution{\trueformula/\envapp\lit{\interpenv\interp}}$,
          $\envapp{\sch\substitution{\trueformula/\lit}}{\interpenv\interp}$ is also satisfiable.
          We define $\interpprop\interpj$ as one of its models
          and $\interpenv\interpj$ as $\interpenv\interp$.
          $\interpj$ is obviously a model of $\sch\substitution{\trueformula/\lit}$
          and indeed $\interpenv\interp(\param)=\interpenv\interpj(\param)$ for every parameter $\param$ of $\sch$.
          The proof of the converse is symmetric.
          \qed
        \end{proof}}
      A schema $\sch$ in which all pure literals have been substituted with $\trueformula$
      is written $\purified(\sch)$.
      \begin{definition}
        \label{def:pure_ext}
        Let $\refinement$ be a looping refinement.
        We call \emph{the pure extension} of $\refinement$ the relation $\refinement'$:
        $\sch_1\refinement'\sch_2\Leftrightarrow\purified(\sch_1)\refinement\purified(\sch_2)$.
      \end{definition}
      \CourteLongue
        {$\refinement'$ is easily proved to be a looping refinement.}
        {\begin{proposition}
          \label{thm:pure_ext_loops}
          The pure extension of a looping refinement is a looping refinement.
          \end{proposition}
          \begin{proof}
            Consider $\sch_1,\sch_2$ s.t. $\sch_1\refinement'\sch_2$,
            i.e. $\purified(\sch_1)\refinement\purified(\sch_2)$.
            Let $\interp$ be a model of $\sch_1$.
            By Proposition \ref{thm:pure_lit}, there exists a model $\interp'$ of $\purified(\sch_1)$ 
            s.t. $\interpenv{\interp'}(\param)=\interpenv\interp(\param)$ for every parameter $\param$ of $\sch_1$.
            Then, as $\refinement$ is a looping refinement and by Definition \ref{def:looping_schemata},
            there is a model $\interpj'$ of $\purified(\sch_2)$ 
            s.t. $\interpenv{\interpj'}(\param) < \interpenv{\interp'}(\param)$
            for some parameter $\param$ of $\purified(\sch_1)$ (and thus of $\sch_1$)
            and $\interpenv{\interpj'}(\param) \leq \interpenv{\interp'}(\param)$ for other parameters of $\sch_1$.
            Then by Proposition \ref{thm:pure_lit}, there exists a model $\interpj$ of $\sch_2$ 
            s.t. $\interpenv{\interpj'}(\param)=\interpenv\interpj(\param)$ for every parameter $\param$ of $\sch_2$.
            From a model $\interp$ of $\sch_1$, we constructed a model $\interpj$ of $\sch_2$
            s.t. $\interpenv\interpj(\param) < \interpenv\interp(\param)$ for some parameter $\param$ of $\sch_1$
            and $\interpenv\interpj(\param) \leq \interpenv\interp(\param)$ for other parameters,
            i.e. we proved that $\sch_1$ loops on $\sch_2$.
            \qed
          \end{proof}}

      \CourteLongue{Finally we describe the redundant constraint extension:}{\subsubsection{Redundant Constraints.}
        This extension is justified by the fact that \SchDP often leads to constraints of the form 
        $\param>0$, 
        then $\param>0\And\param-1>0$,
        then $\param>0\And\param-1>0\And\param-2>0$,
        etc.
        Such constraints contain redundant information,
        which can be an obstacle to the detection of cycles in a proof.}
        \begin{definition}
          %A constraint $\cstr$ is \emph{redundant} in a set of constraints $\asetofconstraints$ iff
          %$\asetofconstraints\setminus\{\cstr\}\models\cstr$.
          Any normal form of a schema $\sch$ by the following rewrite rules:
          \Courte{\vspace*{-0.2cm}}
          \[\begin{aligned}
            \cstr_1\And\dotsb\And\cstr_\numk&\to\cstr_1\And\dotsb\And\cstr_{\numk-1}
            &&\text{if }\{\cstr_1,\dotsc,\cstr_{\numk-1}\}\models\cstr_\numk\\
            \cstr&\to\falseconstraint
            &&\text{if }\cstr\text{ is unsatisfiable}\\
            \end{aligned}
          \Courte{\vspace*{-0.2cm}}
          \]
          is called a \emph\cstrirred schema of $\sch$.
        \end{definition}
        By decidability of satisfiability in linear arithmetic,
        it is easy to compute a \cstrirred schema of $\sch$.

        \begin{definition}
          \label{def:cstr_ext}
          Let $\refinement$ be a looping refinement.
          We call \emph{the \cstrirred extension} of $\refinement$ the relation $\refinement'$ s.t.
          for all $\sch_1,\sch_2$,
          $\sch_1\refinement'\sch_2$ iff
          there exists $\sch_1'$ (resp. $\sch_2'$) a \cstrirred schema of $\sch_1$ (resp. $\sch_2$) s.t.
          $\sch_1'\refinement\sch_2'$.
        \end{definition}
        \CourteLongue
          {Once again $\refinement'$ is easily proved to be a looping refinement.}
          {\begin{proposition}
            \label{thm:cstr_ext_loops}
            The \cstrirred extension of a looping refinement is a looping refinement.
          \end{proposition}
          \begin{proof}
            It is easy to show that if $\sch$ (resp. a \cstrirred of $\sch$) has a model $\interp$
            then any \cstrirred of $\sch$ (resp. $\sch$) has a model $\interpj$
            s.t. $\interpenv\interp(\param)=\interpenv\interpj(\param)$ for every parameter $\param$ of $\sch$.
            Then the proof goes exactly the same way as in the proof of Proposition \ref{thm:pure_ext_loops}.
            \qed
          \end{proof}}%

      \Longue{
        \subsubsection{Generalisation.}
          We can generalize Propositions \ref{thm:pure_ext_loops} and \ref{thm:cstr_ext_loops}:
          \begin{proposition}
            Let $\refinement$ be a looping refinement, and $\star$ a binary relation among schemata s.t. 
            for all schemata $\sch_1,\sch_2$ if $\sch_1\star\sch_2$
            then the satisfiability of $\sch_1$ is equivalent to the satisfiability of $\sch_2$,
            preserving the values of the parameters.
            The relation $\refinement'$ s.t.
            for all $\sch_1,\sch_2$,
            $\sch_1\refinement'\sch_2$ iff
            there exists $\sch_1',\sch_2'$ s.t. $\sch_1\star\sch_1'$, $\sch_2\star\sch_2'$ and
            $\sch_1'\refinement\sch_2'$,
            is a looping refinement.
          \end{proposition}
          And we can generalize Definitions \ref{def:pure_ext} and \ref{def:cstr_ext}:
          $\refinement'$ is called the \emph{$\star$-extension} of $\refinement$.

          Of course this construction has an interest only if $\refinement'$ catches more looping cases than $\refinement$.
          It can be seen as working with normal forms of schemata w.r.t. $\star$
          which can be better suited to $\refinement$ than their non-normal counterparts.
          From the two previous definitions and from the requirement that satisfiability ``fits well'' with $\star$,
          it can be observed that extensions would be seen in some other context as just \emph{optimisations}
          (see e.g. the pure literal rule, or the remark about normal forms).
          In the context of schemata, those are generally more than 
          just optimizations as they may be required for \emph{termination}.
          Interestingly enough circumscribing those extensions to the looping rule allows us to
          keep a high-level description of the main proof system and a modular presentation of looping.
        }%

  \CourteLongue
  {\section{A Decidable Class: \Elementary Schemata}}
  {\section{Decidable Classes}}
    \label{sec:complete}
    \Longue{We now present some classes of schemata for which \SchDP terminates.
    \subsection{\Elementary Schemata}}%
      \label{sec:elem_def}
      \begin{definition}[\Elementary Schema]
        \label{def:monadic} \label{def:limited_progression}\label{def:aligned}
        An iteration $\aniteration$ is \emph{framed} iff 
        there are two expressions $\expr_1,\expr_2$ s.t. $\cstr\Equiv\expr_1\leq\var\And\var\leq\expr_2$.
        $[\expr_1..\expr_2]$ is called the \emph{frame} of the iteration.
        \Longue\par%
        A schema $\sch$ is:
        \begin{itemize}
          \item \emph{Monadic} iff all indexed propositions occurring in $\sch$ have only one index.
          \item \emph{Framed} iff all iterations occurring in it are framed.
          \item \emph{Aligned on $[\expr_1..\expr_2]$} iff it is framed and all iterations have the same frame $[\expr_1..\expr_2]$.
          \item \emph{\Arithmetic{}} iff it is \arithmetic w.r.t. every variable occurring in it.
          %\item \emph{Of limited progression} iff
            %every index in $\sch$ containing a bound variable $\var$ is of the form $\var+\intk$ where $\intk\in\Z$.
          \item \emph\Elementary iff
            it has a unique parameter $\param$, it is monadic,
            %\arithmetic w.r.t. $\param$, of limited progression
            \arithmetic
            and aligned on $[\intk..\param-\intl]$ for some $\intk,\intl\in\Z$.
        \end{itemize}
        \Longue{The definitions extend to a node $\node$ of a tableau $\tab$ by considering its schema $\nodesch\node\tab$.}%
      \end{definition}
      Notice that \elementary schemata allow the nesting of iterations.
      But they are too weak to express the binary multiplier presented in the Introduction
      (since only monadic propositions are considered).
      \CourteLongue{However $\schematicAnd\var1\param{\schematicOr\varj1\param{(\prop_\var\Implies\propq_\varj)}}
      \And\schematicAnd\var1\param{\neg\propq_\var}\And\schematicOr\var1\param{\prop_\var}$,
      for instance, is a \elementary schema.}
      {\begin{example}
        $\schematicAnd\var1\param{\schematicOr\varj1\param{(\prop_\var\Implies\propq_\varj)}}
        \And\schematicAnd\var1\param{\neg\propq_\var}\And\schematicOr\var1\param{\prop_\var}$
        is \elementary.
      \end{example}
      }

      We divide \constraintsplit into two disjoint rules:
      \emph{framed}-\constraintsplit (resp. \emph{non framed}-\constraintsplit)
      denotes \constraintsplit with the restriction that 
      $\cstrsplitlhm$ (following the notations of the rule) is framed (resp. not framed).
      We consider the following strategy $\strat$ for applying the extension rules
      on a \elementary schema:%
      \label{def:strat}
      \begin{enumerate}
        \item First only framed-\constraintsplit applies until irreducibility.
        \item Then all other rules except \instantiaterule apply until irreducibility
          with the restriction that \propsimpl rewrites $\prop_{\expr_1}$
          iff $\expr_1$ contains no variable other than the parameter of the schema%
          \Longue{ (notice that there is only one index because the schema is monadic)}.
        \item Finally only \instantiaterule applies until irreducibility,
          with the restriction that if the unfolded iteration is framed
          then $\IRupbnd$ 
          (in the definition of \instantiaterule)
          is the upper bound of the frame.
          We then go back to 1.
      \end{enumerate}
      For the \looping rule we use equality up to a shift with its pure and \cstrirred extensions%
      \Longue{ (it is trivial that the order in which the extensions are done does not matter)}.
      It is easy to prove that $\strat$ preserves completeness.

      \intervalise and \emptiness never apply when the input schema is \elementary.
      Indeed let $\aniteration$ be an iteration of the schema. 
      $\cstr$ cannot contain an expression of the form $\numk.\expr$, hence \intervalise cannot apply.
      No variable other than $\var$ or the parameter can be free in $\cstr$
      (due to the frame of the form $[\intk..\param-\intl]$),
      thus \emptiness cannot apply.
      However \propsimpl may introduce non framed iterations,
      but no variable other than $\var$ or the parameter can be free in $\cstr$
      because \propsimpl only applies if $\expr_1$ contains no variable other than the parameter of the schema.
      \Longue{All this shall become clear in the next section.}%

    \Longue{\subsection{Termination of \SchDP for \Elementary Schemata}}%
      \label{sec:elem_term}
      \Courte{
        \begin{theorem}
          \label{thm:terminate_regular}
          $\strat$ terminates on every \elementary schema.
        \end{theorem}
        \begin{proof}(Sketch)
      }%
      \CourteLongue{The proof}
      {The proof that $\strat$ terminates for \elementary schemata}
      goes by showing that the set 
      $\{\nodeschfull\node\tab\mid\node\text{ is a node of}\CourteLongue{$ $}{\text{ }}\tab\}$
      --- i.e. the set of schemata generated all along the procedure ---
      is \Longue{(roughly\footnote{This set will actually be restricted to \emph{alignment nodes},
      see Definition \ref{def:alignment_node}.}) }%
      finite up to the \cstrirred and pure extensions of equality up to a shift.
      As $\nodeschfull\node\tab=\schpat{\nodesch\node\tab}\And\schcstr{\nodesch\node\tab}\And\nodeinterpsch\node\tab$,
      this set is equal to
      $\{\schpat{\nodesch\node\tab}\And\schcstr{\nodesch\node\tab}\And\nodeinterpsch\node\tab\mid
      \node\text{ is a node of }\tab\}$.
      So the task can approximately be divided into four:
      prove that the set of patterns is finite up to a shift\Longue{ (Lemma \ref{lem:main})},
      prove that the set of constraints is finite up to a shift\Longue{ (Lemma \ref{thm:cstr_finite_up_to_shift})},
      prove that the set of partial interpretations is finite up to a shift
      \Longue{(Lemma \ref{thm:purity}, Corollary \ref{thm:lit_finite_up_to_shift}) }%
      and combine the three results thanks to Theorem \ref{thm:congruence}%
      \Longue{ (Corollary \ref{thm:combine})}.

      \Longue{\subsubsection{Tracing \SchDP.}}
      Among those tasks, the hardest is the first one,
      because it requires an induction on the structure of $\schpat{\nodesch\node\tab}$.
      For this induction to be achieved properly we need to ``trace'' 
      the evolution under $\strat$ of every subpattern of $\schpat{\nodesch\node\tab}$.
      A subpattern can be uniquely identified by its position.
      So we extend \SchDP into \TSchDP (for \Emph Traced \SchDP),
      by adding to the pair $(\nodesch\node\tab,\nodeinterp\node\tab)$ labelling nodes in \SchDP
      a third component containing a set of positions of $\schpat{\nodesch\node\tab}$.
      Along the execution of the procedure,
      this subpattern may be moved, duplicated, deleted, some context may be added around it, some of its subpatterns may be modified.
      Despite all those modifications, we are able to follow the subpattern thanks to the set of positions in the labels.

      As usual, a position is a finite sequence of natural numbers,
      $\emptyPos$ denotes the empty sequence,
      $\seq_1.\seq_2$ denotes the concatenation of $\seq_1$ and $\seq_2$ and $\posOrder$ denotes the prefix ordering. 
      The \emph{positions} of a pattern $\pat$ are defined as follows: 
      $\emptyPos$ is a position in $\pat$;
      if $\pos$ is a position in $\pat$ then $1.\pos$ is a position in 
      $\neg\pat$, $\schemaAnd\var\cstr\pat$ and $\schemaOr\var\cstr\pat$;
      let $\numi\in\{1,2\}$, 
      if $\pos$ is a position in $\pat_\numi$ then $\numi.\pos$ is a position in 
      $\pat_1\Or\pat_2$ and $\pat_1\And\pat_2$.
      %if $\pos$ is a position in $\pat_2$ then $2.\pos$ is a position in 
      %$\pat_1\Or\pat_2$ and $\pat_1\And\pat_2$.

        For two sequences $\seq_1,\seq_2$ s.t. 
        $\seq_2$ is a prefix of $\seq_1$, $\relpos{\seq_2}{\seq_1}$ is the sequence 
        s.t. $\seq_2.(\relpos{\seq_2}{\seq_1})=\seq_1$.
        In particular for two positions $\pos_1,\pos_2$ s.t. $\pos_2$ is a prefix of $\pos_1$,
        $\relpos{\pos_2}{\pos_1}$ can be seen as the position \emph{relatively to $\pos_2$}
        of the subterm in position $\pos_1$\Longue{ in $\sch$}.
        \begin{definition}[\TSchDP]
          A \emph{\TSchDP tableau} $\tab$ is the same as a \SchDP tableau except that a node $\node$
          is labeled with a triple 
          $(\nodesch\node\tab,\nodeinterp\node\tab,\nodepos\node\tab)$
          where $\nodepos\node\tab$ is a set of positions in $\schpat{\nodesch\node\tab}$.
          \TSchDP keeps the behavior of \SchDP for $\nodesch\node\tab$ and $\nodeinterp\node\tab$,
          we only describe the additional behavior for $\nodepos\node\tab$ as follows:
          $\pos\to\pos_1,\dotsc,\pos_k$ means that $\pos$ is deleted and $\pos_1,\dotsc,\pos_k$ are added to $\nodepos\node\tab$.
          \begin{itemize}
            \item Splitting rules and the \propsimpl rewrite rule leave $\nodepos\node\tab$ as is.
            \item Rewrite rules. 
              We write $\aposq$ for the position of the subpattern of $\schpat\sch$ which is rewritten.
              We omit \emptiness as it never applies.
              \begin{itemize}
                \item \algebraic. For $\pos>\aposq$:
                  \Courte{\vspace*{-0.2cm}}%
                      \[\begin{aligned}
                        \pos&\to\aposq.(\relpos1{(\relpos\aposq\pos)}) &&
                        \parbox{7.8cm}{for rules where $\pat$ occurs on both sides of the rewrite
                        (following the notations of Definition \ref{def:rules}),
                        and if $\pos$ is the position of a subpattern of $\pat$}\\
                        \pos&\to\emptyset && \text{otherwise}\\
                      \end{aligned}\]
                  \Courte{\vspace*{-0.2cm}}
                \item \instantiaterule. 
                  \CourteLongue
                  {For $\pos>\aposq:\quad \pos\to\ \aposq.1.(\relpos\aposq\pos)\sep \aposq.2.1.(\relpos\aposq\pos)$}
                  {\[\pour \pos>\aposq:\quad \pos\to\ \aposq.1.(\relpos\aposq\pos)\sep \aposq.2.1.(\relpos\aposq\pos)\]}
                %\item \emptiness:
                      %\[\begin{aligned}
                        %\pour\aposq<\pos<\aposq': && \pos & \to \aposq.1.1.(\relpos\aposq\pos)\sep\aposq.2.1.(\relpos\aposq\pos)\\
                        %\pour\aposq'<\pos: && \pos & \to \aposq.1.1.(\relpos\aposq{\aposq'}).(\relpos{\aposq'}\pos)\\
                      %\end{aligned}\]
                      %where $\aposq'$ is the position of $\emptinesslhm$ in $\schpat\sch$
                      %(following the notations in \emptiness's definition).
                %\item \propsimpl: $\asetofpos$ is unchanged.
              \end{itemize}
          \end{itemize}
          \Longue{%
          Let $\node,\nodeb$ be nodes of a \TSchDP-tableau $\tab$ s.t. $\nodeb\childparent\node$.
          For two patterns $\pat_1,\pat_2$, we write $\pat_1\trace^\nodeb_\tab\pat_2$ iff
          $\pat_1=\schpat{\nodesch\node\tab}|_{\pos_1}$
          and $\pat_2=\schpat{\nodesch\nodeb\tab}|_{\pos_2}$ 
          for some positions $\pos_1\in\nodepos\node\tab$ and $\pos_2\in\nodepos\nodeb\tab$.
          }
        \end{definition}
        \Courte{From now on, $\tab$ is a \TSchDP tableau whose root schema is \elementary of parameter $\param$
        and of alignment $[\intk..\param-\intl]$ for some $\intk,\intl\in\Z$.}%
        \Longue{
          Notice that $\strat$ is naturally extended to \TSchDP tableaux.

          The \emph{stripped} of a \TSchDP tableau
          is the tree obtained by removing the last component (i.e. the set of positions) of each of its nodes' label.
          The following proposition is trivial:
          \begin{proposition}
            \label{prop:schdptschdp}
            (i) If $\tab$ is a \TSchDP tableau then its stripped is a \SchDP tableau.
            (ii) Conversely if $\tab$ is a \SchDP tableau of root $(\sch,\emptyset,\trueconstraint)$,
            and $\pos$ is a position in $\schpat\sch$,
            then there is a unique \TSchDP tableau $\tab_\pos$ of root $(\sch,\emptyset,\trueconstraint,\{\pos\})$,
            s.t. the stripped of $\tab_\pos$ is equal to $\tab$.
          \end{proposition}
          $\tab_\pos$ is called the \emph{decorated} of $\tab$ w.r.t. $\pos$.
        }
  
      \Longue{\subsubsection{Alignment Nodes.}}
        The set $\{\nodeschfull\node\tab\mid\node\text{ is a node of }\tab\}$
        is actually \emph{not} finite up to a shift.
        We have to restrict ourselves to a particular kind of nodes, called \emph{alignment nodes}.
        \CourteLongue{Then}{Eventually,}
        $\{\nodeschfull\node\tab\mid\node\text{ is an \emph{alignment} node of }\tab\}$
        will indeed be finite up to a shift.
        \Longue{\par From now on, $\tab$ is a \TSchDP tableau whose root schema is \elementary of parameter $\param$
        and of alignment $[\intk..\param-\intl]$ for some $\intk,\intl\in\Z$.}%
        \Longue{\begin{definition}[Alignment Node]}%
          \label{def:alignment_node}
          A node of $\tab$ is an \emph{alignment node} iff it is irreducible by step \alignStepNb of $\strat$\Longue{ (see page \pageref{def:strat})}.
          %An alignment node $\node_\numj$ is called a \emph{$\numj$-alignment node} iff
          %there is a sequence of alignment nodes $\node_0\prec\dotsc\prec\node_\numj$
          %s.t. there is no alignment node between the root of the tableau and $\node_0$
          %and between $\node_\numi$ and $\node_{\numi+1}$ ($1\leq\numi\leq\numj-1$).
          %\Longue{For a $\numj$-alignment node $\node$, $\numj$ is called the \emph{alignment rank} of $\node$.}%
        \Longue{\end{definition}}%
        \CourteLongue{%
          It is easy to check that every alignment node is framed and aligned on $[\numk..\param-\numl-\numj]$ for some $\numj>0$. 
          Thus every alignment node is \elementary.
          Furthermore, by irreducility w.r.t. the \propsplit and \propsimpl rules,
          the parameter $\param$ only occurs in the domain of the iteration 
          (otherwise the corresponding literal would be added into $\literalset$).
        }
        {
          \begin{proposition}
            \label{thm:step1}
            Let $\node,\nodeb$ be nodes of $\tab$ 
            s.t. $\nodeb$ is obtained by applying step 3 on $\node$.
            (i) Every iteration that occurs in $\nodesch\nodeb\tab$ occurs in $\nodesch\node\tab$.
            (ii) Furthermore if $\node$ is aligned on $[\expr_1..\expr_2]$ for some expressions $\expr_1,\expr_2$,
            then either $\schcstr{\nodesch\nodeb\tab}\models\expr_1>\expr_2$
            or $\schcstr{\nodesch\nodeb\tab}\models\expr_1\leq\expr_2$.
          \end{proposition}
          \begin{proof}
            (i) is trivial as only \constraintsplit can apply.
            It applies only if $\schcstr{\nodesch\node\tab}\And\forall\var\neg\cstr$
            is satisfiable (following the notations of the rule).
            If it is not the case then we have immediately $\schcstr{\nodesch\node\tab}\models\expr_1\leq\expr_2$, hence (ii).
            Otherwise \constraintsplit can apply and (ii) is obvious.
            \qed
          \end{proof}
          \begin{proposition}
            \label{thm:step2}
            Let $\node,\nodeb$ be nodes of $\tab$ 
            s.t. $\nodeb$ is obtained by applying step 2 on $\node$.
            If an iteration $\aniteration$ occurs in $\nodesch\nodeb\tab$
            then there is $\pat'$ s.t. $\schemaOp\var\cstr\pat'$ occurs in $\nodesch\node\tab$.
          \end{proposition}
          \begin{proof}
            Either $\aniteration$ comes from the rewrite of $\pat$ into $\pat'$ by rules of step 2
            (in which case the result is obvious),
            or it is new and has been introduced by the rules.
            We show that the latter case is actually impossible.
            By observing the conclusion of each rule that can apply in step 2,
            only \propsimpl can introduce new iterations
            (as \emptiness and \intervalise cannot apply),
            so suppose that $\aniteration$ was introduced by \propsimpl.
            By definition of $\strat$, $\cstr$ must have the form:
            $\delta.\param+\intk_1\neq \param+\intk_2\And \var=0$ 
            where $\delta\in\{0,1\}$, $\intk_1,\intk_2\in\N$ (and $\param$ is the only parameter of the schema).
            But then (non framed) \constraintsplit must have applied on $\aniteration$
            (it can indeed apply because if the condition of application was not fulfilled,
            then the domain of the iteration would be valid,
            and \algebraic would have removed it).
            $\aniteration$ is removed in the right branch of \constraintsplit,
            so we focus on the left branch:
            due to the added constraint,
            $\context(\sch_1)\Implies\exists\var\cstr$ (following the notations of \algebraic)
            is valid.
            Furthermore, as $\var$ was a fresh variable when \propsimpl applied, $\pat$ does not contain $\var$.
            Thus \algebraic must have applied and removed the iteration.
            Consequently $\aniteration$ cannot have been introduced by \propsimpl.
            \qed
          \end{proof}
          \begin{proposition}
            \label{thm:step3}
            Let $\node,\nodeb$ be nodes of $\tab$ 
            s.t. $\nodeb$ is obtained by applying step 3 on $\node$.
            If $\node$ is aligned on $[\expr_1..\expr_2]$,
            and $\schcstr{\nodesch\node\tab}\models\expr_1\leq\expr_2$,
            then $\nodeb$ is aligned on $[\expr_1..\expr_2-\intq]$, for some $\intq>0$.
          \end{proposition}
          \begin{proof}
            As $\schcstr{\nodesch\node\tab}\models\expr_1\leq\expr_2$, \instantiaterule can apply,
            and thus turn all the frames into $[\expr_1..\expr_2-1]$.
            Notice that we may also have $\schcstr{\nodesch\node\tab}\models\expr_1\leq\expr_2-\intq$ for some $\intq>0$,
            in which case \instantiaterule can apply $\intq$ times more per iteration.
            \qed
          \end{proof}

          \begin{lemma}
            \label{cor:aligned}
            An alignment node $\node$ of $\tab$ is aligned on $[\intk..\param-\intl-\intj]$ for some $\intj\in\N$.
            Furthermore if an alignment node $\nodeb\childparent\node$
            is aligned on $[\intk..\param-\intl-\intj']$ for some $\intj'\in\N$,
            then $\intj'>\intj$.
          \end{lemma}
          \begin{proof}
            The result is proved by induction on the number of alignment nodes above $\node$.
            The base case follows from the fact that the root of $\tab$ is \elementary
            and thus aligned on $[\intk..\param-\intl]$.
            By Propositions \ref{thm:step1} (i) and \ref{thm:step2}
            applying step 1 and then step 2 preserves the alignment.
            Let $\node'$ be an alignment node s.t. $\node\childparent\node'$,
            and there is no alignment node between $\node$ and $\node'$.
            By induction $\node'$ is aligned on $[\intk..\param-\intl-\intj]$ for some $\intj\in\N$.
            Because $\node'$ is an alignment node, \constraintsplit must have applied between $\node'$ and $\node$.
            Thus we have either $\schcstr{\nodesch\node\tab}\models\intk>\param-\intl-\intj$
            or $\schcstr{\nodesch\node\tab}\models\intk\leq\param-\intl-\intj$,
            by Proposition \ref{thm:step1} (ii).
            In the first case there are no more iterations and every subsequent node is trivially aligned.
            In the second case, by Proposition \ref{thm:step3},
            every node after step 3 is aligned on $[\intk..\param-\intl-\intj']$ for some $\intj'>\intj$.
            Then, once again, by Propositions \ref{thm:step1} (i) and \ref{thm:step2},
            applying step 1 and step 2 preserves the alignment,
            so every next alignment node has the expected alignment.
            \qed
          \end{proof}
          When an alignment node $\node$ of $\tab$ is aligned on $[\intk..\param-\intl-\intj]$ for some $\intj\in\N$,
          we call $\node$ a \emph{$\intj$-alignment node}.
          \begin{corollary}
            \label{cor:aligned_node_elementary}
            Every alignment node of $\tab$ is \elementary.
          \end{corollary}
          \begin{proof}
            We have to check that no new parameter is introduced,
            that the schema is still monadic, still \arithmetic and still aligned on $[\intk..\param-\intl]$ for some $\intk,\intl\in\Z$.
            The alignment is an obvious consequence of Lemma \ref{cor:aligned}.
            The ``monadicity'' is trivially preserved.
            The only way a new parameter could be introduced is when a connective binding a variable is removed.
            But it is easily seen that each rule which removes such a connective
            also removes the pattern in which the variable is bound,
            so no bound variable can become free.
            Finally the schema remains \arithmetic because a new arithmetic expression
            can only be introduced in \SchDP via an instantiation in \instantiaterule
            (or \intervalise with $\intl.\expr_1$ and $\intk.\expr_2$, but it cannot apply).
            As a \elementary schema is translated w.r.t. every variable,
            every expression occurring in it is either an integer
            or has the form $\var+\intk$ where $\var$ is a variable and $\intk\in\Z$,
            Instantiating a variable in an integer of course does not change the integer.
            Instantiating $\var$ in $\var+\intk$ with an integer turns the expression into another integer.
            Instantiating $\var$ in $\var+\intk$ with another expression $\var'+\intk'$,
            turns the expression into $\var'+\intk'+\intk$, which preserves the form of the expression.
            Hence in all cases \arithmetic property of the schema is preserved.
            \qed
          \end{proof}
          \begin{lemma}
            \label{cor:aligned_node_uniform}
            Let $\tab$ be a tableau whose root schema is \elementary of parameter $\param$.
            For every alignment node $\node$ of $\tab$,
            $\param$ only occurs in the domains of iterations.
          \end{lemma}
          \begin{proof}
            We have to show that indices of all literals do not contain $\param$.
            Suppose that $\nodesch\node\tab$ contains a literal $\lit$ whose index contains $\param$.
            We first show that we have either $\lit\alwaysoccur\nodeinterpsch\node\tab\And\schcstr{\nodesch\node\tab}$
            or $\lit^c\alwaysoccur\nodeinterpsch\node\tab\And\schcstr{\nodesch\node\tab}$.
            Indeed suppose it is not the case.
            We show that \propsplit can apply,
            i.e. that $\lit\alwaysoccur{\nodesch\node\tab}$ or $\lit^c\alwaysoccur{\nodesch\node\tab}$ (1),
            and neither $\lit\maybelong\nodeinterpsch\node\tab\And\schcstr{\nodesch\node\tab}$
            nor $\lit^c\maybelong\nodeinterpsch\node\tab\And\schcstr{\nodesch\node\tab}$ (2):
            \begin{enumerate}
              \item Notice that this is not because $\lit$ occurs in $\nodesch\node\tab$
                that $\lit^c\alwaysoccur{\nodesch\node\tab}$ or $\lit\alwaysoccur{\nodesch\node\tab}$:
                indeed if $\lit$ occurs in an iteration, there can be an environment where this iteration is empty,
                so $\lit$ does not necessarily occur in the corresponding propositional realization.
                But as $\node$ is an alignment node,
                \constraintsplit has applied in step 1,
                adding the constraint that either all iterations were empty, or no iteration was empty.
                In the first case, no iteration remains 
                (because \algebraic must have applied in step 2)
                so $\lit$ necessarily occurs outside an iteration,
                and thus $\lit^c\alwaysoccur{\nodesch\node\tab}$ or $\lit\alwaysoccur{\nodesch\node\tab}$.
                In the second case, we know by Proposition \ref{thm:step2},
                that if the non-emptiness of iterations was true before step 2,
                then it is also true after step 2,
                i.e. at $\node$.
                So we have indeed $\lit^c\alwaysoccur{\nodesch\node\tab}$ or $\lit\alwaysoccur{\nodesch\node\tab}$,
                and \propsplit indeed applies.

              \item Suppose we have either $\lit\maybelong\nodeinterpsch\node\tab\And\schcstr{\nodesch\node\tab}$
                or $\lit^c\maybelong\nodeinterpsch\node\tab\And\schcstr{\nodesch\node\tab}$.
                As we supposed that neither $\lit\alwaysoccur\nodeinterpsch\node\tab\And\schcstr{\nodesch\node\tab}$
                nor $\lit^c\alwaysoccur\nodeinterpsch\node\tab\And\schcstr{\nodesch\node\tab}$,
                this means that there exists a literal $\lit'\in\nodeinterp\node\tab$ 
                satisfying the property \explicitref that
                it has the same propositional symbol as $\lit$,
                not the same index in general,
                but this index may be the same in some environments
                (e.g. $\lit=\prop_\param$ and $\nodeinterp\node\tab=\{\prop_1\}$).
                Then, as $\lit'\in\nodeinterp\node\tab$, \propsimpl has necessarily applied on $\lit$
                by stating the disequality of the indices of $\lit$ and $\lit'$.
                However it cannot be valid that those indices are the same,
                as this would entail
                $\lit\alwaysoccur\nodeinterpsch\node\tab\And\schcstr{\nodesch\node\tab}$
                or $\lit^c\alwaysoccur\nodeinterpsch\node\tab\And\schcstr{\nodesch\node\tab}$.
                So the disequality necessarily holds.
                This is easily seen that it is possible for one $\lit'$,
                but it is not possible \emph{for all} literals in $\nodeinterp\node\tab$ satisfying \explicitref.
                Indeed this would contradict the assumption that
                $\lit\maybelong\nodeinterpsch\node\tab\And\schcstr{\nodesch\node\tab}$
                or $\lit^c\maybelong\nodeinterpsch\node\tab\And\schcstr{\nodesch\node\tab}$.
                This can be done formally by an induction on the number of literals satisfying \explicitref.
                So if this was not possible then the iteration would have been turned into its neutral element by \algebraic,
                and so every occurrence of $\lit$ would have been removed.
                This contradicts the initial assumption on $\lit$.  
            \end{enumerate}
            So we suppose that \propsplit has applied.
            Now, by definition of $\strat$, every occurrence of $\lit$ found in $\nodesch\node\tab$
            satisfies the conditions for the application of \propsimpl in $\strat$
            (as the node is \arithmetic, an index cannot contain two distinct variables).
            As $\lit\alwaysoccur\nodeinterpsch\node\tab\And\schcstr{\nodesch\node\tab}$
            or $\lit^c\alwaysoccur\nodeinterpsch\node\tab\And\schcstr{\nodesch\node\tab}$,
            there are $\lit_1,\dots,\lit_\intq\in\nodeinterp\node\tab$ 
            of indices $\expr_1,\dots,\expr_\intq$
            s.t. all of them have the same propositional symbol as $\lit$, 
            and $\schcstr{\nodesch\node\tab}\Implies\bigvee_{\inti\in1..\intq}\expr=\expr_\inti$ is valid,
            where $\expr$ is the index of $\lit$.
            Thus \propsimpl must have applied on $\lit$ with all those literals,
            introducing iterations stating $\expr\neq\expr_1$, \dots, $\expr\neq\expr_\intq$.
            The outermost iteration has thus necessarily be removed by \algebraic
            and $\lit$ must also have been removed before we reach step 3.
            \qed
          \end{proof}
        }%
        \CourteLongue{%
          It can be shown (see \cite{rapport09} for details) that each of the steps 1-3 terminates. 
          Thus every branch $\abranch$ containing a node $\node$ is either finite 
          or contains an alignment node $\nodeb\childparent\node$,
          i.e. \emph{an alignment node is always reached}.
        }
        {\begin{lemma}
          \label{thm:each_step_terminates}
          Each of the steps 1, 2 and 3 terminates.
        \end{lemma}
          \begin{proof}
            \noindent
            \begin{itemize}
              \item Step 1: as already seen, framed-\constraintsplit applies at most once.
              \item Step 2: 
                \propsplit can add new literals to the set of literals of a node.
                However this is done finitely many times,
                as it is easily seen that there are finitely many literals $\lit$ s.t. $\lit\alwaysoccur\sch$ or $\lit^c\alwaysoccur\sch$.
                For each atom $\prop_{\expr_1}$ s.t. $\expr_1$ contains no variable other than the parameter of the schema,
                \propsimpl applies as many times as there are literals with proposition symbol $\prop$ in the set of literals.
                We just saw that this last number cannot grow infinitely,
                and the number of atoms in $\sch$ cannot increase because \instantiaterule is not allowed in step 2.
                Finally, non-framed \constraintsplit applies as many times as there are non-framed iterations which is precisely the number of times 
                where \propsimpl can apply.
              \item Step 3: only \instantiaterule can apply.
                This terminates because there are finitely many iterations in a schema,
                and because if $\expr_1,\expr_2$ are expressions, no constraint can entail $\expr_1\leq\expr_2-\intq$ for every $\intq\geq0$.
                Notice that if \constraintsplit could apply in the meantime it would not terminate because
                constraints could be modified and thus there could be infinitely many $\IRupbnd$ s.t.
                $\instantiateRuleCondition$ is valid
                (following the notations of \instantiaterule).
                \qed
            \end{itemize}
          \end{proof}
          \begin{corollary}
            \label{thm:alignment_reached}
            Let $\abranch$ be a branch of $\tab$
            containing a node $\node$ then either $\abranch$ is finite or it contains an alignment node $\nodeb\childparent\node$,
            i.e. an alignment node is always reached.
          \end{corollary}
        }

      \Longue{\subsubsection{Main Proof.}}
        \Courte{Once this preliminary work is done, we can tackle the four aforementioned tasks.
        Finiteness (up to a shift) of $\{\schpat{\nodesch\node\tab}\mid
        \node\text{ is an alignment node of }\tab\}$ is proved by induction:
        \TSchDP enables to reason by induction 
        and Theorem \ref{thm:congruence} enables to make use 
        of the inductive hypotheses to prove each inductive case.
        Finiteness of $\{\schcstr{\nodesch\node\tab}\mid \node\text{ is an alignment node of }\tab\}$
        is proved thanks to the \cstrirred extension of equality up to a shift.
        Finiteness of $\{\nodeinterpsch\node\tab\mid \node\text{ is an alignment node of }\tab\}$
        is proved thanks to the pure literal extension of equality up to a shift.
        Finally Theorem \ref{thm:congruence} enables to combine the three results.
        }
        \Longue{
        The unfolding rules of \SchDP may introduce infinitely many distinct literals, 
        e.g. from $\schematicAnd\var1\param\prop_\var$ we generate $\prop_\param, \prop_{\param-1},\ldots$. 
        In principle this obviously prevents termination, but the key point is that
        (as shown by Lemma \ref{thm:purity}) these literals will eventually become pure,
        which ensures that they will not be taken into account by the looping rule.
          \begin{definition}
            \label{def:maxit}
            Let $\sch$ be a \elementary schema.
            Let $A(\sch)$ be the set $\{\intq\in\Z\mid\intq\text{ is the index of a literal in }\sch\}$,
            we write $\minbase(\sch)$ for $\min(A(\sch))$ 
            and $\maxbase(\sch)$ for $\max(A(\sch))$.
            %$[\minbase(\sch)..\maxbase(\sch)]$ is called the \emph{base set} of $\sch$.

            Let $B(\sch)$ be the set $\{\intq\in\Z\mid\var+\intq\text{ is the index of a literal in an iteration of }\sch\}$
            (it is a subset of $\Z$, by limited progression).
            We write $\minit(\sch)=\min(B(\sch))$
            and $\maxit(\sch)=\max(B(\sch))$.
            %$[\minit(\sch)..\maxit(\sch)]$ is called the \emph{induction set} of $\sch$.
          \end{definition}
        \begin{proposition}
          \label{thm:unfolding_inserts_literals}
          Let $\node,\nodeb$ be nodes of $\tab$ s.t. $\nodeb\childparent\node$.
          Then every literal of $\nodesch\nodeb\tab$ whose index is an integer occurs in $\nodesch\node\tab$.
          Any literal occurring in any node and whose index is an integer, occurs in the root schema $\sch$ of $\tab$.
          Consequently its index belongs to $[\minbase(\sch)..\maxbase(\sch)]$.
        \end{proposition}
        \begin{proof}
          First it easily seen that if a literal occurs after application of any rule 
          \emph{other than} \instantiaterule,
          then it already occurred before the application of the rule.
          This is not the case with \instantiaterule which can introduce a new literal,
          due to the substitution in its conclusion.
          Due to the restriction of \instantiaterule in $\strat$,
          this substitution replaces a variable with the last rank of an iteration.
          Furthermore \instantiaterule only applies on alignment nodes.
          By Lemma \ref{cor:aligned}, it is known that such nodes are aligned
          and that the last rank of their iterations depends on the parameter.
          Hence every literal that is introduced by substituting a variable with this last rank,
          cannot have an integer as index.
          So if a literal whose index is an integer occurs after application of any rule 
          (including \instantiaterule),
          then it already occurred before the application of the rule.

          Finally by induction on the length of the derivation,
          it is obvious that any literal occurring in any node and whose index is an integer,
          occurs in the root schema of $\tab$.
          \qed
        \end{proof}
          For the sake of simplicity we assume that \propsplit only applies on 
          $\prop_{\expr_1,\dots,\expr_\intk}$
          if $\prop_{\expr_1,\dots,\expr_\intk}$ occurs in $\nodesch\node\tab$
          (notice that we can have $\prop_{\expr_1,\dots,\expr_\intk}\alwaysoccur\nodesch\node\tab$
          without $\prop_{\expr_1,\dots,\expr_\intk}$ occurring in $\nodesch\node\tab$,
          e.g. $\prop_1\alwaysoccur\schematicAnd\var1\param\prop_\var\And\param\geq1$).
          This simplifies much some technical details,
          and it can be proved that this is not restrictive.
          \begin{lemma}
            \label{thm:purity}
            Let $\sch$ be the root schema of $\tab$.
            There is $\numj_0\in\N$ s.t. for every $\intj$-alignment node $\node$ of $\tab$,
            if $\intj\geq\intj_0$ then
            every literal in $\nodeinterp\node\tab$ of index $\param+\intq$
            where $\intq<\minit(\sch)-\intl-\intj$ or $\intq>\maxit(\sch)-\intl-\intj$
            is pure in $\nodesch\node\tab$.
          \end{lemma}
          \begin{proof}
            Let $\lit\in\nodeinterp\node\tab$.
            $\lit$ is pure in $\nodesch\node\tab$
            iff $\lit^c\not\maybelong\nodesch\node\tab$,
            i.e. iff $\exists\param(\schcstr{\nodesch\node\tab}\And\occurFormula_{\lit^c}(\nodesch\node\tab))$
            (where $\occurFormula_{\lit^c}(\nodesch\node\tab)$ is defined just before Proposition \ref{thm:always_occur_decidable})
            does \emph{not} hold, by Proposition \ref{thm:always_occur_decidable}.
            It is easily seen that, in our case
            (for the sake of simplicity we assume $\lit=\prop_{\param+\intq}$,
            the case $\neg\prop_{\param+\intq}$ is similar):
            \[\begin{aligned}
              \occurFormula_\lit(\nodesch\node\tab)=
              &\bigvee\{\exists\var(\intk\leq\var\And\var\leq\param-\intl-\intj\And\param+\intq=\var+\intq')
              \mid\neg\prop_{\var+\intq'}\litIn\nodesch\node\tab\}\\
              &\Or\bigvee\{\param+\intq=\intq'
              \mid\neg\prop_{\intq'}\litIn\nodesch\node\tab\}\\
              &\Or\bigvee\{\param+\intq=\param+\intq'
              \mid\neg\prop_{\param+\intq'}\litIn\nodesch\node\tab\}
            \end{aligned}\]
            But as $\node$ is an alignment node,
            if there were literals $\neg\prop_{\intq'}\litIn\nodesch\node\tab$
            (resp. $\neg\prop_{\param+\intq'}\litIn\nodesch\node\tab$),
            then \propsimpl would have applied.
            Thus either such literals would have been eliminated
            or the corresponding constraint $\param+\intq=\intq'$
            (resp. $\param+\intq=\param+\intq'$)
            would not hold in $\schcstr{\nodesch\node\tab}$.
            So it only remains to prove that the following does not hold:
            \[\exists\param\left(\schcstr{\nodesch\node\tab}\And\bigvee\{\exists\var(\intk\leq\var\And\var\leq\param-\intl-\intj\And\param+\intq=\var+\intq')
              \mid\neg\prop_{\var+\intq'}\litIn\nodesch\node\tab\}\right)\]
            This amounts to:
            \[\exists\param\left(\schcstr{\nodesch\node\tab}\And\bigvee\{\intk+\intq'\leq\param+\intq\And\param+\intq\leq\param-\intl-\intj+\intq'
              \mid\neg\prop_{\var+\intq'}\litIn\nodesch\node\tab\}\right)\]
            For every $\intq'$ s.t. $\neg\prop_{\var+\intq'}\litIn\nodesch\node\tab$,
            we have $\intq'\leq\maxit(\sch)$,
            by definition of $\maxit(\sch)$.
            So if $\intq>\maxit(\sch)-\intl-\intj$,
            then the above formula does not hold and we get the result.

            Now if $\intq<\minit(\sch)-\intl-\intj$ then $\lit$ is not pure in general, 
            however we can find $\numj_0\in\N$ s.t. if $\numj\geq\numj_0$ then it is actually impossible to have $\intq<\minit(\sch)-\intl-\intj$.
            We show that literals s.t. $\intq<\minit(\sch)-\intl-\intj$ can only be literals of the root schema $\sch$,
            so once all of them are pure, no other literal s.t. $\intq<\minit(\sch)-\intl-\intj$ will be introduced.
            Therefore we take $\intj_0$ to be the minimal $\intj$ s.t. $\param+\intq>\param+\intl-\intj+\maxit(\sch)$.
            First notice that, as $\lit$ has been introduced in $\nodeinterp\nodeb\tab$ by \propsplit at some node $\nodeb$,
            and thanks to the restriction made on \propsplit just before the lemma,
            $\lit$ was occurring in $\nodesch\nodeb\tab$.
            Now either this literal was already occurring in the root schema
            or it has been introduced by an \instantiaterule.
            As $\node$ is aligned on $\intk..\param-\intl-\intj$,
            all literals that have been introduced so far by \instantiaterule
            have an index of the form $\param-\intl-\intj'+\intq'$ 
            where $\intj'<\intj$ and $\intq'\in B(\sch)$
            (see Definition \ref{def:maxit} for the definition of $B(\sch)$).
            As $\minit(\sch)=\min(B(\sch))$ and $\intq<\minit(\sch)-\intl-\intj$,
            $\lit$ cannot have been introduced by \instantiaterule.
            Thus $\lit$ is indeed a literal of the root schema.
            Hence we can take $\intj_0$ as above
            (informally, iterations will be unfolded until all literals of the root schema are pure,
            when this is done we have our $\intj_0$).
            \qed
          \end{proof}
          \begin{corollary}
            \label{thm:lit_finite_up_to_shift}
            Let $\sch$ be a \elementary schema of parameter $\param$ and $\tab$ a tableau of root schema $\sch$,
            then $\left\{\nodeinterpsch\node\tab\mid\node\text{ is an alignment node}\right\}$
            is finite up to the pure extension of equality up to a shift on $\param$.
          \end{corollary}
          \begin{proof}
            It amounts to prove that
            $\literalset:= \bigAnd \{ \lit\in\nodeinterp\node\tab \mid
            \node\text{ is an alignment node},$ $\lit\text{ is not pure in } 
            \nodesch\node\tab\And\nodeinterpsch\node\tab\}$
            is finite up to a shift on $\param$.
            For every proposition symbol $\prop$ and
            every $\intq\in[\minit(\nodesch\node\tab)..\maxit(\nodesch\node\tab)]$,
            we define the set 
            $C(\intq,\prop):=\{\prop_{\param-\intl-\intj+\intq}\in\nodeinterp\node\tab\mid
            \text{ $\node$ is a $\numj$-alignment node},
            \numj\geq \numj_0\}$.
            $D(\intq,\prop)$ denotes the same set with $\neg\prop_{\param-\intl-\intj+\intq}$.
            $E$ is the set of literals that occurred before a $\numj$-alignment node with $\numj\leq\numj_0$.
            Finally $F:=\{
            \prop_\intq\in\nodeinterp\node\tab\mid
            \intq\in\Z,
            \node\text{ is a $\numj$-alignment node},
            \numj\geq\numj_0\}$.
            It is clear that:
            \[\literalset=\bigcup_{\intq,\prop}C(\intq,\prop)\cup\bigcup_{\intq,\prop}D(\intq,\prop)\cup E\cup F\]
            $C(\intq,\prop)$ and $D(\intq,\prop)$ are clearly finite up to a shift on $\param$.
            As there are finitely many $\prop$ and $\intq$, so are the sets
            $\bigcup_{\intq,\prop}C(\intq,\prop)$
            and
            $\bigcup_{\intq,\prop}D(\intq,\prop)$.
            $E$ is finite.
            Finally $F$ is finite because all its elements are literals of the root schema $\sch$
            thanks to Proposition \ref{thm:unfolding_inserts_literals}.
            Consequently $\literalset$ is indeed finite up to a shift.
            \qed
          \end{proof}
          \begin{lemma}
            \label{thm:cstr_finite_up_to_shift}
            Let $\sch$ be a \elementary schema of parameter $\param$ and $\tab$ a tableau of root schema $\sch$,
            then $\left\{\schcstr{\nodesch\node\tab}\middle|\node\text{ is an alignment node}\right\}$
            is finite up to the \cstrirred extension of equality up to a shift on $\param$.
          \end{lemma}
          \begin{proof}
            As \intervalise never applies, the only rule that introduces constraints is \constraintsplit.
            For a framed-\constraintsplit, 
            the only constraints that may be introduced in an alignment node are 
            of the form $\forall\var\neg(\intk\leq\var\And\var\leq\param-\intl-\intj)$
            or $\exists\var(\intk\leq\var\And\var\leq\param-\intl-\intj)$,
            for some $\intj\in\N$.
            Non-framed \constraintsplit introduces only constraints 
            that come from the emptiness of an iteration added by \propsimpl.
            Thus those constraints have the form $\expr\star\exprf$ where $\star\in\{=,\neq\}$,
            $\expr$ comes from a literal in $\nodesch\node\tab$
            and $\exprf$ comes from a literal in $\nodeinterp\node\tab$.
            Thus if we are in a $\intj$-alignment node and $\expr$ contains $\param$ then
            $\expr$ belongs to the set $[\param-\intl-\intj+\minit(\sch)..\param-\intl-\intj+\maxit(\sch)]$
            by Lemma \ref{thm:purity};
            if $\expr$ does not contain $\param$
            then it belongs to the set $[\minbase(\sch)..\maxbase(\sch)]$
            by Proposition \ref{thm:unfolding_inserts_literals};
            and $\exprf$ belongs to the set $[\minbase(\sch)..\maxbase(\sch)]\cup[\param-\intl-\intj+\minit(\sch)..\param-\intl+\maxit(\sch)]$.

            We now prove that the set of added constraints is finite up to the \cstrirred extension of equality up to a shift.
            We distinguish various cases depending on the shape of the introduced constraints.
            Finally, we will combine those results thanks to Theorem \ref{thm:congruence}.
            \begin{itemize}
              \item Framed constraint $\exists\var(\intk\leq\var\And\var\leq\param-\intl-\intj)$:
                the set of generated constraints of this form is:
                \[\hspace*{-0.1cm}
                \begin{array}{l}
                  \exists\var(\intk\leq\var\And\var\leq\param-\intl)\\
                  \exists\var(\intk\leq\var\And\var\leq\param-\intl)
                  \And\exists\var(\intk\leq\var\And\var\leq\param-\intl-1)\\
                  \exists\var(\intk\leq\var\And\var\leq\param-\intl)
                  \And\exists\var(\intk\leq\var\And\var\leq\param-\intl-1)
                  \And\exists\var(\intk\leq\var\And\var\leq\param-\intl-2)\\
                  \text{etc.}\\
                \end{array}\]
                but we can remove the redundant constraints and obtain:
                \[\begin{array}{l}
                  \exists\var(\intk\leq\var\And\var\leq\param-\intl)\\
                  \exists\var(\intk\leq\var\And\var\leq\param-\intl-1)\\
                  \exists\var(\intk\leq\var\And\var\leq\param-\intl-2)\\
                  \text{etc.}\\
                \end{array}\]
                which is trivially \finshift.
              \item Framed constraint $\forall\var\neg(\intk\leq\var\And\var\leq\param-\intl-\intj)$:
                Once this constraint is added, there are no more iterations in the schema,
                so no other constraint of this form will be added.
                Thus the set of all constraints of this form that may be added in all the nodes is
                $\{\forall\var\neg(\intk\leq\var\And\var\leq\param-\intl-\intj) \mid \intj\in\N\}$
                which is obviously finite up to a shift.
              \item Non-framed constraint with $\expr\in[\param-\intl-\intj+\minit(\sch)..\param-\intl-\intj+\maxit(\sch)]$
                and $\exprf\in[\param-\intl-\intj+\minit(\sch)..\param-\intl+\maxit(\sch)]$:
                then $\expr\star\exprf$ is either valid or unsatisfiable.
                If it is valid then it is of course redundant so we do not even need to consider it.
                If it is unsatisfiable then, by \cstrirredy,
                we can consider that it is $\falseconstraint$.
                When an unsatisfiable constraint is added, the branch is closed, so no other constraint may be added.
                Thus the set of such constraints generated in this case is just $\{\falseconstraint\}$, trivially finite.
              \item Non-framed constraint with $\expr\in[\minbase(\sch)..\maxbase(\sch)]$
                and $\exprf\in[\param-\intl-\intj+\minit(\sch)..\param-\intl-\intj-\intk+\maxbase(\sch)]$,
                i.e. the considered set of constraints is:
                \[A\eqdef\left\{\expr\star\exprf\middle|
                \begin{array}{l}
                  \expr\in[\minbase(\sch)..\maxbase(\sch)]\\
                  \exprf\in[\param-\intl-\intj+\minit(\sch)..\param-\intl-\intj-\intk+\maxbase(\sch)]\\
                  \intj\in\N\\
                \end{array}\right\}\]
                It is a finite union of sets of the form $\{\param-\intj+\intq\mid\intj\in\N\}$ where $\intq\in\Z$.
                All such sets are \finshift, so $A$ is \finshift.
                Then $[\minbase(\sch)..\maxbase(\sch)]$ is obviously \finshift.
                so we get the result by the third corollary of Theorem \ref{thm:congruence}
                (with deviation $0$ as no expression in $[\minbase(\sch)..\maxbase(\sch)]$ contains $\param$).
                Notice that the full interval on which $\exprf$ ranges
                ($[\param-\intl-\intj+\minit(\sch)..\param-\intl+\maxit(\sch)]$)
                has been split on purpose, so that $A$ can indeed be a finite union of \finshift sets.
              \item Non-framed constraint with $\expr\in[\minbase(\sch)..\maxbase(\sch)]$
                and $\exprf\in[\param-\intl-\intj-\intk+\maxbase(\sch)+1..\param-\intl+\maxit(\sch)]$,
                when $\star$ is $=$:
                this constraint states $\expr=\exprf$.
                However we know that $\intk\leq\param-\intl-\intj$,
                so $\param\geq\intk+\intl+\intj$.
                As $\expr=\exprf$, we have $\param=\param+\expr-\exprf$.
                So $\param+\expr-\exprf\geq\intk+\intl+\intj$,
                thus $\exprf\leq\param+\expr-\intk-\intl-\intj$.
                As $\expr\leq\maxbase(\sch)$, we obtain $\exprf\leq\param+\maxbase(\sch)-\intk-\intl-\intj$.
                This contradicts the above lower bound,
                so $\expr=\exprf$ is actually unsatisfiable and we get the result as in the third case.
              \item Non-framed constraint with $\expr\in[\minbase(\sch)..\maxbase(\sch)]$
                and $\exprf\in[\param-\intl-\intj-\intk+\maxbase(\sch)+1..\param-\intl]$,
                when $\star$ is $\neq$:
                This is the hard case, indeed we can easily obtain a set which is not \finshift.
                even with the \cstrirred extension.
                For instance the infinite set:
                \[\begin{array}{l}
                0\neq\param-\intl\\
                0\neq\param-\intl\And0\neq\param-\intl-1\\
                0\neq\param-\intl\And0\neq\param-\intl-1\And0\neq\param-\intl-2\\
                \text{etc.}\\
                \end{array}\]
                is not \finshift and, contrarily to the previous cases,
                we cannot use the \cstrirred extension to simplify it.
                However at node $\node$, $\schcstr{\nodesch\node\tab}$ entails $\param-\intl-\intj\geq\intk$
                (because $\node$ is aligned on $[\intk..\param-\intl-\intj]$)
                and thus $\param-\intl-\intj-\intk\geq0$.
                On the other hand $\exprf\geq\param-\intl-\intj-\intk+\maxbase(\sch)+1$,
                thus $\exprf\geq\maxbase(\sch)+1$.
                So, as $\expr\leq\maxbase(\sch)$: $\exprf>\expr$.
                Hence the constraint $\exprf\neq\expr$ is finally redundant.
            \end{itemize}
            Finally it is easily seen that combining all different cases preserves finiteness up to a shift
            by Theorem \ref{thm:congruence}.
            Simply because by inspecting all the cases,
            one can see that all the expressions of a constraint inserted at a $\intj$-alignment node,
            are of the form $\param-\intj+\intq$ for some $\intq$ belonging to a finite set.
            So all the cases are ``synchronized''.
            \qed
          \end{proof}
        \begin{lemma}[Main Lemma]
          \label{lem:main}
          Let $\sch$ be a \elementary schema of parameter $\param$ and $\tab$ a tableau of root schema $\sch$,
          then $\left\{\schpat{\nodesch\node\tab}\middle|\node\text{ is an alignment node}\right\}$
          is finite up to a shift on $\param$.
        \end{lemma}
        \begin{proof}
          We prove that 
          $\left\{
            \pat
          \middle|
            \schpat{\nodesch\node\tab}|_\pos\trace^\node_{\tab_\pos}\pat,
            \node\text{ is an alignment node},
            \pat\text{ is a pattern}
          \right\}$
          \newline($\tab_\pos$ is the decorated of $\tab$ w.r.t. $\pos$)
          is finite up to a shift on $\param$
          for every position $\pos$ in $\schpat{\nodesch\node\tab}$. 
          We get the intended result when $\pos=\emptyPos$,
          indeed it is easily seen that this position is invariant by \TSchDP
          hence if $\pat$ is s.t. 
          $\schpat{\nodesch\node\tab}|_\emptyPos\trace^\node_{\tab_\emptyPos}\pat$
          then $\pat=\schpat{\nodesch\node\tab}$
          (as $\tab_\emptyPos$ is the decorated of $\tab$ w.r.t. position $\emptyPos$,
          $\node$ may indifferently be considered as a node of $\tab$ or a node of $\tab_\emptyPos$).

          Let $\schpat{\nodesch\node\tab}|_\pos$ be a subpattern of $\schpat{\nodesch\node\tab}$ at some position $\pos$
          and $\pat'$ a pattern s.t. $\schpat{\nodesch\node\tab}|_\pos\trace^\node_{\tab_\pos}\pat'$.
          $\pat'$ is the result of applying some transformations to some other $\pat$ s.t. 
          $\schpat{\nodesch\node\tab}|_\pos\trace^{\node'}_{\tab_\pos}\pat$.
          Those transformations may be a combination of:
          (i) identity (if no rule applied to the subpattern between two alignment nodes),
          (ii) rewrite of a pattern above $\pat$,
          (iii) rewrite of a subpattern of $\pat$,
          (iv) rewrite of $\pat$ itself,
          or (v) instantiation of a variable (in case \instantiaterule applies somewhere above $\pat$).
          We have to check that none of those transformations can generate an infinite set of new schemata.
          This is trivial for (i).
          (ii) is invisible when tracing $\schpat{\nodesch\node\tab}|_\pos$ 
          (as the trace follows the moves of $\schpat{\nodesch\node\tab}|_\pos$)
          and thus is an identity as far as we are concerned (notice that this is why tracing was designed for).
          For the other cases the proof goes by induction on the structure of $\pat$:
          \begin{itemize}
            \item Suppose $\pat$ is a literal of index $\expr$.
              \begin{description}
                \item[(iii)]
                  Impossible.
                \item[(iv)]
                  Only \propsimpl can rewrite a literal.
                  This is possible only if no variable other than $\param$ occurs in $\expr$,
                  in which case \algebraic we apply then.
                  As seen multiple times, the introduced iterated connective
                  will necessarily be deleted in the next alignment node
                  (either by removing the full iteration,
                  or by removing only the connective).
                  Hence no schema is generated.
                \item[(v)]
                  This is possible only if there \emph{is} a variable other than $\param$ in $\expr$ 
                  (as $\param$ is never instantiated)
                  in which case $\pat$ is turned into a literal whose index does not refer to a variable other than $\param$,
                  then the expansion and algebraic simplifications rules apply as in Case (iv).
              \end{description}
            \item Suppose $\pat=\pat_1\Op\pat_2$ where $\Op\in\{\And,\Or\}$.
              \begin{description}
                \item[(iii)]
                  It implies that there are $\pat'_1,\pat'_2$ s.t.
                  $\pat_1\trace^\node_{\tab_{1.\pos}}\pat'_1$
                  and
                  $\pat_2\trace^\node_{\tab_{2.\pos}}\pat'_2$.
                  By Lemmata \ref{cor:aligned} and \ref{cor:aligned_node_uniform}
                  \emph{all} expressions involving $\param$ in both $\pat_1$ and $\pat_2$ have the form $\param-\intl-\intj$
                  hence $\deviation(\pat_1,\pat_2)=0$ (where $\deviation$ denotes the deviation, Section \ref{sec:eq_shift}).
                  By induction the sets of possible $\pat_1'$ and $\pat_2'$ are finite up to a shift,
                  so we can apply the first corollary of Theorem \ref{thm:congruence} and conclude.
                \item[(iv)]
                  The only possible rule is \algebraic in which case the result is obtained by induction.
                \item[(v)]
                  For every substitution $\asubstitution$,
                  $\pat\sigma=\pat_1\sigma\Op\pat_2\sigma$,
                  so if $\pat\trace\pat\sigma$
                  then $\pat_1\trace\pat_1\sigma$ and $\pat_2\trace\pat_2\sigma$,
                  and we conclude by induction.
              \end{description}
            \item Suppose $\pat=\schematicOp\var\intk{\param-\intl-\intj}\pate$ where $\bigOp\in\{\bigAnd,\bigOr\}$, $\intj\in\N$.
                  By Lemma \ref{cor:aligned}, we know that every iteration must have this form.
                  \begin{description}
                    \item[(iii)]
                      This is handled as in the previous case
                      except that we use the second corollary of Theorem \ref{thm:congruence} 
                      instead of the first one.
                    \item[(iv)]
                      The only rewrite can be \instantiaterule.
                      For every $\intp\in\N$, when \instantiaterule applies $\intp$ times,
                      $\pat$ is turned into $\pate_1\Op\dots\Op\pate_\intp\Op\schematicOp\var\intk{\param-\intl-\intj-\intp}\pate$.
                      But $\pat\trace^{\alpha'}_\tab\pate_1$, \dots, $\pat\trace^{\node'}_\tab\pate_\intp$
                      so by induction hypothesis on $\pat$ they all belong to the same \finshift set.
                      So if $\intp$ is big enough,
                      there are patterns of the form $\pate_\intq$ that will loop on each other ($\intq\in1..\intp$).
                      Formally there is $\intq_0\in\N$ s.t.
                      for every $\intp\in\N$ and every $\intq\in1..\intp$, 
                      if $\intq>\intq_0$ then there is a $\intq'\leq\intq_0$ s.t. 
                      $\pate_\intq\stdeqshift\pate_{\intq'}$.
                      By Lemmata \ref{cor:aligned} and \ref{cor:aligned_node_uniform},
                      only iterations contain $\param$ and all of them are aligned,
                      thus there is actually no shift on $\param$ meaning that
                      $\pate_{\intq}=\pate_{\intq'}$.
                      Hence, by \algebraic, $\pate_1\Op\dotsb\Op\pate_\intp$
                      simplifies into $\pate_1\Op\dotsb\Op\pate_{\intq_0}$ at worst.
                      Finally all schemata obtained from $\pat$ are of the form
                      $\pate_1\Op\dots\Op\pate_{\intq_0}\Op\schematicOp\var\intk{\param-\intl-\intj-\intp}\pate$.
                      There are finitely many such schemata by induction hypothesis on $\pate$
                      (and thus on $\pate_1,\dots,\pate_{\intq_0}$),
                      by the first and second corollaries of Theorem \ref{thm:congruence} 
                      (the deviation is null),
                      and because $\intq_0$ is a constant.
                    \item[(v)]
                      As $\param-\intl-\intj$ does not contain other variables than $\param$ it is not affected by the instantiation.
                      All bound variables are assumed distinct so the instantiation cannot replace $\var$.
                      Thus, writing $\asubstitution$ for the substitution,
                      $\pat\asubstitution=\schematicOp\var\intk{\param-\intl-\intj}{(\pate\asubstitution)}$,
                      and we conclude by induction.
                      \qed
                  \end{description}
          \end{itemize}
        \end{proof}
        \begin{corollary}
          \label{thm:combine}
          Let $\sch$ be a \elementary schema of parameter $\param$ and $\tab$ a tableau of root schema $\sch$,
          then $\left\{\nodeschfull\node\tab\mid\node\text{ is an alignment node}\right\}$
          is finite up to the \cstrirred and pure extensions of equality up to a shift on $\param$.
        \end{corollary}
        \begin{proof}
          This follows from Definition \ref{def:looping_nodes}
          and from Theorem \ref{thm:congruence} applied to the results of Corollary \ref{thm:lit_finite_up_to_shift},
          Lemma \ref{thm:cstr_finite_up_to_shift} and Main Lemma.
          Lemma \ref{thm:purity} ensures that the deviation is lower than $\maxit(\sch)-\minit(\sch)$.
          \qed
        \end{proof}
        \begin{theorem}
          \label{thm:terminate_regular}
          $\strat$ terminates on every \elementary schema.
        \end{theorem}
        \begin{proof}
          It easily follows from the previous Corollary and the fact that $\strat$ uses the pure extension of equality up to a shift.
          Corollary \ref{thm:alignment_reached} is also required to ensure that it is indeed sufficient to restrict ourselves to alignment nodes.
          \qed
        \end{proof}
        }%
      \Courte{
          \par
          See \cite{rapport09} for the detailed proof.
          \qed
        \end{proof}
      }

    \Longue{
      \subsection{Extensions}
        In the light of the previous proof, we can easily extend the class of \elementary schemata to broader terminating classes.
        First we can relax a little the alignment condition:
        \begin{definition}
          A schema $\sch$ is:
          \begin{itemize}
            \item \emph{down-aligned} iff it is framed and 
                  the frames of all iterations have the same lower bound $\intk\in\Z$
                  and have an upper bound of the form $\param-\intl$, where $\intl\in\Z$.
            \item \emph{up-aligned} iff it is framed and 
                  the frames of all iterations have the same upper bound $\param-\intl$, where $\intl\in\Z$
                  and have any $\intk\in\Z$ as their lower bound.
            \item \emph{broadly aligned} iff 
                  all iterations of $\sch$ have frames of the form $[\intk_1..\param-\intk_2]$, $\intk_1,\intk_2\in\Z$.
          \end{itemize}
        \end{definition}
        \begin{theorem}
          \label{thm:terminate_down_regular}
          $\strat$ terminates on every schema which is monadic, of limited progression and down-aligned.
        \end{theorem}
        \begin{proof}(Sketch)
          Such a schema is almost \elementary except that down-alignment is substituted to alignment.
          It is easily seen that, after the first passing in step 2,
          either the constraint $\intk\leq\param-\intl$ or $\intk>\param-\intl$ has been added to the node,
          where $\intl=\min\{\intl'\mid\param-\intl'\text{ is the upper bound of an iteration in }\sch\}$.
          If it is $\intk\leq\param-\intl$ then it implies that $\intk\leq\param-\intl'$ for every $\intl'\geq\intl$.
          In step 3, all iterations are unfolded until no longer possible.
          Hence here, all iterations will be unfolded until their upper bound reaches $\param-\intl-1$ 
          (even those of frames $[\intk..\param-\intl']$, $\intl'>\intl$).
          As a consequence all iterations are now aligned and we are back in the same case as for \elementary schemata.
          We call this phase, where all iterations progressively become aligned, the \emph{rectification}.
          Rectification terminates because of a similar argument 
          to the one proving the termination of Step 3 in the proof of Lemma \ref{thm:each_step_terminates}.
          In the case where $\intk>\param-\intl$ has been added,
          it is easily seen that there will be finitely many unfoldings of iterations of frame $[\intk..\param-\intl']$, $\intl'>\intl$
          (actually there will be at most $\intm-\intl'$ such unfoldings per iteration, 
          where $\intm=\max\{\intl'\mid\param-\intl'\text{ is the upper bound of an iteration in }\sch\}$)
          then all iterations will be empty.
          \qed
        \end{proof}
        \begin{theorem}
          \label{thm:terminate_up_regular}
          $\strat$ terminates on every schema which is monadic, of limited progression and up-aligned.
        \end{theorem}
        \begin{proof}(Sketch)
          In this case schemata will, in general, never become aligned:
          suppose we have two iterations $\schematicOp\var{\intk_1}{\param-\intl}\pat$
          and $\schematicOpbis\varj{\intk_2}{\param-\intl}\pat'$ with $\intk_1<\intk_2$.
          Then any constraint $\intk_1\geq\param-\intl-\intj$ implies $\intk_2\geq\param-\intl-\intj-\intk_1+\intk_2$ so 
          when 
          $\schematicOp\var{\intk_1}{\param-\intl}\pat$
          will be unfolded until $\param-\intl-\intj$,
          $\schematicOpbis\varj{\intk_2}{\param-\intl}\pat'$
          will be unfolded until $\param-\intl-\intj-\intk_1+\intk_2$.
          We will never reach alignment.
          However it is easily seen that the difference between two upper bounds (here $\intk_2-\intk_1$)
          will always remain lower than the deviation of the original schema.
          Hence slight modifications in the proof of Main Lemma enable to conclude.
          The hard point lies in the application of \algebraic in the item $(iv)$ of the iteration case,
          indeed now we cannot conclude from 
          $\pat'_{\exprf-\intk+\intq}\stdeqshift\pat'_{\exprf-\intk+\intq'}$
          that 
          $\pat'_{\exprf-\intk+\intq}=\pat'_{\exprf-\intk+\intq'}$
          as there is no alignment.
          However as the ``mis-alignment'' is confined to a finite set,
          the sequence $(\pat'_{\exprf-\intk+1}\Op\dotsb\Op\pat'_\exprf)_{\intk\in\N}$ still cannot grow infinitely.
          \qed
        \end{proof}
        \begin{theorem}
          \label{thm:terminate_broadly_regular}
          $\strat$ terminates on every schema which is monadic, of limited progression and broadly-aligned.
        \end{theorem}
        \begin{proof}(Sketch)
          This proof is close to the previous one.
          Actually we do not really need the fact that the upper bound is the same in the previous proof.
          \qed
        \end{proof}
        \begin{definition}
          A schema $\sch$ is:
          \begin{itemize}
            \item \emph{variable-aligned} on $[\expr_1..\expr_2]$, for two \linears $\expr_1,\expr_2$ iff 
                  every iteration of $\sch$ is framed either on $[\expr_1..\expr_2]$, or on $[\expr_1..\var+\intq]$ 
                  where $\var$ is a non-parameter variable and $\intq\in\Z$.
            \item \emph{simply variable-aligned} iff it is variable-aligned and $\intq=0$.
            \item \emph{positively variable-aligned} iff it is variable-aligned and $\intq\geq0$.
            \item \emph{negatively variable-aligned} iff it is variable-aligned and $\intq\leq0$.
            \item \emph{broadly variable-aligned} iff 
                  all iterations of $\sch$ have frames of the form $[\intk_1..\param-\intk_2]$, or $[\intk_1..\var-\intk_2]$, 
                  where $\intk_1,\intk_2\in\Z$.
          \end{itemize}
          An iteration of frame $[\expr_1..\var+\intq]$ is called an \emph{$\var$-iteration}.
          Let $\strat'$ be the strategy $\strat$ except that \emptiness is disallowed.
        \end{definition}
        \begin{theorem}
          \label{thm:terminate_variable_regular}
          $\strat'$ terminates on every schema which is monadic, of limited progression and 
          simply variable-aligned on $[\intk..\param-\intl]$ for some $\intk,\intl\in\Z$.
        \end{theorem}
        \begin{proof}(Sketch)
          It is easily seen that variable-alignment is preserved all along the procedure 
          (this fact plays the same role as Lemma \ref{cor:aligned}):
          indeed, the only way an $\var$-iteration $\schematicOp\varj\intk\var\pat$
          may be unfolded is by unfolding the iteration binding $\var$
          (which necessarily exists as $\var$ is not a parameter).
          Let us write it $\schematicOpbis\var\intk\expr{\pat'}$,
          $\expr$ is either a non-parameter variable or a \linear of the form $\param-\intl-\intj$.
          When this iteration is unfolded, it is turned into $\schematicOpbis\var\intk{\expr-1}{\pat'}\Opbis\pat\substitution{\expr/\var}$.
          We have now two copies of $\schematicOp\varj\intk\var\pat$:
          one inside $\schematicOpbis\var\intk{\expr-1}{\pat'}$,
          and one inside $\pat\substitution{\expr/i}$.
          The last one has actually been instantiated: $\schematicOp\varj\intk\expr\pat$.
          As this iteration has the same frame as $\schematicOpbis\var\intk\expr{\pat'}$ it also meets the requirements to be unfolded,
          which indeed happens, turning the iteration into $\schematicOp\varj\intk{\expr-1}\pat$.
          This new iteration is framed on $[\intk..\expr-1]$ like every other non $\var$-iteration in the node.
          As \emptiness is disallowed all non-instantiated $\var$-iterations are kept as is.
          Finally there is a finite number of such instantiations as each time the number of iterations below the observed iteration decreases.
          As a consequence all generated schemata are \arithmetic w.r.t. $\param$, and the proof is then very similar to the \elementary case.
          \qed
        \end{proof}
        \begin{theorem}
          \label{thm:terminate_positive_regular}
          $\strat'$ terminates on every schema which is monadic,
          of limited progression and positively variable-aligned on $[\intk..\param-\intl]$ for some $\intk,\intl\in\Z$.
        \end{theorem}
        \begin{proof}(Sketch)
          It is a combination of the previous proof and proof of Theorem \ref{thm:terminate_down_regular}.
          Except that now rectification not only occurs at the beginning of the procedure but each time an $\var$-iteration is unfolded.
          Indeed each time $\var$ is instantiated in an $\var$-iteration $\schematicOp\varj\intk{\var+\intq}\pat$, this iteration has to be rectified.
          There are still finitely many schemata that are generated as instantiating $\var$-iterations can only lead to finitely many different iterations
          up to a shift.
          \qed
        \end{proof}
        \begin{theorem}
          \label{thm:terminate_negative_regular}
          $\strat'$ terminates on every schema which is monadic,
          of limited progression and negatively variable-aligned on $[\intk..\param-\intl]$ for some $\intk,\intl\in\Z$.
        \end{theorem}
        \begin{proof}(Sketch)
          It is a combination of the proofs of Theorems \ref{thm:terminate_variable_regular} and \ref{thm:terminate_up_regular}.
          Except that now the maximum deviation used in Theorem \ref{thm:congruence} will not be the deviation of the original schema $\sch$,
          but rather the deviation of $\sch$ in which \emph{all iterations have been unfolded once}.
          Indeed schemata are not aligned anymore, even after rectification:
          when Step 2 terminates, the constraint $\intk<\param-\intl-\varj$
          where $\intl=\min\{\intl'\mid\param-\intl'\text{ is the upper bound of an iteration in }\sch\}$
          has been added.
          Hence if an $\var$-iteration $\schematicOp\varj\intk{\var-\intq}\pat$, $\intq>0$, is instantiated,
          we get $\schematicOp\varj\intk{\param-\intl-\intj-\intq}\pat$
          which cannot be unfolded as nothing ensures that $\intk\leq\param-\intl-\intj-\intq$.
          So we have to deal with mis-alignment.
          As in the proof of Theorem \ref{thm:terminate_up_regular},
          it is easily seen that this is not a problem as we have a maximum deviation as noted above.
          \qed
        \end{proof}
        Finally the following theorem is obtained by combining all previous proofs:
        \begin{theorem}
          \label{thm:terminate_broadly_variable_regular}
          $\strat'$ terminates on every schema which is monadic, of limited progression and broadly variable-aligned.
        \end{theorem}
      }%

  \section{Conclusion}
    \label{sec:conclusion}
    We have presented a proof procedure, called \SchDP, for reasoning with propositional formula schemata. 
    The main originality of our calculus is that the inference rules may apply at a deep position in a formula,
    a feature that is essential for handling nested iterations.
    A looping mechanism is introduced to improve the termination behavior. 
    We defined an abstract notion of looping which is very general,
    then instantiated this relation into a more concrete version that is decidable,
    but still powerful enough to ensure termination in many cases.

    We identified a class of schemata, called \elementary schemata, for which \SchDP always terminates.
    This class is much more expressive than the class of regular schemata handled in \cite{tab09}.
    The principle of the termination proof is (to the best of our knowledge) original:
    it goes by investigating how a given subformula is affected by the application of expansion rules on the ``global'' schema. 
    This is done by defining a ``traced'' version of the calculus in which 
    additional information is provided concerning the evolution of a specific subformula 
    (or set of subformulae, since a formula may be duplicated).
    This also required a thorough investigation of the properties of the looping relation. %-> VERSION COURTE: si toujours pertinent
    We believe that these ideas could be reused to prove termination of other calculi,
    sharing common features with \SchDP
    (namely calculi that operate at deep levels inside a formula and that allow cyclic proofs).

    We do not know of any similar work in automated deduction.
    Schemata have been studied in logic (see e.g. \cite{COR06,BZ94,ORE91})
    but our approach is different from these (essentially proof theoretical) works
    both in the particular kind of targeted schemata
    and in the emphasis on the automation of the proposed calculi.
    However one can find similarities with other works.

    Iterations can obviously recall of fixed-point constructions,
    in particular in the (modal) $\mu$-calculus%
    \footnote{In which many temporal logics e.g. CTL, LTL, and CTL* can be translated.}
    \cite{handbook_modal_mucalculus}
    (with $\stdAnd\aformula$ translated into something like $\mu X.\aformula\And X$).
    However the semantics are very different:
    that of iterated schemata is restricted to \emph{finite models}
    (since every parameter is mapped to an integer, the obtained interpretation is finite),
    whereas models of the $\mu$-calculus may be infinite.
    Hence the involved logic is very different from ours and actually simpler from a theoretical point of view:
    the $\mu$-calculus admits complete proof procedures and is decidable,
    whereas schemata enjoy none of those properties.
    The relation between schemata and the $\mu$-calculus \CourteLongue{is}{might actually be} analogous to the relation
    between finite model theory \cite{fmt} and classical first-order logic.
    The detailed comparison of all those formalisms is
    worth investigating but out of the scope of the present work.
    Other fixed-point logics exist that can embed schemata such as  
    least fixpoint logic \cite{immerman_lfp} or the first-order $\mu$-calculus \cite{firstorder_mucalculus}.
    However they are essentially studied for their theoretical properties
    i.e. complete or decidable classes are seldom investigated.
    Actually the only such study that we know of is in \cite{baelde}
    and iterated schemata definitely do not lie in the studied class nor can be reduced to it.

    One can also translate schemata into first-order logic by turning the iterations into (bounded) quantifications
    i.e. $\stdAnd\aformula$ (resp. $\stdOr\aformula$) becomes $\forall\var(1\leq\var\leq\param\Implies\aformula)$
    (resp. $\exists\var(1\leq\var\leq\param\And\aformula)$).
    This translation is completed by quantifying universally
    on the parameters and by axiomatizing first-order linear arithmetic.
    Then automated reasoning is achieved through a first-order theorem prover. 
    As arithmetic is involved, useful results would probably be
    obtained only with inductive theorem provers \cite{COM01,handbook_bundy}.
    However there are very few decidability results that can be used with such provers.
    Moreover most of those systems are designed to prove formulae
    of the form $\forall\vec x.\aformula$ where $\aformula$ is \emph{quantifier-free}.
    The translation sketched above clearly shows most translated schemata do not match this form.
    \Longue{Actually this is already the case of any schema involving only one iterated \emph{disjunction}.
    Indeed adding existential quantification in inductive theorem proving is known to be a difficult problem.
    Notice finally that this translation completely hides the \emph{structure} of the original problem.}

    \Longue{
      Finally, as we have seen in Section \ref{sec:termination},
      decidability of \elementary schemata lies in the detection of \emph{cycles} during the proof search.
      This idea is not new, it is used e.g.
      in tableaux methods dealing with modal logics in transitive frames
      \cite{gore_tableau_methods}, or $\mu$-calculi \cite{mucalc_tab_cleaveland}.
      However our cycle detection is quite different because we use it to actually prove by induction.
      Notice in particular that, contrarily to the mentioned tableaux methods, we cannot in general ensure termination.
      It is more relevant to consider our case as a particular instance of \emph{cyclic proofs},
      which are studied in proof theory precisely in the context of proofs by induction.
      Both \cite{Brotherston} and \cite{InductiveReasoningFossacs}
      show that cyclic proofs seem as powerful as systems dealing classically with induction.
      A particular advantage of cyclic proofs is that finding an invariant is not needed,
      making them particularly suited to automation.
      This is also extremely useful for the formalization of mathematical proofs, because it allows one to express
      a potentially infinite proof steps sequence, thus avoiding the explicit use of the induction principle.
      This last feature has been used to avoid working with more expressive logical formalisms \cite{HLWP08}.
      However once again studies on cyclic proofs are essentially theoretical and no complete class is identified at all.
    }

    Future work includes the implementation of the \SchDP calculus 
    and the investigation of its practical performances%
    \footnote{An implementation of the less powerful but simpler \SchCal procedure is available at 
    \texttt{http://regstab.forge.ocamlcore.org}.}.
    It would also be interesting to extend the termination result in Section \ref{sec:complete} to non monadic schemata%
    \CourteLongue.{ so as to be able to express e.g. the binary multiplier of the Introduction.}
    Extension of the previous results to more powerful logics (such as first-order logic or modal logic) naturally deserves to be considered.
    \Longue{Finally the proof of Theorem \ref{thm:terminate_regular} seems to be a powerful tool.
    We hope that the underlying ideas could be useful in other proof systems.
    In particular investigating more thoroughly the looping relation could give rise to interesting connections.}

  \bibliographystyle{plain}
  \bibliography{rapport}

  \Longue{\newpage}%
  \appendix
  \section{An Example of a \SchDP Proof}
    \label{sec:example}
    %For the referees convenience,
    %we provide in this section a detailed example of a \SchDP proof 
    %(also available in the extended version of the paper \cite{rapport09}).
    %This proof is constructed by applying the \SchDP calculus with the strategy $\strat$. 
    %Looping is applied using the pure and constraint-irreducible extensions of equality up to shift.

    We want to prove that $A+0=A$ where $+$ denotes the addition specified 
    by the schema $Adder$ described in the Introduction.
    A SAT-solver can easily prove this for a fixed $\param$ (say $\param=9$). 
    We show how to prove it for all $\param\in\N$ with \SchDP.
    This simple example has been chosen for the sake of conciseness,
    but commutativity or associativity of the adder could have been proven too.

    We express the fact that the second operand is null:
    \CourteLongue
    {$\schematicAnd\var1\param \neg B_\var$,}
    {\[\schematicAnd\var1\param \neg B_\var\]}
    and the conjecture i.e. the fact that the result equals the first operand:
    \CourteLongue
    {$\schematicAnd\var1\param A_\var\Leftrightarrow S_\var$.}
    {\[\schematicAnd\var1\param A_\var\Leftrightarrow S_\var\]}
    We negate the conjecture in order to prove it by refutation:
    \CourteLongue
    {$\schematicOr\var1\param A_\var\Xor S_\var$.}
    {\[\schematicOr\var1\param A_\var\Xor S_\var\]}
    Finally we want to refute%
    \CourteLongue
    { $Adder\And\schematicAnd\var1\param\neg B_\var\And\schematicOr\var1\param A_\var\Xor S_\var$.}
    {:\[Adder\And\schematicAnd\var1\param\neg B_\var\And\schematicOr\var1\param A_\var\Xor S_\var\]}

    The following figure is only a sketch of the real tableau:
    several rules are often applied at once,
    denoted by vertical dots labelled with the names of the used rules.
    \CourteLongue{C}{We use the conventions that c}losed leaves are marked by $\closed$,
    leaves looping on a node $\node$ 
    \Courte{(see Section \ref{sec:looping_detection})}
    by $\cycle\node$.
    Changed parts of a node are underlined,
    ``$\same$'' means ``same value as parent node's''.
    \Longue{We recall that }$\sch_1\Equiv \sch_2$ and $\sch_1\Xor \sch_2$ are shorthands
    for $(\sch_1\Implies \sch_2)\And(\sch_2\Implies \sch_1)$ and $\neg(\sch_1\Equiv \sch_2)$ respectively.
    All bound variables should be renamed so as to have different names,
    this is not done for the sake of readability.

    \everymath{\scriptsize}
    \Tree 
      [.{
        {$(1)$}\\
        $\begin{array}c
          (\schematicAnd\var1\param{Sum_\var}\And\schematicOr\var1\param{A_\var\Xor S_\var}\And\neg C_1\\
          \And\schematicAnd\var1\param{Carry_\var}\And\schematicAnd\var1\param{\neg B_\var}\And\param\geq1,\emptyset)\end{array}$\\
          \vdots\labeled\instantiaterule\\
        $\begin{array}c
          (\Mathem{\schematicAnd\var1{\param-1}{Sum_\var}\And
          Sum_\param\And\schematicOr\var1{\param-1}{A_\var\Xor S_\var}\Or(A_\param\Xor S_\param)\And\neg C_1}\\
          \Mathem{\And\schematicAnd\var1{\param-1}{Carry_\var}\And Carry_\param\And\schematicAnd\var1{\param-1}{\neg B_\var}\And\neg B_\param}
          \And\param\geq1,\same)\\
        \end{array}$}
        [.{$(\same,\{\Mathem{A_\param}\})$}
        [.{$(\same,\{A_\param,\Mathem{\neg S_\param}\})$\\
            \vdots\labeled\propsimpl\\
            $(2)$\gauche}
          ]
          !{\qsetw{-1cm}}
          [.{$(3)$\\$(\same,\{A_\param,\Mathem{S_\param}\})$\\
            \droite\vdots\labeled\propsimpl\\
                $\begin{array}c
                  (\schematicAnd\var1{\param-1}{Sum_\var}\And Sum_\param\And\schematicOr\var1{\param-1}{A_\var\Xor S_\var}\\
                  \Or(\Mathem{\schemaAnd\varj{\param\neq\param\And\varj=0}{A_\param}\Xor
                  \schemaAnd\varj{\param\neq\param\And\varj=0}{S_\param}})\And\neg C_1\\
                  \schematicAnd\var1{\param-1}{Carry_\var}\And Carry_\param\And\schematicAnd\var1{\param-1}{\neg B_\var}\And\neg B_\param
                  \And\param\geq1,\same)\\
                \end{array}$\\
                \droite\vdots\labeled\algebraic\\
                $(\schematicAnd\var1{\param-1}{Sum_\var}\And Sum_\param\And\schematicOr\var1{\param-1}{A_\var\Xor S_\var}\And\neg C_1$\\
                $\schematicAnd\var1{\param-1}{Carry_\var}\And Carry_\param\And\schematicAnd\var1{\param-1}{\neg B_\var}\And\neg B_\param
                \And\param\geq1,\same)$\\
                }
                {
                $(\dots\And\Mathem{\param-1\geq1},\same)$\\
                  \vdots\\
                  \labeled{\substack{\propsplit,\ \propsimpl\\
                  \text{and }\algebraic\text{; so that}\\
                  Sum_\param\text{ and }Carry_\param\text{ are removed;}\\
                  \text{branches trivially closed are omitted}
                  }}\\
                  $\cycle1$\\
                }
                { $(\schematicAnd\var1{\param-1}{Sum_\var}\And Sum_\param\And\Mathem\falseformula\And\neg C_1$\\
                  $\schematicAnd\var1{\param-1}{Carry_\var}\And Carry_\param\And\schematicAnd\var1{\param-1}{\neg B_\var}\And\neg B_\param
                  \And\Mathem{\param-1<1},\same)$\\
                  \ \droite\droite\vdots\labeled\algebraic\\
                 \!\!$\closed$
                }
          ]
        ]
        !{\qsetw{8.5cm}}
        [.{$(\same,\{\Mathem{\neg A_\param}\})$}
          [.{$(\same,\{\neg A_\param,\Mathem{S_\param}\})$\\
            \vdots\labeled\propsimpl\\
            $(2')$\gauche\\}
          ]
          [.{
            $(\same,\{\neg A_\param,\Mathem{\neg S_\param}\})$\\
            $(3')$\\
            }
          ]
        ]
      ]

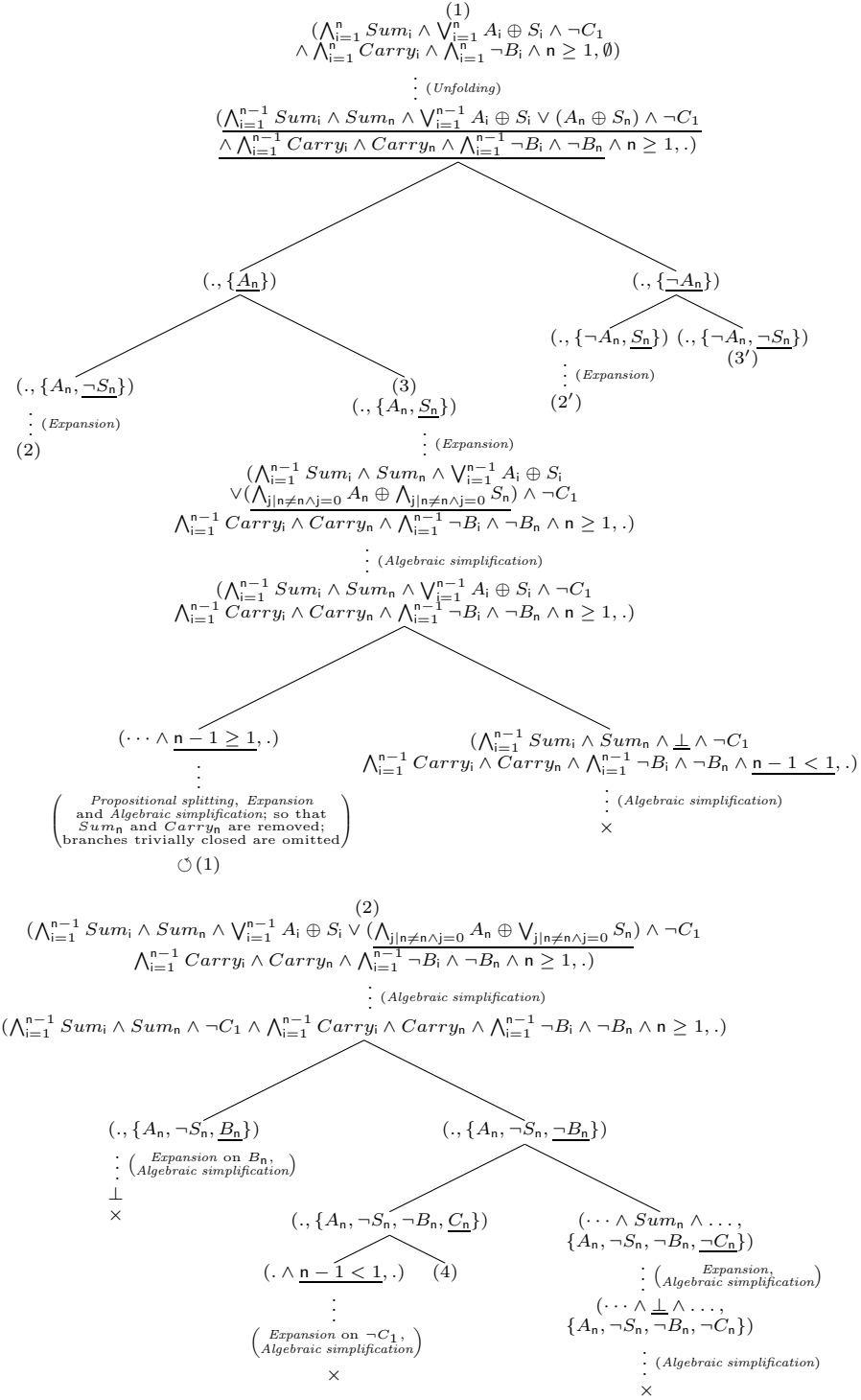
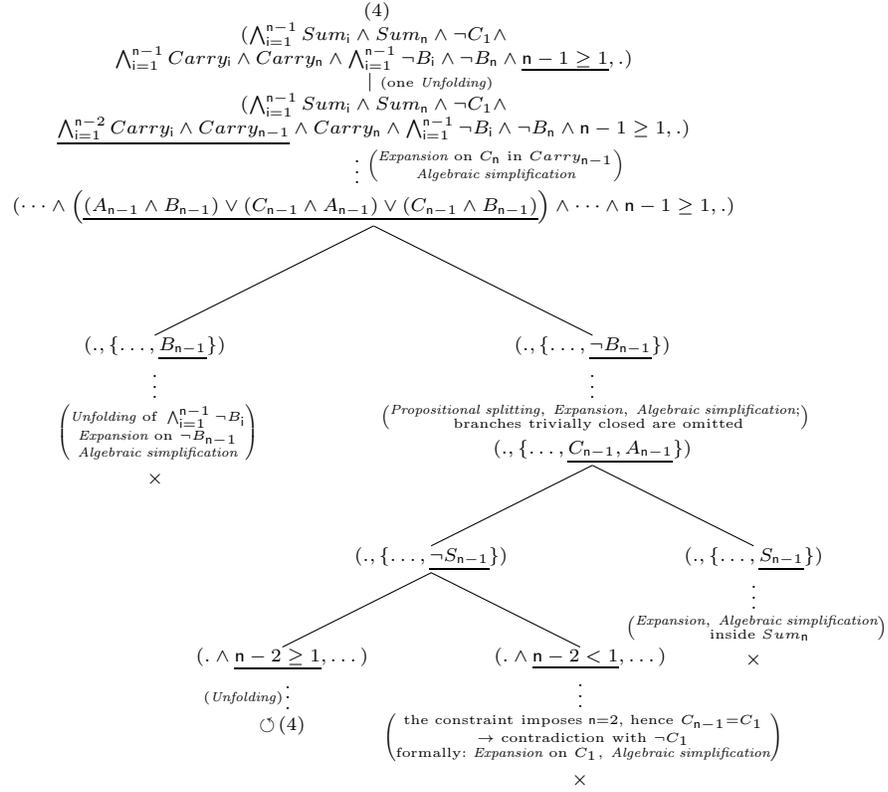
    \Tree 
      [.{
        { $(2)$}\\
        $\begin{array}c
          (\schematicAnd\var1{\param-1}{Sum_\var}\And Sum_\param\And\schematicOr\var1{\param-1}{A_\var\Xor S_\var}\Or
          (\Mathem{\schemaAnd\varj{\param\neq\param\And\varj=0}{A_\param}\Xor
          \schemaOr\varj{\param\neq\param\And\varj=0}{S_\param}})\And\neg C_1\\
          \schematicAnd\var1{\param-1}{Carry_\var}\And Carry_\param\And\schematicAnd\var1{\param-1}{\neg B_\var}\And\neg B_\param
          \And\param\geq1,\same)\\
        \end{array}$\\
        \droite\vdots\labeled\algebraic\\
        $(\schematicAnd\var1{\param-1}{Sum_\var}\And Sum_\param\And\neg C_1\And
          \schematicAnd\var1{\param-1}{Carry_\var}\And Carry_\param\And\schematicAnd\var1{\param-1}{\neg B_\var}\And\neg B_\param
          \And\param\geq1,\same)
        $}
        {$(\same,\{A_\param,\neg S_\param,\Mathem{B_\param}\})$\gauche\\
        \droite\vdots\labeled{\substack{\propsimpl\text{ on }B_\param,\\\algebraic}}\\$\falseformula$\gauche\\$\closed$\gauche}
        !{\qsetw{1cm}}
        [.{$(\same,\{A_\param,\neg S_\param,\Mathem{\neg B_\param}\})$}
        [.{$(\same,\{A_\param,\neg S_\param,\neg B_\param,\Mathem{C_\param}\})$}
            !{\qsetw{1cm}}
            {$(\same\And\Mathem{\param-1<1},\same)$\\
            \vdots\\
            \labeled{\substack{\propsimpl\text{ on }\neg C_1,\\\algebraic}}
            \\$\closed$}
            {$(4)$}
          ]
          {$(\dots\And Sum_\param\And\dots,$\\$\{A_\param,\neg S_\param,\neg B_\param,\Mathem{\neg C_\param}\})$\\
          \hspace*{2cm}\vdots\labeled{\substack{\propsimpl,\\ \algebraic}}\\
           $(\dots\And\Mathem\falseformula\And\dots,$\\$\{A_\param,\neg S_\param,\neg B_\param,\neg C_\param\})$\\
           \hspace*{2cm}\vdots\labeled\algebraic\\
           $\closed$\quad\ \ 
          }
        ]
      ]
    
    \bigskip
    \Tree 
      [.{
          { $(4)$}\\
          $(\schematicAnd\var1{\param-1}{Sum_\var}\And Sum_\param\And\neg C_1\And$\\
          $\schematicAnd\var1{\param-1}{Carry_\var}\And Carry_\param\And\schematicAnd\var1{\param-1}{\neg B_\var}\And\neg B_\param
          \And\Mathem{\param-1\geq1},\same)$\\
          \hspace*{1.5cm}$|\ \scriptscriptstyle(\text{one }\instantiaterule)$\\
          $(\schematicAnd\var1{\param-1}{Sum_\var}\And Sum_\param\And\neg C_1\And$\\
          $\Mathem{\schematicAnd\var1{\param-2}{Carry_\var}\And Carry_{\param-1}}
          \And Carry_\param\And\schematicAnd\var1{\param-1}{\neg B_\var}\And\neg
          B_\param\And\param-1\geq1,\same)$\\
          \droite\droite\vdots\labeled{\substack{\propsimpl\text{ on }C_\param\text{ in }Carry_{\param-1}\\\algebraic}}\\
          $(\dots\And\left(\Mathem{(A_{\param-1}\And B_{\param-1})\Or(C_{\param-1}\And A_{\param-1})\Or(C_{\param-1}\And B_{\param-1})}\right)\And\dots
          \And\param-1\geq1,\same)$\\
        }
        {$(\same,\{\dotsc,\Mathem{B_{\param-1}}\})$\\
        \vdots\\
        \labeled{\substack{\instantiaterule\text{ of }\schematicAnd\var1{\param-1}{\neg B_\var}\\
        \propsimpl\text{ on }\neg B_{\param-1}\\
        \algebraic}}\\
        $\closed$}
        !{\qsetw{-1cm}}
        [.{$(\same,\{\dotsc,\Mathem{\neg B_{\param-1}}\})$\\
        \vdots\\
        \labeled{\substack{\propsplit,\ \propsimpl,\ \algebraic;\\\text{ branches trivially closed are omitted}}}\\
          $(\same,\{\dotsc,\Mathem{C_{\param-1},A_{\param-1}}\})$}
          !{\qsetw{1cm}}
          [.{$(\same,\{\dotsc,\Mathem{\neg S_{\param-1}}\})$}
            {$(\same\And\Mathem{\param-2\geq1},\dotsc)$\\
            \labeled\instantiaterule\vdots\gauche\\
            $\cycle4$}
            {$(\same\And\Mathem{\param-2<1},\dotsc)$\\
            \vdots\\
            \labeled{\substack{\text{the constraint imposes }\param=2\text{, hence }C_{\param-1}=C_1\\
            \to\text{ contradiction with }\neg C_1\\
            \text{formally: }\propsimpl\text{ on }C_1,\ \algebraic}}\\
            $\closed$}
          ]
        !{\qsetw{4.9cm}}
          {$(\same,\{\dotsc,\Mathem{S_{\param-1}}\})$\\
          \vdots\\
          \labeled{\substack{\propsimpl,\ \algebraic\\\text{ inside }Sum_\param}}\\
          $\closed$}
        ]
      ]

    \bigskip
    $(2')$ and $(4')$ are very similar to $(2)$ and $(4)$.

\end{document}